\documentclass[11pt]{amsart}

\usepackage{a4wide}
\usepackage{color}
\usepackage{esint,amssymb}
\usepackage{graphicx}
\usepackage{mathtools}
\usepackage[colorlinks=true, pdfstartview=FitV, linkcolor=blue, citecolor=blue, urlcolor=blue,pagebackref=false]{hyperref}
\usepackage{microtype}

\usepackage{bm}
\usepackage{scalerel} 
\usepackage{dsfont}
\usepackage[font={footnotesize}]{caption}

\usepackage{accents}
\newcommand{\dbtilde}[1]{\accentset{\approx}{#1}}
\newtheorem{theo}{Theorem}
\newtheorem{prop}{Proposition}[section]
\newtheorem{lem}[prop]{Lemma}
\newtheorem{coro}[prop]{Corollary}
\newtheorem{remark}[prop]{Remark}

\theoremstyle{plain}

\theoremstyle{definition}
\newtheorem{defi}[prop]{Definition}

\numberwithin{equation}{section}

\newcommand{\R}{\mathbb{R}}

\renewcommand{\P}{\mathbb{P}}

\newcommand{\ep}{\varepsilon}

\newcommand{\g}{\mathbf{g}}
\newcommand{\f}{\mathbf{f}}

\def \be{\begin{equation}}
\def \ee{\end{equation}}

\def \t0{\rightarrow 0} 

\def \hal{\frac{1}{2}}

\def \supp{\mathrm{supp }} 
\def \div{\mathrm{div} \,} 
\def \1{\mathbf{1}} 

\def \p{\partial}
\def \ep{\varepsilon}
\def \dist{\mathrm{dist}}

\def\cd{\mathsf{c_{\d}}}

\def\({\left(}
\def\){\right)}
\def\nab{\nabla}

\def \PNbeta{\mathbb{P}_{N, \beta}} 

\renewcommand{\subset}{\subseteq}

\renewcommand{\subset}{\subseteq}

\renewcommand{\bar}{\overline}

\renewcommand{\tilde}{\widetilde}

\renewcommand{\div}{\divg}
\newcommand{\indic}{\mathds{1}}

\renewcommand{\hat}{\widehat}

\newcommand{\ed}[1]{{\color{magenta}{#1}}}

\def\Xint#1{\mathchoice
   {\XXint\displaystyle\textstyle{#1}}%
   {\XXint\textstyle\scriptstyle{#1}}%
   {\XXint\scriptstyle\scriptscriptstyle{#1}}%
   {\XXint\scriptscriptstyle\scriptscriptstyle{#1}}%
   \!\int}
\def\XXint#1#2#3{{\setbox0=\hbox{$#1{#2#3}{\int}$}
     \vcenter{\hbox{$#2#3$}}\kern-.5\wd0}}
\def\dashint{\Xint-}

\def \XN{{X}_N}

\def\N{{n_\mathcal{O}}}

\def\n{n}
\def\m{\mu}
\def\l{\ell}
\def\L{\mathsf{L}}
\def\K{\mathsf{K}}
\def\F{\mathsf{G}}
\def\G{\mathsf{F}}
\def\H{\mathsf{H}}
\def\P{\mathsf{P}}
\def\Q{\mathsf{Q}}
\def\car{\Box}

\def\C{\mathcal{C}}
\def\HN{\mathcal{H}_N}

\def\Esp{\mathbb{E}} 

\def \ZNbeta{Z_{N,\beta}}

\def\pa{\partial}

\def \mut{\overline{\mu}_{t}}

\def\g{\mathsf{g}}
\def\d{\mathsf{d}}
\def\mn{\mathrm{n}}

\def\muv{\mu_V}

\def\indic{\mathbf{1}}

\def\rr{\mathsf{r}}
\def\rrc{\tilde{\mathsf{r}}}
\def\rrh{\hat{\mathsf{r}}}
\def\div{\mathrm{div}\,}

\def \Escr{E^0}

\def \C{X}
\def\Old{\mathcal{O}}
\def\New{\mathcal{N}}

\def\mub{\mu_{\theta}}
\def\mut{\mu_{\theta}}
\def\zb{\zeta_\theta}

\def\fae{\f_{\alpha_i,\eta_i}}
\def\faej{\f_{\alpha_j, \eta_j}}

\def\rb{\rho_\beta}
\def\mf{f_\d}
\def\dm{1}
\def\GG{\mathsf{h}}
\def\xib{\min(\beta^{\frac{1}{\d-2}},1)}
\def\omc{\overset{\circ}{\Omega}}

\def\bN{{\bar N}}
\def\rrtt{\dbtilde{\mathsf{r}}}

\begin{document}

\title[Local Laws for Coulomb Gases]
{Local laws and rigidity for Coulomb gases at any temperature}

\begin{abstract} 
We study Coulomb gases in any dimension $\d \geq 2$ and in a broad temperature regime. We prove local laws on the energy, separation and  number of points down to the microscopic scale. These yield the existence of limiting point processes after extraction generalizing the Ginibre point process for arbitrary temperature and dimension. The local laws come together with a quantitative expansion of the free energy with a new explicit error rate in the case of a uniform background density. These estimates have explicit temperature dependence, allowing to treat regimes of very large or very small temperature, and exhibit a new minimal lengthscale for rigidity and screening depending on the temperature. They apply as well to energy minimizers (formally zero temperature). The method is based on a bootstrap on scales and reveals the additivity of the energy modulo surface terms, via the introduction of subadditive and superadditive approximate energies.


\end{abstract}

\author[S. Armstrong]{Scott Armstrong}
\address[S. Armstrong]{Courant Institute of Mathematical Sciences, New York University, 251 Mercer St., New York, NY 10012}
\email{scotta@cims.nyu.edu}

\author[S. Serfaty]{Sylvia Serfaty}
\address[S. Serfaty]{Courant Institute of Mathematical Sciences, New York University, 251 Mercer St., New York, NY 10012}
\email{serfaty@cims.nyu.edu}

\date{\today}

\maketitle
\setcounter{tocdepth}{1}
\tableofcontents

\section{Introduction} 
\subsection{Statement of the problems}
 We are interested in the $N$-point canonical Gibbs measure
for a $\d$-dimensional Coulomb gas ($\d \ge 2$) at inverse temperature $\beta$, in a confining potential~$V$, defined by
\begin{equation}\label{def:PNbeta}
d\PNbeta(\XN) = \frac{1}{\ZNbeta} \exp \left( - \beta N^{\frac{2}{\d}-1} \HN(\XN)\right) d\XN,
\end{equation}
where $\XN = (x_1, \dots, x_N)\in (\R^{d})^N$ and the Hamiltonian~$\HN(\XN)$, which represents the energy of the system in state~$\XN$, is defined by
\begin{equation} \label{def:HN}
\HN(\XN) := \hal \sum_{1 \leq i \neq  j \leq N} \g(x_i-x_j) +N \sum_{i=1}^N  V(x_i),
\end{equation}
where
\begin{align} 
\label{wlog2d} 
\g(x) : = 
\left\{ 
\begin{aligned}
& - \log |x| & \text{if} & \ \d=2, \\
& |x|^{2-\d} & \text{if} & \ \d \geq 3.
\end{aligned} 
\right.
\end{align}
Thus~$\HN(\XN)$ is the sum of the pairwise repulsive Coulomb interaction between the particles and the total potential energy of the particles due to the confining potential~$V$, the intensity of which is of order~$N$. The normalizing constant $\ZNbeta$ in the definition~\eqref{def:PNbeta}, called the \textit{partition function}, is given by 
\begin{equation}\label{defZ}
\ZNbeta := \int_{(\R^\d)^N} \exp \left( - \beta N^{\frac{2}{\d}-1}  \HN(\XN)\right) d\XN.
\end{equation}
We have chosen particular units of measuring the inverse temperature by writing~$\beta N^{\frac{2}{\d}-1}$ instead of~$\beta$. As seen in~\cite{lebles}, this turns out to be a natural choice due to scaling considerations as our~$\beta$ corresponds to the effective inverse temperature governing the microscopic scale behavior. This choice does not reduce the generality of our results since, as we will see, our estimates are explicit in their dependence on~$\beta$ and~$N$, which allows to let $\beta $  depend on $N$. 

\smallskip

This Coulomb gas model, also called a ``one-component plasma,'' is a standard ensemble of statistical mechanics which has attracted much attention in the physics literature: see for instance  \cite{martinreview,aj,CDR,sm,kiessling,messerspohn} and references therein. Its study in the two-dimensional case is more developed, thanks in particular to its connection with Random Matrix Theory (see~\cite{dyson,mehta,forrester}): when
$\beta=2$ and $V(x)=|x|^2$, the~$\PNbeta$ in~\eqref{def:PNbeta} is the law of the (complex)  eigenvalues of the Ginibre ensemble of~$N\times N$ matrices with normal Gaussian i.i.d.~entries~\cite{ginibre}. Several additional motivations come from quantum mechanics, in particular via the plasma analogy for the fractional quantum Hall effect~\cite{Gir,stormer,laughlin2}. For all of these aspects one may consult to~\cite{forrester}. The Coulomb case with~$\d = 3$, which can be seen as a toy model for matter, has been studied in~\cite{jlm,LiLe1,LN}. The theory of higher-dimensional Coulomb systems is much less well-developed.   
 
\smallskip

In such Coulomb systems, if $\beta $ is not too small  and if $V$ grows fast enough at infinity, then the empirical measure
$$\hat\mu_N := \frac{1}{N} \sum_{i=1}^N \delta_{x_i}$$ converges, as $N \to \infty$, to a deterministic equilibrium measure $\mu_V$ with compact support, which can be identified as the unique minimizer among probability measures of the quantity
 \begin{equation}
 \label{defE}
 \mathcal E(\mu)= \hal \int_{\R^\d\times \R^\d} \g(x-y) d\mu(x) d\mu(y)+ \int_{\R^\d} V(x) d\mu(x).\end{equation}
See for instance~\cite[Chap. 2]{ln} for the statement of such a result. As the length scale of $\supp\,\mu_V$ is of order~$1$ (it is independent of $N$), we will call this the {\it macroscopic scale}, while the typical interparticle distance is of order $N^{-\frac1\d}$, we will call this the {\it microscopic scale}, or \emph{microscale}. Intermediate length scales will be called \emph{mesoscales}.

\smallskip

In this paper, we work with a deterministic correction to the equilibrium measure which we call the  {\it thermal equilibrium measure}, which is appropriate for all temperatures and defined as the probability density $\mub$ minimizing 
\be \label{1.9}
\mathcal{E}_{\theta} (\mu):= \mathcal{E} (\mu) + \frac{1}{\theta} \int_{\R^\d} \mu\log \mu\ee
where we set 
\be\label{relbt}
\theta:= \beta N^{\frac{2}{\d}}.\ee
Let us point out that here and in all the paper we use alternatively the notation $\mu$  both for the measure  and for its density. 
   By contrast with $\muv$, $\mub$ is positive and regular in the whole of $\R^\d$ with exponentially decaying tails.
 In fact the quantity $\theta= \beta N^{\frac{2}{\d}}$ corresponds to the inverse temperature that governs the {\it macroscopic distribution of the particles}. The precise dependence of $\mub$ on $\theta$ is studied in \cite{ascomp} where
it is shown that when $\theta \to \infty$, then $\mub$ converges to $\muv$, with quantitative estimates (see below). 
 
\smallskip

The measure $\mub$ is well-known to be the limiting density of the point distribution in the regime in which $\theta$ is fixed independently of $N$ and we send $N\to \infty$, that is, for $\beta \simeq N^{-\frac 2d}$; see for instance \cite{kiessling,messerspohn,CLMP,bodgui}. In this paper we show that~$\mub$ is also a more precise description of the distribution of points, compared to the standard equilibrium measure, even in the case~$\theta \gg1$. This allows us to obtain more precise quantitative results valid for the full range of~$\beta$ and~$N$ and, in particular, in the regime of very small~$\beta$.

\smallskip
 
One of the important goals in the study of Coulomb systems is to show  concentration around the (thermal) equilibrium measure   and  
estimates on the so-called linear statistics
\begin{equation}
\label{pll}
\int_{\R^\d} \varphi( \sum_{i=1}^N \delta_{x_i} - N \mub)\end{equation}
for (not necessarily smooth) test functions $\varphi$ which may be supported in microscopic sized balls. The study of random variables such as~\eqref{pll} allows us to quantify the weak convergence of the empirical measure~$\hat \mu_N$ to the deterministic thermal equilibrium measure~$\mub$. In particular, we can obtain estimates on the number of points in microscopic balls ({\it local laws}). If the fluctuations of~\eqref{pll} are much smaller than for a Poisson point cloud, one speaks of {\it rigidity} or {\it hyperuniformity} (see~\cite{To}).

\smallskip

In this paper, we prove explicit controls  on these quantities, which then yield the existence of limiting point processes along subsequences of properly rescaled configurations. While we cannot rule out the possibility of several point processes arising as limits of different subsequences, we are able for the first time to show their existence by controlling the number of points in microscopic boxes. This also provides solutions  to a number of widely used hierarchies and sum rules  on correlation functions, in this important case of Coulomb interactions (see discussion below the statement of Corollary~\ref{coro1}). 

\smallskip

A second goal of this paper is 
 to give an expansion in $N$ for $N\gg 1$ of 
the \emph{free energy}~$-\frac{1}{\beta} \log \ZNbeta,$ which we will complete in the Neumann jellium case here (note that the mere existence of an order $N$ term, in other words, a thermodynamic limit, has been known since~\cite{LN}). This  opens the way to obtaining in the companion paper  \cite{s2}  an explicit  error rate for the  free energy expansion in the general case (in which $\muv$  or $\mut$ are not necessarily constant). This result is crucial to obtain, for the first time in~\cite{s2}, a Central Limit Theorem for the fluctuations of the type~\eqref{pll} in dimensions $\d\ge 3$ (such a result was obtained in dimension 2 in \cite{ls2,bbny2}, but the method  
requires a more precise rate to be applicable in higher dimension). 
The third motivation is to formulate a local Large Deviations Principle (LDP) with microscopic averages for the limiting point processes, analogous to results of~\cite{lebles,loiloc}.

\smallskip

Such questions have recently attracted attention in two dimensions \cite{bbny1,bbny2,loiloc,ridervirag,ahm,ls2,chm}, and to a much lesser extent  in higher dimension: concentration bounds were  given in \cite{rs,chm,gz1}, free energy expansions in \cite{lebles}, and rigidity was described in~\cite{chatterjee} (in dimension $2$) and \cite{ganguly} (in general dimension) for a ``hierarchical'' Coulomb gas model (that is, a version of the model with a simplified interaction which essentially makes renormalization arguments  easier), with estimates for the number variance in a set and for smooth linear statistics.
Of course, much more is known in the well studied related problem of the one-dimensional log gas or $\beta$-ensemble, see \cite{joha,bey1,bey2,shch,BorGui1,BorGui2,ss2,bls,bl,llw}. However, as far as we know none of  these works consider the regime of large temperature.

\smallskip

The program we carry out in this paper was already partly accomplished in dimension 2 in \cite{loiloc,bbny1}, with local free energy expansions and local laws valid down to {\it mesoscales} $\l \ge N^{-\alpha}$ with $\alpha<\hal$, via a bootsrap on the scales.
The high-level approach of the proof is the same in particular as the one of~\cite{loiloc}, however by revisiting it thoroughly we bring in the following novelties:

\begin{itemize}
\item We treat arbitrary dimension $\d\geq 2$. 
\item We unveil the importance of the  thermal equilibrium measure, even for large $\theta$ and notice the existence of two effective temperatures, one  that governs the macroscopic distribution of the points ($\theta$) and one that governs their microscopic behavior ($\beta$).  
\item The local laws are for the first time valid  {\it down to the microscale}, giving for the first time access to  the proof of existence of limiting point processes. 
\item The local laws are obtained with quantitative bounds in probability (exponential moments), and not just  with high probability as in previous works.
\item We obtain estimates with an explicit dependence in~$\beta$, as well as~$N$, allowing to consider very small or very large temperature regimes. These estimates reveal a new $\beta$-dependent minimal length scale~$\rb $ 
down to which the local laws hold. We prove that for $\d=2,3,4$ this lengthscale is $\sim N^{-\frac1\d} \max (1, \beta^{-\frac12})$ which we believe to be optimal. 
\item We give an explicit rate of convergence for the free energy expansion in the constant background case.
\item We introduce new sub- and superadditive energy quantities. It is by using estimates on their additivity defect, which are obtained by a bootstrap or renormalization-type argument, that we are able to quantify the convergence rate of the free energy and prove our main results.
\item We revisit the ``screening procedure" used in previous papers, turning it into a truly probabilistic procedure and tuning it in order to get explicit and optimal quantitative estimates. We optimize the screening lengthscale during the bootstrap procedure, showing it can be made as small as the minimal lengthscale $\rb$.
\end{itemize}

\subsection{Statement of main results}In all the paper we assume that 
\begin{equation} 
\label{integrabilite}
\int_{\R^\d} \exp\left( -\min(1,\theta) V \right) <\infty ,
\end{equation} and that 
\be \label{Vpg} V+\g \to + \infty \ \quad \text{as} \ |x|\to \infty,\ee
which ensures the existence of $\mu_V$ and $\mub$ (see \cite{ascomp}).

The local laws are more easily stated at the level of the ``blown-up configurations": for any $(x_1, \dots, x_N)$ we let $x_i'= N^{\frac1\d} x_i$ and 
 we also let $\mu=\mub'$ be the push-forward of $\mub$ under this rescaling, i.e. the measure with density $\mub'(x)= \mub(N^{-\frac1\d}x)$.   
The local laws are proven in the ``bulk" of $\mub$. 
After a suitable ``splitting" that removes the constant leading order term (see Section \ref{secsplit}) we are led to computing local laws with respect to a generic background $\mu$, hence our choice of notation here.

In dimension $\d \ge 3$, we will not use {\it any property of $\mu$ besides the fact that it is bounded above and below in a set $\Sigma$}. 
In dimension $\d=2$, we will use the same fact and only  three additional ones:
\begin{itemize}
\item $\mu$ has sufficiently small tails, in the form of the assumption 
\be \label{asstail}
\mu(\Sigma^c) \le \frac{C N}{\log N} \ \text{for some constant $C>0$.} 
\ee
We comment after Theorem \ref{th1} on what is known in that respect, in particular the assumption is true if $\beta$ is not too small; 

\item $\mu$ satisfies 
\be\label{assgmm} \iint_{\R^2\times \R^2} \g(x-y) d\mu(x) d\mu(y) \ge - C N^2 \log N,\ee
which holds with $C=\hal$ as an immediate consequence of the fact that $\mathcal E(\mut)$ is finite and the rescaling;

\item
$\mu$ satisfies
\be \label{asslog}
\int_U \log z \, d\mu(z)<\infty,\ee  
which is also true here since  since $\mathcal{E}(\mut) <\infty$ implies  $\int_{\R^\d} V d\mut<\infty$ which  in view of ~\eqref{Vpg} implies it.

\end{itemize}


Throughout the paper,~$C$ denote a constant which only depend on $\d$, upper and lower bounds on $\mu$ and the constants in~\eqref{asstail}--\eqref{assgmm} and may vary in each occurrence. 

\smallskip

As we will see, the dependence of our estimates in $\beta$ for $\beta$ small is a bit different in dimension $2$ than in higher dimensions. This is a manifestation of the fact that the Poisson point process has (or at least is expected to have) infinite Coulomb energy in dimension 2 (see \cite{leble} for a discussion). 
Reflecting this, throughout the paper, we will use the notation  
  \be\label{defchib}
\chi(\beta)= \begin{cases}   1  & \text{if} \ \d\ge 3\\
  1 +\max (-\log \beta, 0)& \text{if} \ \d=2,\end{cases}  \end{equation}
  and emphasize that $\chi(\beta)=1$ unless $\d=2$ and $\beta $ is small.
  
In all our formulas, we will have terms which appear only in dimension $\d$, we denote them with a $\indic_{\d}$.
The precise meaning of the next-order energy $\G^{\car_R(x)}$ localized in a cube $\car_R(x)$ of center $x$ and radius $R$  is alluded to below and defined precisely in Section \ref{sec2}.

\begin{theo}
[Local laws]
\label{th3} 
Assume  $\mu$ defined above satisfies $0<m \le \mu\le \Lambda $ in a set $\Sigma$, and in dimension $\d=2$ assume also~\eqref{asstail}, \eqref{assgmm}, and \eqref{asslog}. There exists a constant $C>0$ depending only on $\d, m$, $\Lambda $  and in dimension $2$ the constants of~\eqref{asstail} and~\eqref{assgmm},  such that the following holds.
 There exists $\rb$ of the form
  \be \label{minoR0}   \rb  = C \max\(1, \beta^{-\hal} \chi(\beta)^{\frac12} ,\beta^{\frac{1}{\d-2}-1}\indic_{\d\ge 5}\) 
\ee
such that 
if   $\car_R(x)$ is a cube of size $R\ge \rb$ centered at $x$,  with 
\be
\label{conddist0}
\dist(\car_R(x) , \p \Sigma) \ge \ed{
C \max\( \(\frac{N^{\frac1\d} }{ \max(1, \beta^{-\hal} \chi(\beta)^{\frac12})}\)^{-\frac23}   N^{\frac1\d} ,   N^{\frac{1}{ \d+2}}\) }, \ee 
 we have
\begin{enumerate}
\item (Control of energy) 
  \begin{equation}\label{locallawint0}
  \left|\log \Esp_{\PNbeta} \( \exp\( \hal \beta \G^{\car_{R}(x)}\)\)
  \right| \le 
C   \beta\chi(\beta ) R^\d 
  \end{equation}
  
  \item (Control of fluctuations) Letting $D$ denote $\int_{\car_R(x)} \( \sum_{i=1}^N \delta_{x_i'} - d\mu\)$ we have 
  \begin{equation}\label{loclawpoints0}
  \left|\log \Esp_{\PNbeta}\( \exp\( \frac{ \beta}{C }   R^{2(1-\d)} \rb^{\d-1} D^2   \) \)\right| \le
  C  \beta\chi(\beta)\rb^\d \end{equation}
  and
  \begin{equation}\label{loclawpoints00}
   \left|\log \Esp_{\PNbeta}\( \exp\( \frac{\beta}{C}    \frac{D^2}{R^{\d-2} } \min (1, \frac{ |D|}{R^\d})  \) \)\right|\le
  C  \beta\chi(\beta) R^{\d}.\end{equation}
\item (Control of linear statistics) If $\varphi$ is a  $1$-Lipschitz function  supported in $\car_R(x)$, then 
\begin{equation}\label{loclawphi}
\left|\log \Esp_{\PNbeta}\( \exp \frac{\beta }{C R^\d}\(  \int_{\R^\d} \varphi( \sum_{i=1}^N \delta_{x_i'}-\mu) \)^2      \)\right|\le
  C \beta \chi(\beta) R^{\d}\|\nab \varphi\|_{L^\infty}^2.
\ee  \item (Minimal distance control)
For any point $x_i'$ of the blown-up configuration satisfying the relation~\eqref{conddist0}, denoting
$$\rr_i=\min\( \min_{j\neq i} |x_i'-x_j'|, \frac14\) $$ we have 
\be
\label{loclawdistmin0}
\left| \log \Esp_{\PNbeta} \( \exp \left( \frac{\beta }{2} \g(\rr_i  )  \right) \) \right|
\leq   C \beta\chi(\beta)\rb^\d.
\ee
\end{enumerate}\end{theo}

\subsubsection*{Comments on the assumptions}
The equilibrium measure $\mu_V$ is caracterized by the fact that there exists a constant $c$ such that $\g* \mu_V + V-c$ is zero in the support of $\mu_V$ and  nonnegative outside. 
 In \cite{ascomp} it is proven that if \eqref{integrabilite} and \eqref{Vpg} hold, and if in addition 
 \be \label{assLapV} \Delta V \ge \alpha>0 \quad \text{in a neighborhood of } \ \supp \, \muv\ee
 and the potential $\g* \mu_V+ V-c $ is bounded below by a positive constant 
uniformly away from the support of $\mu_V$,   then 
 for $x\in \supp \, \muv$ we have   $\muv(x) \ge m>0$. 
 In particular, we can take $\Sigma$ to be the blown-up of $\supp\, \mu_V$ and  the assumption $\mub' \ge m >0$  holds in $\Sigma$.
 We note that  if $V$ is more regular \cite{ascomp} also provides an explicit expansion  of $\mub - \mu_V$ of  the form
\be \label{125}\mub\simeq \mu_V+ \frac{1}{\cd \theta} \Delta \log \Delta V + \frac{1}{\cd \theta^2}\Delta \( \frac{\Delta \log \Delta V}{\Delta V}\)+... \quad \text{in } \   \supp\, \muv,\ee
see \cite{ascomp} for precise results. 
It is also proven in \cite{ascomp} that  under the previous stated assumptions, we will have 
\be \label{muscom}\mub'(\Sigma^c) \le  \frac{C N }{\sqrt{\theta}}\ee
hence in dimension 2 the extra assumption~\eqref{asstail} is verified as soon as 
$$\beta \ge \frac{\log^2 N}{N} .$$
In view of~\eqref{125} one may also substitute $\mu$ by $\muv'= \muv(N^{-\frac1\d} x)$ in the local laws above while making only a small error.

\smallskip

 If $\theta $ is fixed then the lower bound $\mub'\ge m>0$ 
  is true on any compact subset of $\R^\d$. If $\theta \ll 1$ then $\mub\to 0$ pointwise as the measure $\mub$ spreads to infinity,
and one needs to give a stronger weight to the confining potential to confine the system, effectively making the interaction weaker and the points more independent, thus  this is a situation that needs to be studied separately (see for instance \cite{RSY2} for a discussion of such a ``thermal regime" in a radial situation).

\subsubsection*{Comments on the results}
We note that we can always reduce to $m=1$ by scaling in space and then obtain the explicit dependence on $m$ of all the constants by a rescaling of the quantities.

 An application of Markov's inequality  easily allows to estimate the probability of deviations from these laws. For instance,  the probability that the number of points in a cube deviates by more than $o(R^\d)$ from $N\int_{\car_R} \mub$ is very small, and 
 ~\eqref{loclawpoints0} provides a bound on  the variance of the number of points in $\car_R$  by 
 $C \rb R^{2(\d-1)}$.
We note that~\eqref{loclawpoints0} is stronger than even  the results of \cite{bbny1,loiloc} in dimension $2$. The relation~\eqref{loclawphi} can  be improved using more involved techniques if $\varphi$ is assumed to be more regular, this was shown in dimension 2 in \cite{bbny1,ls2,bbny2} and this is the object of \cite{s2} in higher dimension.

\smallskip

\smallskip

A closely related setup to our Coulomb gas is that of the {\it jellium} model (see for instance   \cite{lls,lewinlieb} and references therein) which is defined as follows. We are given $N=R^\d$ points constrained to be in a cube of size $R$ denoted by $\car_R:= \left( -\frac12 R,\frac12R \right)^\d$, neutralized by a uniform background of unit density, which has a law given by the Gibbs measure
\begin{equation}
\label{jell}
d\mathbb{Q}_{N, \beta}(\XN)=\frac{1}{Z_{R,\beta}^{\mathrm{jel}}} \exp\left( -\beta H_N^{\mathrm{jel}}(\XN) \right) d\XN,
\end{equation}
where
\begin{equation*}
H^{\mathrm{jel}}(\XN)
=
 \iint_{\R^\d\times\R^\d \setminus \triangle} \g(x-y) \,
 d\!\( \sum_{i=1}^N  \delta_{x_i}- \indic_{\car_R}\) \!(x)\,
 d\!\( \sum_{i=1}^N \delta_{x_i}- \indic_{\car_R}\)\! (y),
\end{equation*}
the set $\triangle :=\{ (x,x)\,:\,x\in\R^\d\}$ denotes the diagonal in $\R^\d\times \R^\d$ and $\indic_S$ the indicator of a set $S$. This perspective is related to the analysis in the present paper: we consider a variant of~\eqref{jell} with $\g$ replaced by the Neumann Green function of the cube $\car_R$, the partition function of which we denote by~$\K(\car_R)$ (see Theorem~\ref{th1}, below). As a byproduct of our analysis (we just apply the arguments verbatim with~$\mu=\indic_{\car_R}$ and replacing~$\PNbeta$ by~$\mathbb{Q}_{N, \beta}$), we thereby obtain analogous quantitative local laws and free energy expansions for~$\mathbb{Q}_{N, \beta}$ as we do for~$\PNbeta$.

\subsubsection*{The minimal lengthscale and the temperature dependence}
One of the main difficulties in handling the possibly large temperature regime is to obtain the factor $\beta\chi(\beta)$ instead of $1$ in the right-hand side of these estimates when $\beta$ is small. This is made possible by the use of the thermal equilibrium measure instead of the usual equilibrium measure.

The other main difficulty is to get the local laws  down to the minimal scale  $\rb$ of ~\eqref{minoR0}.  We believe that  the lengthscale 
$\max(1, \beta^{-\frac12} \chi(\beta)^{\frac12})$ is optimal in all dimension (or optimal up to the logarithmic correction in dimension $\d=2$). The conjectured scenario is that the Coulomb gas ressembles a Poisson process for lengthscales smaller than $\beta^{-\frac12} N^{-\frac1\d}$ and becomes rigid (in the sense that the number of points in cubes become constrained by the size of the cube) only at lengthscales larger than $\beta^{-\frac12} N^{-\frac1\d}$ as evidenced by Theorem \ref{th3}. 
If $\d \ge 5$ the additional condition in~\eqref{minoR0} makes the result most likely suboptimal, and is a limitation of the method due to boundary effects.

\smallskip

We are able to see the minimal lengthscale  $\beta^{-\frac12}$  (viewed at the blown-up level) arise in our proof because when implementing the bootstrap procedure, we control the (free) energy errors by $\beta \tilde \ell R^{\d-1} $ while controlling at the same time the volume errors by $R^{\d-1}/\tilde \ell$ (we believe these errors to be optimal), where $\tilde \ell$ is the lengthscale that we need to screen the configurations. 
Optimizing the total error 
\be\label{toter}
\beta \tilde \ell R^{\d-1}+ \frac{R^{\d-1}}{\tilde \ell}\ee
leads to $\tilde \ell =\beta^{-\hal}$, and since we always need  to keep $\tilde \ell<R$, the bootstrap terminates exactly for $R$ and $\tilde \ell$ of order $\beta^{-\hal}$. This way we can say that the configurations can effectively be screened with screening lengthscale $\beta^{-\hal}$ and down to that scale. 

\smallskip

Note that  the maximal size of a set $\Sigma$ in which $\mu=\mub'$ is bounded below by a positive constant independent of $N$ is (of order) $N^{\frac1\d}$, hence the results of the theorem are nonempty if and only if $ \rb \ll N^{\frac1\d}$ which is equivalent in dimension  $3\le \d \le 5$  to  $\theta \gg 1$ (we expect the same to be true if $\d \ge 5$). In the case $\d=2$ the results are nonempty if and only if  $\beta\gg
 \frac{\log N}{N}$.  Note that as soon as $\theta \ge \theta_0>0$, the third item in~\eqref{conddist0} can be absorbed into the first one, up to a  constant depending on $\theta_0$.
 
As mentioned above, the effective temperature at the macroscale is $\theta$ which gives rise to  a natural lengthscale  for variations of the macroscopic density  $\mu_\theta$ of  $\theta^{-\frac12}=\beta^{-\frac12} N^{-\frac1\d}$. This  strikingly coincides  with the 
 minimal lengthscale for microscopic ridigity  $\rb$.

It remains to understand more precisely what happens when~$\theta$ is fixed or~$\theta \to 0$. In the latter regime it would be more appropriate to strengthen the confinement, thus weakening the interaction.

\smallskip

The fact that  ~\eqref{loclawpoints0} gives a bound on all the moments of the number of points in a compact set centered at $x$ satisfying~\eqref{conddist0}   immediately yields the following statement. 

\begin{coro}[Limiting point processes]
\label{coro1}  
Under the same assumptions as in Theorem~\ref{th3}, for every $\beta>0$ fixed independently of $N$ and every point $x\in \Sigma$ with 
\begin{equation*}
\dist(x,\partial \Sigma) \geq 
\ed{C   N^{\frac{1}{ \d+2}} },
\end{equation*}
the law of the point configuration $\{x_1'-x, \dots, x_N'-x\}$ converges as $N \to \infty$, up to  extraction of a subsequence, to a limiting point process with simple points and finite correlation functions of all order.
\end{coro}
This is the first time that the existence of a limit point process is shown besides the particular determinantal case of $\beta =2$ in $\d=2$, for which the limit process is known to be the Ginibre point process, with an explicit correlation kernel. These processes can thus be thought of as $\beta$-Ginibre processes, at least in dimension $\d=2$.  We expect that they should satisfy a variational characterization as in Corollary \ref{coro2}.

\smallskip

In the 70's there was a large statistical mechanics literature (see  \cite{glm1,glm2,martinreview} and references therein) on sum rules and various relations for correlation functions of interacting particle systems, in particular Kirkwood-Salzbourg, BBGKY, KMS,  DLR equations. These can be shown to be equivalent relations in the case of regular interaction kernels but in the case of singular interactions like the Coulomb one, the existence of solutions to these hierarchies was not known. Corollary \ref{coro1} takes a small step toward putting these ideas on firmer ground by showing, up to a subsequence, the existence of limiting point processes.

\smallskip

Our next main result gives a quantitative estimate  of~$\log \K(\car_R)$ in the particular variant of the  Neumann jellium mentioned after~\eqref{jell}. Observe that the error term in~\eqref{1.26}, below, is negligible as soon as $ R \gg \rb$.
Extending this to varying background measures is one of the main objects of \cite{s2}.
  
 \begin{theo}[Free energy expansion, Neumann jellium case]\label{th1}
 There exists a function $\mf: (0, \infty) \to \R$  and a constant $C>0$ depending only on $\d$ such that 
  \be\label{bornesurf}
 -C\le \mf(\beta)\le C \chi(\beta)\ee 
 \be\label{bornesurfp} 
 \text{$\mf$ is locally Lipschitz in $(0,\infty)$ with} \  |\mf'(\beta)|\le \frac{C\chi(\beta)}{\beta}
 ,\ee 
 and such that 
 if $R^\d$ is an integer we have
  \be \label{1.26} \frac{\log \K(\car_R)}{\beta R^\d}=-   \mf(\beta) + 
 O\( \chi(\beta)\frac{\rb }{R}   + \frac{\beta^{-\frac1\d} \chi(\beta)^{1-\frac1\d} }{R}\log^{\frac1\d} \frac{R}{\rb}  \)
 \ee
 where $\rb$ is as in Theorem \ref{th3} and 
 the $O$ depends only on $\d$.

\end{theo}

 The function $\mf$ which depends only on $\beta$ (and $\d$) already implicitly appears in \cite{lebles} (combine relations (1.16) and (1.18) in \cite{lebles}) where it is given a variational interpretation :
 \be \label{fb}
\mf(\beta)= \min_{P}\(  \hal \widetilde{\mathbb{W}}(P)+ \frac{1}{\beta} \mathsf{ent}[P|\mathbf{\Pi}^1]\)\ee
where the minimum is taken over stationary point processes $P$ of intensity $1$,  $\widetilde{\mathbb{W}}(P)$ is the average with respect to $P$ of the ``Coulomb renormalized energy" (per unit volume) for an infinite point configuration with uniform background $1$ (see for instance \cite{rs,ln}, it is the $\widetilde {\mathbb{W}} (\cdot, 1)$ of \cite{lebles}), and $\mathsf{ent}[P|\mathbf{\Pi}^1]$ is the specific relative entropy (see \cite{fv}) of the point process $P$ with respect to the Poisson point process of intensity $1$. 
  Dimension $\d=2$ is particular since it is the only one where $\mf$ is not  expected to be bounded as $\beta \to 0$, in fact we expect the bound we have in $|\log \beta|$ to be optimal and to reflect the fact that the Poisson point process  should have infinite Coulomb energy $\widetilde{\mathbb{W}}$ in dimension $2$, in contrast with dimension $\d\ge 3$ where its energy is always finite as shown in \cite{leble}.  
  Note that the formula~\eqref{fb} implies that $\mf$ is a convex function of  $\frac{1}{\beta}$.

The error term in $1/R$ in~\eqref{1.26} corresponds exactly to a surface term. Such an error agrees with   the predictions on the next order term that are found in the physics literature in dimension $\d=2$ \cite{sha,cftw}, which are  made for a gas with quadratic confinement hence constant equilibrium measure, and which find a next order term in $\sqrt{N}$ for $N$ points ($\sqrt N$ corresponds to $R$ in dimension 2).

\smallskip
 
Once these results are proven, we briefly explain  how one can deduce a ``local'' large deviations principle,  generalizing the macroscopic scale LDP of~\cite{lebles} and the two-dimension mesoscale LDP of~\cite{loiloc} to arbitrary dimension and down to the smallest (microscopic) scale. More precisely, given $x_0$ in $\supp \, \muv$,  for a configuration 
$X_N$, defining  its blown up version to  be $X_N'= N^{\frac1\d} X_N$  we define the ``local empirical field" averaged in a cube of microscopic scale size $R$ around $x_0 \in  \supp \, \muv$ by 
\be \label{defii}
i_N^{x_0,R} (X_N):= \dashint_{\car_R(N^{1/\d} x_0)} \delta_{T_{ x}  X_N'|_{\car_R(N^{1/\d}x_0)} } dx\ee
where $T_x$ is the translation by $x$ and $|_{\car_R(N^{1/\d} x_0)}$ denotes the restriction of the configuration to $\car_R(N^{1/\d} x_0)$. 
In other words we look at a spatial average of the (deterministic) point process formed by the configuration.
We denote by $ \mathfrak{P}_{N,\beta}^{x_0,R}$ the push-forward of  $\PNbeta$ by $i_N^{x_0,R}$.
Finally we introduce the rate function of \cite{lebles} which is defined over the set of stationary point processes of intensity $m$
(equipped with the topology of weak convergence) by
\be \label{defFbeta}
\mathcal F_\beta^m (P):=\frac{ \beta}{2} \widetilde{\mathbb{W} }^m (P) + \mathsf{ent}[P|\Pi^m]\ee
where $\widetilde{\mathbb{W}}^m$ is the renormalized energy, precisely  defined in this context in \cite{lebles} (and originating in \cite{ssgl,ss1,rs})\footnote{it corresponds to the notation $\widetilde{\mathbb{W}}(\cdot, m)$ in \cite{lebles}.}, $\Pi^m$ is the (law of the) Poisson process of intensity $m$ over $\R^\d$, and $\mathsf{ent}$ is the specific relative entropy.
 We also have
 \be \min \mathcal F_\beta^m=\beta m^{2-\frac2\d} \mf(\beta m^{1-\frac2\d})-\(  \frac{\beta}{4} m\log m\)\indic_{\d=2} + m\log m\ee where $\mf$ is as in the previous theorem, this is the scaled version of~\eqref{fb}, and as  already seen in \cite{lebles}, if $\d\ge 3$ an effective temperature $\beta m^{1-\frac{2}{\d}}$ depending on the density of points appears here (as well as every time the density dependence is kept explicit).
 
 \smallskip
 
 We recall that in minimizing~\eqref{defFbeta} there is a competition (depending on $\beta$)
 between the energy term $\widetilde{\mathbb{W}}^m$ which prefers ordered configurations (energy-minimizing configurations are expected to be crystalline in low enough dimensions) and  the relative entropy term which favors disorder and configurations that are more Poissonnian. The choice of temperature scaling that we made in~\eqref{def:PNbeta} is  precisely the one for which these two competing effects are of comparable strength for fixed $\beta$.

\begin{theo}[Local large deviations principle]
\label{theoldp} 
Assume that, on its support, the equilibrium measure~$\muv$ is bounded below and belongs to~$C^{0, \kappa}$ for some $\kappa>0$. Assume that~$N^{\frac1\d}\gg R \gg \rb$ as $N \to \infty$ and $x_0 \in \supp\, \muv$ satisfies, for some $C>0$ depending only on $\d$ and $\muv$,
\begin{align*}
& 
\dist (x_0, \partial\, \supp\, \muv)
\\ & \quad
\geq C \,
\ed{ \max\( \(\frac{N^{\frac1\d} }{ \max(1, \beta^{-\hal} \chi(\beta)^{\frac12})}\)^{-\frac23}    ,   N^{\frac{1}{ \d+2}-\frac1\d}\) }+\frac{C}{\sqrt\theta}.
\end{align*}
Then we have the following:
\begin{itemize}
\item
If $\beta $ is independent of $N$,  the sequence $\{ \mathfrak{P}_{N,\beta}^{x_0,R}\}_N$ satisfies a LDP at speed $R^\d$ with rate function 
$\mathcal F_{\beta}^{\mu_V(x_0)}-\min \mathcal F_{\beta}^{\mu_V(x_0)}  .$
\item
If $\beta \to 0$ as $N\to \infty$, then  $\{ \mathfrak{P}_{N,\beta}^{x_0,R}\}_N$ satisfies a LDP at speed $R^\d$ with rate function $\mathsf{ent}[ P|\Pi^m]$. 
\item
If $\beta \to \infty$ as $N\to \infty$, then 
$\{ \mathfrak{P}_{N,\beta}^{x_0, R}\}_N$ satisfies a LDP  at speed $\beta R^\d$ with rate function
 $\hal(\widetilde{\mathbb{W}}^{\mu_V(x_0)}-\min \widetilde{\mathbb{W}}^{\mu_V(x_0)} )$.
 \end{itemize}
  \end{theo}

By Theorem~\ref{theoldp}, we recover for microscopic averages what was proven in~\cite{lebles} for limits of macroscopic averages and in~\cite{loiloc} for mesoscopic averages in dimension 2, and extend it to general~$\beta$. We note that the regime with~$R \simeq N^{\frac1\d}$ was treated in~\cite{lebles} for fixed~$\beta$ and can be extended without difficulty to the present setting of general $\beta$. It is for simplicity that we present results only for mesoscopic and microscopic averages (i.e., by taking that assumption that $N^{\frac1\d}\gg R \gg \rb$).

\begin{coro}\label{coro2}
Under the assumptions of Theorem~\ref{theoldp}, we have the following:
\begin{itemize}
\item If $\beta $ is independent of $N$, the point processes defined as subsequential limits of the push forward of $\PNbeta $ by the map  $i_N^{x_0,R}$ of~\eqref{defii}  must, after rescaling by $\mu_V(x_0)^{\frac1\d}$ and  for almost every $x_0$,  achieve the minimum in~\eqref{fb} among stationary point processes of intensity $1$. 
\item 
If $\beta \to 0$, they must be equal to the Poisson point process of intensity $1$.
\item If $\beta \to \infty$, they must minimize $\widetilde{ \mathbb{W}}^{1}$ among stationary point processes of intensity $1$.  \end{itemize}
\end{coro} 
Note that the point processes considered here are not exactly the same as those of Corollary~\ref{coro1} since they are obtained by first averaging over cubes.
Their stationarity  is a simple consequence of that averaging (see \cite{lebles} for a proof).
  Unfortunately we do not know whether a minimizer for~\eqref{fb} is unique (uniqueness has however been very recently proven for the $1$-dimensional log gas analogue in \cite{lebleandco}), it may very well not be, in particular if a phase transition happens at inverse temperature $\beta$. If it were, then this would provide the existence of a unique possible limit point process
along the whole sequence $N \to \infty$. 

\smallskip

Our results apply as well to minimizers of~$\HN$ (formally the case~$\beta=\infty$),  they then improve on the results obtained in two dimensions in~\cite{aoc} and~\cite{rns},   and their generalization to higher dimension in~\cite{PRN}.  It shows (as for the related problem in~\cite{aco}) that  the rate of convergence of the next-order energy is in~$\frac{1}{R}$, and gives equidistribution of points and energy down to the microscales, see Theorem~\ref{th4} in Section~\ref{secmini} for a precise statement.

\subsection{Method of proof}

As in \cite{bbny1,loiloc} and as first introduced in this context in \cite{rns}, the method relies on a \emph{renormalization procedure}, namely a  bootstrap on the length scales which couples the free energy expansion and the local law information: local laws at large (macroscopic scales) are used as a first input, and allow to get a first expansion of the free energy, which in turn yields local laws at a smaller scale, and then a better rate in the free energy expansion, etc, until one reaches the minimal scale $\rb$.

\smallskip


\smallskip

The starting point of our approach is, as in the previous papers \cite{ss1,rs,ps, lebles}, the ``electric" reformulation of the energy $\HN$, i.e. its rewriting in terms of the (suitably renormalized) Dirichlet energy of the {\it Coulomb (or electric) potential} generated by the points and the background $\mub$, which really leverages on the Coulomb nature of the interaction and the fact that the Coulomb kernel is, up to a multiplicative factor, the fundamental solution to a local differential operator,  the Laplacian. More precisely,
 we will see that after removing some fixed leading order term from $\HN$, we reduce to
 \be\label{nop}d\PNbeta(X_N)= \frac{1}{N^N\K(\R^\d)} \exp(-\beta \G(X_N)) d\mu^{\otimes N}(X_N)\ee
 where $\K(\R^\d)$ is the normalization constant and $\G$
is  a ``next-order energy" of the form 
\be\label{int1}\G(X_N)= \frac{1}{2\cd}\int_{\R^\d} |\nab u|^2\ee
where $$u=\g* \( \sum_{i=1}^N \delta_{x_i'}- \mub'\)$$
is the solution of
\be\label{equau}
-\Delta u= \cd\(  \sum_{i=1}^N \delta_{x_i'}- \mub'\),\ee 
where $\cd$ is such that $-\Delta \g=\cd \delta_0$. Here
 $x_i'=N^{\frac1\d} x_i$ and $\mu=\mub'(\cdot) = \mub(N^{\frac1\d}\cdot)$ represent the blown-up system, and in~\eqref{int1}  the integral needs to be understood in a ``renormalized" sense, 
see Section~\ref{sec2} for more precise definitions. The quantity 
$\G^{\car_R}$ encountered in Theorem~\ref{th1}  is then the analogue of~$\int_{\car_R} |\nab u|^2$ here.

\smallskip

 Our improvement of the scaling of the error in the free energy expansion is based on the idea of {\it quantifying the additivity of the energy} over subregions of the main domain.
  In the Coulomb gas setting, the additivity of the energy---once expressed in terms of the  Coulomb potential---was already observed and used crucially in \cite{ssgl,ss1,lebles}. It was proven via a  screening procedure inspired by the work of  \cite{aco} on a related problem and  introduced in the Coulomb context in  \cite{ssgl}, then improved in \cite{rs,ps}, which yielded non explicit error terms. In fact, this is the reason why the results of~\cite{loiloc} were limited to two dimensions. 
 
\smallskip

In this paper, we combine the screening procedure with the  idea of using two different convergent quantities to quantify the additivity error in the free energy: 
 the first quantity 
denoted $\G(X_N, \car_R)$ is the equivalent of~\eqref{int1} with~\eqref{equau} solved over the cube with zero Neumann boundary condition, while the second one denoted~$\F(X_N, \car_R)$, which is smaller, is the equivalent of~\eqref{int1} where~\eqref{equau} is solved over the  cube~$\car_R$  with  zero Dirichlet boundary condition. The true energy is naturally bounded below by~$\F$ and above by~$\G$, and we will obtain quantitative bounds on it indirectly, by estimating the difference between~$\F$ and~$\G$. These quantities are the analogues of those used in~\cite{aco} for the study of energy minimizers of a related problem. This idea of using two quantities which converge monotonically (after dividing by the volume) to the same limit was already present in~\cite{aco} and is related to a classical technique for estimating eigenvalues of the Laplacian under various boundary conditions that goes by the name \emph{Dirichlet-Neumann bracketing}. A similar idea also arose in a different context in the works~\cite{as,amourrat} on quantitative stochastic homogenization, and the central idea in these works of quantifying the additivity of the energy by a bootstrap (or renormalization) argument inspired the strategy of the present paper (see~\cite{akm} and references therein for more on these developments). 
The main difference here from previous works is that we must apply such ideas in a probabilistic setting, in the context of a Gibbs measure, rather than a deterministic variational problem. 

\smallskip
  
This requires us to revisit and significantly revise  the previous screening construction of \cite{ssgl,rs,ps}. We simplify it, optimize it and  turn  it into a probabilistic procedure by sampling the screening points from a given 
Gibbs measure instead of constructing them by hand. This allows to reduce the  energy and volume errors to surface terms as explained in~\eqref{toter}, which is crucial when treating the regime of small $\beta$.   In particular compared to \cite{lebles} we dispense with the use of several parameters which needed to be sent to $0$ with no explicit rates for the convergences. This is made possible by  a new  truncation approach borrowed from  \cite{lsz,ls2} and improved here.  
The precisely quantified  screening error allows to  estimate the additivity error of the free energies associated to (a variant of $\F$) and $\G$.  As in \cite{aco,as},  
   in view of their monotonicity one  then naturally  obtains a rate of convergence to the limit.   
    
\smallskip
      
    Let us now give a more precise glimpse into the bootstrap argument used to prove the central estimate, which is~\eqref{locallawint0}. We denote $\K(U)$ or $\K^\beta (U)$  the partition function associated  to the energy $\G$ in the set $U\subset \R^\d$.  We start by proving a first bound  of the form 
    \be \label{roughb}
    \left|\log \K(U)\right|
    \leq C \beta \chi(\beta) |U|\ee
    (modulo some additional error terms in dimension $\d \ge 5$). 
    The upper bound holds thanks to  the general lower bound  $\G(X_N) \ge - CN$ where $N$ is the number of points, equal to $\mu(U)$ (see Lemma \ref{minoG}).
    The lower bound holds thanks to a Jensen argument inspired by \cite{gz} (see Proposition \ref{minoK}). 
    Combining the lower bound for $\beta$ and the upper bound for $\beta/2$ we obtain that the local law~\eqref{locallawint0} holds at the largest scale $N^{\frac1\d}$. The result for smaller scales is then proved by a bootstrap: assuming it is true down to scale $2R$, we try to prove 
    that it is true down to scale $R$, as long as $R\ge \rb$.
    Let us consider a hyperrectangle  $\Omega\subset \Sigma$  of sidelengths comparable to $R$, such that $\mu(\Omega)$ is an integer, and let us denote   $\mathrm{n}= \mu(\Omega)$.

For any configuration  $\XN$ of points in $\R^\d$, let us denote by
$n$ the number of points it has in $\Omega$.   
To control the left hand side of~\eqref{locallawint0}, we start by  using ~\eqref{nop} to write that 
\begin{align}
\label{start}
&
\Esp_{\PNbeta}\( \exp\(\hal \beta \G^{\Omega} (\XN) \) \) 
\\ & \qquad \notag
\le \displaystyle  \frac{ 
\displaystyle \int_{(\R^\d)^N}   \exp\(- \hal \beta  \G^{\Omega} (\XN)\)    \exp\(- \beta \G^{\Omega^c} (\XN)\) d\mu(x_1) \dots d\mu(x_N)}
{\displaystyle \int_{(\R^\d)^N } \exp\( - \beta \G (\XN)\) d\mu(x_1)\dots d\mu(x_N)}.
\end{align}
We wish  to bound from above  the numerator and  bound from below the denominator.
To bound the numerator from above we  use the comparison between Neumann-based and Dirichlet-based energies which easily yields
$$\G^\Omega (\XN) \ge \F^\Omega(\XN|_\Omega) \ge \F(X_N|_\Omega, \Omega) \qquad \G^{\Omega^c}(X_N)  \ge   \F^{\Omega^c} (X_N|_{\Omega^c})  \ge \F(X_N|_{\Omega^c}, \Omega^c)$$
hence separating the integral according to the value of $n$, we find
\begin{align}
\label{sum} 
& \int_{(\R^\d)^N} \exp\(-\hal \beta \F^{\Omega} (\XN) \) \exp\(- \beta \F^{\Omega^c} (\XN)\) d\mu(x_1)\dots d\mu(x_N)
\\ & \qquad \notag
\leq  
\sum_{n=0 }^N 
\binom{N}{n}
\int_{\Omega^n} \exp\(-\hal\beta \F(X_n, \Omega) \) d\mu^{\otimes n} (X_n) 
\\ & \notag \qquad \qquad 
\times \int_{(\Omega^c)^{N-n}} \exp\(-\beta \F(X_{N-n}, \Omega^c) \) d\mu^{\otimes (N-n)} (X_{N-n}) .
\end{align}
On the other hand for the denominator we may use the subadditivity of~$\G$, which translates into a superadditivity of the associated partition function, to write that 
\begin{multline*}  \int_{(\R^\d)^N } \exp\( - \beta \G (\XN)\) d\mu(x_1) \dots \mu(x_N)
\\  
\ge 
\binom{N}{\mn} 
\int_{\Omega^\mn} \exp\(-\beta \G(X_\mn, \Omega) \) d\mu^{\otimes \mn} (X_\mn) \int_{(\Omega^c)^{N-\mn}} \exp\(-\beta \G(X_{N-\mn}, \Omega^c) \) d\mu^{\otimes (N-\mn)} (X_{N-\mn}) .\end{multline*}
 We can expect the sum  above to concentrate near $n \simeq \mn$, because other terms correspond to a large discrepancy in  the number of points in~$\Omega$, which we can show leads to a large energy in~$\Omega$.
 Reducing to such terms, in order to bound the left-hand side of~\eqref{start} the next step is 
  to bound from above the Dirichlet energy associated to $\F$ in terms of that associated to~$\G$. That is, we show that we may replace~$\F$ with~$\G$ in the right-hand side of~\eqref{sum}, up to a suitably small error. Then there only remains~$\K^{\beta/2}(\Omega)/\K^\beta(\Omega)$ in the right-hand side of~\eqref{start}, for which we have the desired bound (in~$C \beta \chi(\beta) R^\d$) thanks to~\eqref{roughb}.
      The core of the work is thus to
    prove that  $$  \int_{\Omega^n} \exp\(-\beta \F(X_n, \Omega) \) d\mu^{\otimes n} (X_n)\le  \int_{\Omega^\mn} \exp\(-\beta \G(X_\mn ,\Omega) \) d\mu^{\otimes \mn} (X_\mn)$$  and the same in $\Omega^c$, up to a small error.  This is accomplished thanks to the configuration-by-configuration  screening procedure, which replaces each configuration $X_n$ of $n$ points by a configuration $X_\mn$ of the correct number of points $\mn$  that coincides with $X_n$ except on a thin boundary layer.  The energy and volume errors associated to the procedure  it are kept  as the small surface errors mentioned in ~\eqref{toter}
     by using the fact that the local laws hold at the slightly larger scale $2R$ which provides good energy controls.
     
\smallskip
         
  The  local law ~\eqref{locallawint0} also allows to show the   additivity of the  free energy itself,  up to  the surface error  terms in $\beta R^{\d-1}$ (roughly) for a cube of size $R$.  As in \cite{aco,amourrat},  this is the best that one can hope with such a method. This implies the existence
  of the order $N$ term in  the free energy expansion as in~\cite{LiLe1,LN}, except now with an explicit convergence rate. In that sense, what we prove is a \emph{quantitative thermodynamic limit}.  
 Note that an expansion  up to order $N$  of the free energy with the variational interpretation~\eqref{fb} for the order $N$ coefficient and an error $o(N)$, was already obtained in all dimensions in \cite{lebles}, it came as a corollary of the  LDP. In the two-dimensional case, an error term of $N^{1-\ep}$ for some small (explicit) $\ep>0$ was obtained in \cite{bbny2} by a Yukawa approximation argument.

\smallskip

In \cite{ps,lebles}, we treated Riesz interactions and one-dimensional logarithmic interactions as well as Coulomb interactions (with the important motivation of log gases). This introduced some (not only notational) complications because Riesz kernels are kernels of nonlocal operators and a dimension extension is needed. This is why we  leave the generalization to Riesz and one-dimensional log gases to future work.


\subsection{Outline of the paper}
The paper is organized as follows.
 In Section \ref{sec2} we introduce the precise definitions of the sub and superadditive energies, the appropriate renormalizations (whose specifics are new), and of the corresponding partition functions.
 In Section \ref{sec3} we give some preliminary results, including the sub and superadditivity of the energies, a priori bounds on the energies and partition functions.  Estimates showing how the energies control the fluctuations of the configurations and adapted from previous work are gathered in Appendix \ref{appendixa}.
 In Section \ref{sec4} we give the main newly optimized result of the screening procedure, that allows to bound from above  the additivity error. The proof of the screening itself is postponed to Appendix \ref{appa}.
  Section~\ref{sec5} is the core of the proof that accomplishes the bootstrap procedure: starting from the a priori bounds on the largest scale, it shows how the screening allows to obtain energy controls on smaller and smaller scales.
  In Section \ref{sec6} we investigate the consequences of the bootstrap procedure and deduce from the local laws  the proof of the almost additivity of the free energy hence the  free energy expansion with a rate, in the uniform background case.
  In Section \ref{sec9} we describe the proof of the LDP result of Theorem \ref{theoldp}. Finally in Section \ref{secmini} we adapt our results to the case of energy minimizers to obtain Theorem \ref{th4}.

\smallskip

\noindent
{\bf Acknowledgements.} We thank Thomas Lebl\'e for helpful discussions and comments. The first author was supported by NSF grant DMS-1700329 and a grant of the NYU-PSL Global Alliance. The second author was supported by NSF grant DMS-1700278 and by a Simons Investigator grant.

\section{Additional  definitions}\label{sec2}
 
\subsection{Splitting formula and rescaling}\label{secsplit}
We adapt here the splitting formula, introduced in \cite{ssgl,rs}. It is an exact formula that allows to separate the leading order term in the energy from the next order term, already giving the leading order coefficient in the free energy expansion. 
Here we provide a new  formula by  expanding the energy around  the thermal equilibrium measure $N\mub$, yielding more exact results and allowing to prove the local laws even when the  temperature gets large.

\smallskip

We recall that $\theta= \beta N^{\frac2\d}$ and that the thermal equilibrium measure $\mub$ minimizing ~\eqref{1.9} 
satisfies 
\be \label{eqmb}
\g* \mub+V + \frac{1}{\theta}\log \mub=C\quad \text{in} \ \R^\d\ee
where $C$ is a constant.
  We then define 
 \be \label{zetab}
 \zb:= -  \frac{1}{\theta}\log \mub.
 \ee

\begin{lem}[Splitting formula with the thermal equilibrium measure]
For any configuration $\XN \in (\R^\d)^N$,  we have 
\begin{multline}\label{splitting}
\HN(\XN)= N^2 \mathcal {E}_\theta(\mub)+ N\sum_{i=1}^N \zb(x_i) \\
+ \hal \iint_{\triangle^c} \g(x-y)d\( \sum_{i=1}^N \delta_{x_i} -N\mub  \) (x)d\( \sum_{i=1}^N \delta_{x_i} -N\mub\)(y)
\end{multline}
where $\mathcal{E}$ is as in~\eqref{defE},  $\zb$  as in~\eqref{zetab} 
and $\triangle$ denotes the  diagonal in $\R^\d\times \R^\d$.\end{lem}
\begin{proof}
It suffices to rewrite $\HN(\XN)$ as
$$\HN(\XN)= \iint_{\triangle^c} \g(x-y)d \(  \sum_{i=1}^N \delta_{x_i}\)(x)d \( \sum_{i=1}^N \delta_{x_i}\) (y) + N \int_{\R^\d} V(x) d \(  \sum_{i=1}^N \delta_{x_i}\)(x),$$
expand the integral after writing 
$\sum_{i=1}^N \delta_{x_i}= N\mub+ \( \sum_{i=1}^N \delta_{x_i}- N\mub\)$ and use~\eqref{eqmb}.
\end{proof}

Let us point out that as mentioned in the introduction, from this formula we see $-\frac{1}{\theta} \log \mub$ appearing as an effective confining potential (in place of $\zeta$ in the previous splitting formula of \cite{ss2,rs}).
We next  rescale the  coordinates by  setting $\XN'$ to be the configuration $(N^{\frac1\d}x_1, \dots, N^{\frac1\d} x_N)$. 
The blown-up thermal equilibrium  measure is $\mub'(x)= \mub(xN^{-\frac{1}{\d}})$, it is  a measure of mass $N$ which slowly varies. We also define the rescaling of $\zb $ to be $\zb'(x)= N^{\frac2\d} \zb(x N^{-\frac1\d})$. By definition~\eqref{zetab} we thus have 
\be \label{zetabp}
\zb'(x)=  - \frac{1}{\beta}  \log \mub' (x).\ee
We also have the following scaling relation 
\begin{multline}
\label{scalingF} 
 \iint_{\triangle^c} \g(x-y)d\( \sum_{i=1}^N \delta_{x_i} -N\mub  \) (x)d\( \sum_{i=1}^N \delta_{x_i} -N\mub\)(y)
\\=  N^{1-\frac{2}{\d}}    \iint_{\triangle^c} \g(x-y)d\( \sum_{i=1}^N \delta_{x_i'}- \mub'\)
 (x) d\( \sum_{i=1}^N \delta_{x_i'}- \mub'\) (y) - \(\frac{N}{2}\log N\) \indic_{\d=2}.\end{multline}

We may now define for any point configuration $\XN$ and density $\mu$,  the next-order energy to be
\begin{equation}
\label{defF}
\G(\XN, \mu):=  \hal \iint_{\triangle^c} \g(x-y)d \(  \sum_{i=1}^N \delta_{x_i} -\mu\) (x) d\(  \sum_{i=1}^N \delta_{x_i} -\mu\) (y),\end{equation}  and the next-order partition function to be 
\begin{equation}\label{pdef}
\K(\mu):= N^{-N} \int_{(\R^\d)^N} \exp\( - \beta \G(\XN, \mu)\) d\mu^{\otimes N} (\XN) .\end{equation}

Inserting~\eqref{splitting},~\eqref{zetabp} and~\eqref{scalingF} into~\eqref{defZ},  and using the change of variables $\XN'= N^{\frac1\d}\XN$ and~\eqref{1.9}, we directly find
\begin{equation}\label{resplitz}
\ZNbeta= \exp\(-\beta N^{1+\frac2\d} \mathcal {E}_\theta(\mub)+ \(\frac{\beta}{4}N \log N\)\indic_{\d=2} \)  \K(\mub').\end{equation} 

Note that a main difference with using the previous splitting formula is that here no effective confining potential term remains, and the reduced partition functions are defined with integrations against $\mu^{\otimes N}$ instead of the Lebesgue measure, which makes handling the entropy terms much more convenient.

\smallskip

From now on we will thus be interested in expanding the logarithm of partition functions of the type~\eqref{pdef} for generic densities $\mu$ such that $\int_{\R^\d} d\mu=N$.



\subsection{Electric formulation and truncations}
We now focus on reexpressing $\G(\XN ,\mu)$ in ``electric form", i.e via the electric (or Coulomb) potential generated by the points.
This is the crucial ingredient that exploits the Coulomb nature of the interaction and makes the energy additive.
We rely here on  a rewriting via  truncations as in \cite{rs,ps} but  using  as in \cite{lsz,ls2} the nearest neighbor distance as a   specific truncation  distance so that   no error term is created. This technical improvement is crucial and in particular allows to dispense with the ``regularization procedure" of \cite{lebles}.

We consider the potential $h$ created by the configuration $\XN$ and the background $\mu$, defined by 
\begin{equation}
\label{def:hnmu} 
h(x) := \int_{\R^\d} \g(x-y) d\(\sum_{i=1}^N \delta_{x_i} - \mu\)(y).\end{equation}
Since $\g$ is (up to the constant $\cd$), the fundamental solution to Laplace's equation in dimension $\d$, we have 
\be -\Delta h= \cd\( \sum_{i=1}^N \delta_{x_i} - \mu\).\ee

We note that $h$ tends to $0$ at infinity because $\int \mu=N$ and the ``system" formed by the positive charges at $x_i$ and the negative background charge $N\mu$ is neutral.
We would like to formally rewrite $\G(\XN,\mu)$  defined in~\eqref{defF} as $\int |\nab h|^2$, however this is not correct due to the singularities of $h$ at the points $x_i$ which make the integral diverge. 
This is why we use a truncation procedure which allows to give a renormalized meaning to this integral. 

For any number $\eta>0$, we denote
\be\label{def:truncation} 
\f_{\eta} (x) := (\g(x)-\g(\eta))_+,\ee where $(\cdot)_+$ denotes the positive part of a number, 
and point out that $\f_\eta$ is supported in $B(0,\eta)$. This is a truncation of the Coulomb kernel.
We also denote by $\delta_{x}^{(\eta)}$ the uniform measure of mass $1$ supported on $\partial B(x, \eta)$, which is a smearing of the Dirac mass at $x$ on the sphere of radius $\eta$. We notice that 
\be \label{fconv} \f_{\eta}=\g*\( \delta_0-\delta_0^{(\eta)}\)\ee
so that 
\be\label{deltaf}-\Delta \f_{\eta}= \cd\(  \delta_0-\delta_0^{(\eta)}\).\ee

 For any $\vec{\eta}=(\eta_1, \dots, \eta_N)\in \R^N$, and any  function $h$ satisfying a relation of the form 
\begin{equation}\label{formu}
-\Delta h =  \cd\(\sum_{i=1}^N \delta_{x_i}- \mu\)
\end{equation}
we then define the truncated potential
\begin{equation}\label{formu2}
h_{\vec{\eta}}= h- \sum_{i=1}^N \f_{\eta_i}(x-x_i).\end{equation}
  We note that  in view of~\eqref{deltaf} the function $h_{\vec{\eta}}$ then satisfies 
  \be -\Delta h_{\vec{\eta}} = \cd\( \sum_{i=1}^N \delta_{x_i}^{(\eta_i)}- \mu\).\ee

We then define a particular choice of truncation parameters: if $\XN= (x_1, \dots, x_N)$ is a $N$-tuple of points in $\R^\d$ we denote for all $i=1, \dots, N$,
\begin{equation}\label{def:trxi}
\rr_i := \frac{1}{4} \min\left(\min_{j \neq i} |x_i-x_j|, 1 \right)
\end{equation}
which we will think of as the \textit{nearest-neighbor distance} for $x_i$. 
The following is proven in \cite[Prop. 2.3]{ls2} and \cite[Prop 3.3]{smf} (here we just need to rescale it).
\begin{lem} \label{lem:monoto} 
Let $\XN$ be in $(\R^\d)^N$. If $ (\eta_1, \dots, \eta_N)$ is such that $0 < \eta_i \le \rr_i$ for each $i = 1, \dots, N$, we have
 \begin{equation}
\label{fnmeta}
 \G(\XN,\mu) = \frac{1}{2\cd} \left(\int_{\R^\d}|\nab h_{ \vec{\eta}}|^2  -\cd \sum_{i=1}^N  \g(\eta_i)  \right)  
 - \sum_{i=1}^N \int_{\R^{\d}} \f_{\eta_i}(x - x_i) d\mu(x),
\end{equation} where $h$ is as in~\eqref{def:hnmu}.
 \end{lem}
 This shows in particular that the expression in the right-hand side is independent of the truncation parameter, as soon as it is small enough.  Choosing $\eta_i=\rr_i$ thus yields an exact (electric) representation for $\G$. In Appendix \ref{appendixa} we provide monotonicity results which  show that taking truncation parameters $\eta_i$ larger than $\rr_i$ can only decrease the value of the right-hand side of~\eqref{fnmeta}.

\subsection{Dirichlet and Neumann local problems}
We now  introduce new local versions of these problems, that will serve to  define the sub and super additive energy approximations.
Let us consider $U$ a subset of $\R^\d$ with piecewise $C^1$ boundary, bounded or unbounded.  Most often, $U$ will be $\R^\d$ or a hyperrectangle or the complement of a hyperrectangle.
Although $N$ originally denoted the number of points in $\R^\d$ and defined the blown-up scale at which we are working, when ambiguous we will also use the notation $N$ to denote the total number of points a system has in a generic set $U$, which may not be the whole space.

\smallskip

The main quantity we will use is one obtained by solving a relation of the type~\eqref{formu} with a zero Neumann boundary condition.
We need to introduce a new modified version of the minimal distance to make the energy subadditive:
 we let
\begin{equation}\label{defrrc}
\rrh_i := \frac{1}{4}\min \( \min_{x_j\in U, j\neq i} |x_i-x_j|, \dist(x_i, \p U), \dm\).
\end{equation} 
In order to have an energy which is always bounded from below, we need to add some energy to points that approach the boundary. 
To that effect we   define
\be \label{defGG}\GG(x_i):= \(\g\(\tfrac14\dist(x_i, \pa U)\) - \g\(\tfrac14\)\)_+.\ee
If $\mu (U)=N$, an integer,  
 for a configuration $\XN$ of points in $U$,  we now define
\begin{equation}\label{minneum}
\G(X_N,\m,U) :=\frac{1}{2\cd}\( \int_{U}|\nab u_{\rrh}|^2 - \cd \sum_{i=1}^N \g(\rrh_i)\)  - \sum_{i=1}^N \int_{U} \f_{\rrh_i}(x - x_i) d\m(x)+
\sum_{i=1}^N \GG(x_i)
\end{equation}
where  $u$ solves
\begin{equation}\label{defv}
\left\{\begin{array}{ll}
 -\Delta u = \cd\Big( \sum_{i=1}^{N} \delta_{x_i}- \m \Big) &\ \text{in} \ U \\
 \frac{\pa u}{\pa \nu}=0 &\ \text{on} \ \partial U. \end{array}\right.
\end{equation}
Note that  under the condition $ \m(U)  =N$ the solution of~\eqref{defv} exists and is unique up to addition of a constant.
Unless ambiguous, we will denote $\G(\XN, U)$ instead of $\G(\XN, \mu, U)$.
We note that from~\eqref{fnmeta}, $\G(\cdot, \R^\d)$ coincides with $\G$ defined in~\eqref{defF}.

\smallskip

We will use  a localized version of this energy: if $u$ solves~\eqref{defv} and $\Omega $ is a strict (closed) subset of $U$, we define 
\be\label{rrc}  
\rrc_i := \frac14
\left\{ 
\begin{aligned}
& \min \(\min_{x_j \in \Omega, j\neq i} |x_i-x_j|, \dist(x_i, \partial U \cap \Omega),  \dm\) && \text{if} \ \dist(x_i, \partial \Omega \backslash \partial U ) \ge \hal,\\
& \min\(1, \dist (x_i, \pa U\cap \Omega) \)&& \text{otherwise.}
\end{aligned}
\right. 
\ee
and
\begin{align}
\label{Glocal}
&
\G^{\Omega}(X_N, U)
\\ & \quad \notag
:=
\frac{1}{2\cd}\( \int_{\Omega} |\nab u_{\rrc}|^2 - \cd \sum_{i, x_i \in \Omega} \g(\rrc_i) \) -
\sum_{i, x_i\in \Omega}\int_U \f_{\rrc_i}(x - x_i) d\m(x) +
\sum_{i,x_i \in \Omega} \GG(x_i).
\end{align}
We will also use the following variant of $\rrc_i$ which only differs near $\pa \Omega\setminus \pa U$:
\be\label{rrtt}  
\rrtt_i := \frac14
\left\{ 
\begin{aligned}
& \min \(\min_{x_j \in \Omega, j\neq i} |x_i-x_j|, \dist(x_i, \partial U \cap \Omega) , \dm\) && \text{if} \ \dist(x_i, \partial \Omega \backslash \partial U ) \ge 1,\\
& \min\(1, \dist (x_i, \pa U\cap \Omega) \)&& \text{otherwise.}
\end{aligned}
\right. 
\ee
Let us point out that when $\Omega=U$ then $\rrc_i=\rrtt_i=\rrh_i$.

Our second quantity is obtained by minimizing the energy with respect to all possible functions $u$ compatible with the points 
in the sense of satisfying~\eqref{formu}, it naturally leads to a Dirichlet problem and to a superadditive energy.
For a configuration $X_N $ of points in $U$, imitating~\eqref{fnmeta} we    define the energy relative to the set $U$ as 
\begin{equation}\label{mindiri}
\F(X_N,\mu, U):= \frac{1}{2\cd} \( \int_{U} |\nab v_{\rrc} |^2 - \cd \sum_{i=1}^N \g(\rrc_i)\) - \sum_{i=1}^N \int_{U} \f_{\rrc_i}(x - x_i) d\m(x) \end{equation}
where $\rrc$ is as in~\eqref{rrc} with $\varnothing$ in place of $U$ and~$U$ in place of~$\Omega$, and  
\begin{equation}\label{defu}
\left\{\begin{array}{ll}
 -\Delta v = \cd\Big( \sum_{i=1}^{N} \delta_{x_i}- \m \Big) &\ \text{in} \ U \\
 v_{\rrc}=0 &\ \text{on} \ \partial U. \end{array}\right.
\end{equation}
We will often omit (unless ambiguous) the dependence in $\mu $  in the notation and simply write $\F(\XN, U)$. Using standard variational arguments, we may check that we have
\begin{multline} 
\label{3.4}
\F(X_N, U)
=
\min \Bigg\{ \frac{1}{2\cd}\( \int_{U} |\nab u_{\rrc} |^2 - \cd \sum_{i=1}^N \g(\rrc_i) \)- 
\sum_{i=1}^N \int_{U} \f_{\rrc_i}(x - x_i) d\m(x) \,:\, \\
 -\Delta u = \cd\Bigg( \sum_{i=1}^{N} \delta_{x_i}- \m \Bigg) \ \text{in} \ U\Bigg\}.\end{multline}
We will not use $\F$ very much but rather  a variant (mixed version of the energy), for  $\Omega$ a subset of $U$ that may touch $\partial U$.  
 For $\XN$ a configuration of $N$ points in $\Omega \cap U$, imitating the definition of $\F$ we set
\begin{equation}
\label{minmix}
\H_U(X_N, \Omega) :=  \frac{1}{2\cd}\( \int_{\Omega\cap U} |\nab w_{\rrc}|^2 - \cd \sum_{i=1}^N \g(\rrc_i) \)
- \sum_{i=1}^N \int_{\Omega \cap U} \f_{\rrc_i} (x-x_i) \,d\mu(x),
\end{equation}
where  $\rrc$ is as in~\eqref{rrc}  and 
\be \label{defw}
\left\{\begin{array}{ll}
 -\Delta w = \cd\Big( \sum_{i=1}^{N} \delta_{x_i}- \m \Big) &\ \text{in} \ \Omega \cap U \\
 \frac{\pa w}{\pa \nu}=0 &\ \text{on} \ \partial U\cap \Omega\\
 w_{\tilde\rr}=0 &\  \text{on} \ \partial (\Omega \cap U) \backslash \partial U
 .\end{array}\right.
\end{equation}
We can check that 
\begin{multline}\label{232}
\H_U(X_N, \Omega)\\= \min \Bigg\{ \frac{1}{2\cd}\( \int_{\Omega\cap U} |\nab w_{\rrc}|^2 - \cd \sum_{i=1}^N \g(\rrc_i) \)
- \sum_{i=1}^N \int_{\Omega \cap U} \f_{\rrc_i} (x-x_i) d\mu(x) \,:\,\\
-\Delta w=\cd\( \sum_{i=1}^N \delta_{x_i} -\mu\) \  \text{in} \ U\cap \Omega, \quad \frac{\partial w}{\partial \nu}=0\  \text{on} \ \partial U \cap \Omega\Bigg\}.
\end{multline}
We then define a localized version: if $\omega$ is a subset of $\Omega$,
\be  \H_U^{\omega}(X_N, \Omega) :=  \frac{1}{2\cd}\( \int_{\omega\cap U} |\nab w_{\rrc}|^2 - \cd \sum_{i, x_i \in \omega }\g( \rrc_i) \)
- \sum_{i, x_i \in \omega} \int_{\Omega \cap U} \f_{\rrc_i} (x-x_i) d\mu(x)  ,\ee
where  $\rrc_i$  is now relative to $\partial \omega$. 
We note that if $U=\R^\d$ or if $\Omega $ is a strict subset of $U$, $\H_U$ coincides with $\F$.

\subsection{Partition functions}

We next define a partition function relative to~$U$. If~$\m(U) =N$, then we define
\begin{equation}
\label{defK}
\K(U, \mu):=  N^{-N}\int_{U^N} 
\exp\left(- \beta \G(X_N,\m,U) \right)
\, d\mu^{\otimes N} (X_N).\end{equation}
We also define the  associated Gibbs measure by \begin{equation}\label{defQ}
\Q(U, \m)
:= 
\frac{1}{N^N\K(U, \mu)} 
\exp\left( -\beta \G(\XN,\mu, U)\right)
\, d\mu^{\otimes N} (X_N).
\end{equation}
We may also consider in the same way (although we will not give details)
\begin{equation}\label{defP}
\P_N (U, \m) := \frac{1}{N^N \L_N(U, \mu)} 
\exp\left( -\beta  \F(\XN,\mu, U) \right) 
\, d\mu^{\otimes N} (X_N).
\end{equation}



We will assume without loss of generality that the points in $X_{N}$ never intersect the boundary of the considered cubes, which is legitimate since this would correspond to a zero-measure set in the integrals defining $ \K$. 
As above, we will simply denote (unless ambiguous) these quantities by  $\K(U)$ and $\Q(U)$.
We note that    $\K(\R^\d,\mu)$  coincides with  $\K(\mu)$ defined in~\eqref{pdef}, and that in view of the splitting formula~\eqref{splitting} and~\eqref{resplitz},   $\Q(\R^\d)$   coincides with the original Gibbs measure $\PNbeta$ defined in~\eqref{def:PNbeta}.

    \section{Preliminary results}\label{sec3}

\subsection{Partitioning into hyperrectangles with quantized mass}
We will use throughout the paper the following definition.
\begin{defi}\label{def32}
For any $R$ we let $\mathcal Q_R$ be the set of hyperrectangles $Q$ whose sidelengths belong to $[R, 2R]$ 
and which are such that 
$\mu(Q)$ is an integer.
\end{defi}

\begin{lem}\label{rect}Assume $\mu \ge m>0$ in a set $U$.
 There exists a constant $R_0>0$ depending only on $\d$  and $m$ such that 
   given any $R>R_0$ 
 there exists a   collection $\mathcal K_R$ of closed hyperrectangles with disjoint interiors belonging to $\mathcal Q_R$,   and such that
\begin{equation}\label{inclusion} \{x\in U: d(x,\partial U)\le R \}\subset  U\setminus\bigcup_{K\in\mathcal K_R} K \subset \{x\in U : d(x,\partial U)\le  2R \}.\end{equation}
Moreover, if $U$ is a hyperrectangle, we can require that $\bigcup_{K\in\mathcal K_R} K=U$.
\end{lem}

\begin{proof} 
The proof can easily be adapted from  \cite[Lemma 7.5]{ss1}.
\end{proof}


  \subsection{Sub and superadditivity}
  Here, we show that $\G$ is subadditive, as desired (one can also easily check that $\F$ is superadditive).
  We will use various results on the monotonicity of the energy with respect to the truncation parameter, which are stated and  proven in Appendix \ref{appendixa}.
  In the rest of the paper, when talking about ``disjoint union of two sets", we mean the union of the closures of two sets whose interiors are disjoint.
\begin{lem}\label{additi}
For any configuration $\XN$ defined in $U$ with $N=\mu(U)$, if $\Omega$ is a subset of $U$ and $\omega$ a subset of $\Omega$, we have 
\be 
\label{restri1} \G^{\Omega}(\XN, U) \ge \H_U(\XN|_{\Omega}, \Omega)\qquad
\H_U^{\omega}(\XN, \Omega) \ge \H_U^\omega(\XN|_{\omega}, \omega) .
\ee
and if $\omega$ is the disjoint union of $\omega_1 $ and $\omega_2$, 
\be \label{add3}
\H^{\omega}_U(X_N, \Omega) \ge \H^{\omega_1}_U(X_N, \Omega) + \H^{\omega_2}_U(X_N, \Omega) . 
\ee

\end{lem}
\begin{proof}
 Let us first change $\rrc_i$ relative to $\omega$ into $\rrc_i $ relative to $\omega_1 $ for $x_i \in \omega_1$ respectively $\rrc_i $ relative to $\omega_2$ for $x_i \in \omega_2$. This   increases these  truncation parameters,  hence in view of Lemma \ref{monoto},  it may only  decrease the computed value of $\H_U$.
Splitting the obtained integral into two pieces we deduce that 
\begin{multline*}\H_U^{\omega}(\XN,  \Omega) \ge \frac{1}{2\cd}\( \int_{\omega_1} |\nab w_{\rrc_i}|^2 -\cd \sum_{i, x_i \in \omega_1} \g(\rrc_i)\)  - \sum_{i, x_i \in \omega_1}\int_{\Omega \cap U} \f_{\rrc_i} (x-x_i) d\mu(x) \\
+  \frac{1}{2\cd}\( \int_{\omega_2} |\nab w_{\rrc_i}|^2 -\cd \sum_{i, x_i \in \omega_2} \g(\rrc_i)  \)- \sum_{i, x_i \in \omega_2}\int_{\Omega\cap U}\f_{\rrc_i} (x-x_i) d\mu(x),
\end{multline*}
where $w$ is as in~\eqref{defw} and the $\rrc_i$ are those relative to $\omega_1$, resp. $\omega_2$.
 It follows that~\eqref{add3} holds. 
 The first item of~\eqref{restri1} is a consequence of the minimality property~\eqref{232}. The second item is proven  by using the minimality property~\eqref{232}.  
\end{proof}

 As already observed and used in \cite{ssgl,ss1,rs,ps,lebles}, solving Neumann problems  allows to get   subadditive energy estimates over subcubes by using the following lemma (whose proof we omit) which exploits that the Neumann electric field is the $L^2$ projection of any compatible electric field onto gradients.  
  \begin{lem}[Projection lemma]\label{projlem}
  Assume that $U$ is a compact subset of $\R^\d$ with piecewise $C^1$ boundary.
  Assume $E$ is a vector-field satisfying a relation of the form
  \begin{equation}\label{eqe}
   \left\{\begin{array}{ll}  
  -\div E= \cd\( \sum_{i=1}^N \delta_{x_i} -\mu\) &\quad \text{in}  \ U\\
  E \cdot \nu=0 & \text{on} \ \partial U,\end{array}\right.\end{equation}
and $u$ 
  solves 
$$  \left\{\begin{array}{ll}  
  -\Delta u= \cd\( \sum_{i=1}^N \delta_{x_i}- \mu\)& \quad \text{in}  \ U\\
  \frac{\p u}{\p \nu}=0 & \text{on} \ \partial U.\end{array}\right.$$ 
  Then
  $$\int_{U} |\nab u_{\rrc}|^2 \le \int_U |E-\sum_{i=1}^N \nab \f_{\rrc_i}(\cdot -x_i) |^2.$$ 
  \end{lem}

We now check that the energies $\G$ is  sub-additive, as desired. One can check that $\F$ is superadditive as a consequence of~\eqref{add3}.

\begin{lem}[Sub and superadditivity] Assume $U$ is the union of two sets $ U_1$, $ U_2$ with disjoint interiors and   piecewise $C^1 $ boundaries. If $\XN$ is a configuration in $U_1 $ and $Y_{N'}$ a configuration in $U_2$ with $\mu(U_1)=N$, $\mu (U_2)=N'$, then 
\begin{equation}\label{subad1}
\G(X_N\cup Y_{N'},  U) \le \G(X_{N},  U_1) + \G(Y_{N'},  U_2) .\end{equation}
\end{lem}
\begin{proof}
 For~\eqref{subad1}, let $u$ and $u'$ be the solutions to the Neumann problems associated with the definition of $\G$ in~\eqref{minneum} and set $E= \nab u$, $E'= \nab u'$. 
 We have 
 \begin{equation}\label{dve}
 - \div E= \cd \Big( \sum_{i=1}^{N} \delta_{x_i} - \m\Big)  \ \text{in} \ U_1 \qquad -\div E'= \cd\Big(\sum_{i=1}^{N'}\delta_{y_i}- \m\Big)\ \text{in} \ U_2 .\end{equation}
  We may now define $E^0= E\indic_{U_1}+ E'\indic_{U_2}$ and note that it satisfies 
  \begin{equation}\label{eqq2}
 \left\{\begin{array}{ll}
 -\div E^0= \cd\Big(\sum_{p\in \XN\cup Y_{N'}} \delta_p - \mu\Big)
 &\quad \text{in} \ U\\
 E^0\cdot  \nu = 0 & \quad\text{on} \ \p U \end{array}\right.\end{equation}
 Indeed, no divergence is created across $\partial U_1\cap \partial U_2$ thanks to the vanishing normal components on both sides.
The result then follows from Lemma \ref{projlem}.
\end{proof}
  
The subadditivity property has the following counterpart for the partition functions.
\begin{lem}Assume $U$  is  partitioned into $p$ disjoint sets $Q_i$,  $i\in [1,p]$ which are such that $\mu(Q_i)=N_i$ with $N_i$ integer.
We have
\begin{equation}\label{superad2}
\K(U) \ge    \frac{N! N^{-N}}{{N_1}! \dots {N}_p! N_1^{-N_1} \dots N_p^{-N_p}}    \prod_{ i=1  }^p \K(Q_i).\ee
\end{lem}
\begin{proof}
It suffices to partition the phase space into sets of the form $\{x_{i_1}, \dots, x_{i_{N_j}} \in Q_j\}$ for each $j=1,\dots, p$, then to use  
~\eqref{subad1},  noting  that 
the number of ways to distribute $N$ points in the $p$ sets with ${N}_i$ points in each set is 
$
 \frac{N!}{{N}_1! \dots {N}_p! } $.
\end{proof}

\subsection{Preliminary energy and free energy controls}
We start with a  rough bound  on~$\G$ which yields an upper bound for~$\K$.

 \begin{lem}[Upper bound for $\K(U)$]\label{minoG}
Assume $\mu(U)=N$, then we have for any $\XN $, 
 \begin{equation}\label{minog}
 \G(\XN, U)\ge - C N,\end{equation}
  and
  \begin{equation}\label{majok}
\log \K(U) \le  C \beta N,
 \end{equation} 
 where $C>0$ depends only on $\d$ and $\Lambda$. \end{lem}
\begin{proof}
The relation~\eqref{minog} is a consequence of~\eqref{14} and~\eqref{majok} follows directly.
\end{proof}

Obtaining a  lower bounds is a much more delicate task. For that we use an argument inspired by \cite{gz}.
We  have the following a priori lower bound, in which the logarithmic correction $\chi(\beta)$ (in dimension 2, for $\beta$ small) appears for the first time.
At this point we need to distinguish between the number of points a configuration has in a generic set $U$, that we will denote $\bar N$, and the number of points in the original problem, denoted $N$, which corresponds to $\mu(\R^\d)$ and also dictated  the blow-up lengthscale $N^{-\frac1\d}$.

\begin{prop}\label{minoK}
Assume $U$ is either $\R^\d$ or a  finite disjoint union of  hyperrectangles with parallel sides belonging to $ \mathcal Q_R$ for some $R \ge \max(1,\beta^{-\frac1\d})$ all included in $\Sigma$, or the complement of such a set. 
Let $\mu$ be a density such that $0<m\le \mu \le \Lambda $ in the set $\Sigma$ and satisfying~\eqref{asslog}. Assume $\mu(U)=\bar N$ is an integer. If $\d=2$ assume in addition (if $U$ is unbounded) that 
 \be \label{assinter}
\mu(\Sigma^c \cap U) \le C \frac{\bN}{\log N}\ee
and
\be \label{gmm}
\iint_{U^2} \g(x-y) d\mu(x) d\mu(y)  \ge - C \bN^2 \log N.\ee
There exists $C>0$, depending only on~$m,\d$ and $\Lambda$ and the constants in~\eqref{assinter} and~\eqref{gmm}, such that
\begin{equation}\label{minok}
 \log \K(U) \ge  - C\begin{cases}\beta  \chi(\beta)  \bN & \text{in} \ \d=2,\\
 \beta \bN + |\pa U| \xib & \text{in} \ \d\ge 3.
 \end{cases}\ee
\end{prop}
We note that~\eqref{gmm} is automatically satisfied by scaling with $C=\hal$ if $\mu$ is the blown-up of $ \mut$ by $N^{\frac{1}{\d}}$.

\begin{proof}{\bf Step 1: the case of the whole space}.
In the whole space with $\mu(U)=\mu(\R^\d)=N$ we have 
\be\label{reneman0}
\G(X_N, \R^\d)= \hal \iint_{\triangle^c} \g(x-y) d\(\sum_{i=1}^N \delta_{x_i}- \mu\)(x)  d\(\sum_{i=1}^N \delta_{x_i}- \mu\)(y).\ee
Starting from~\eqref{defK} we have 
$$
 \K(\R^\d) = N^{-N} \int_{(\R^\d)^N} \exp\(- \beta \G(\XN, \R^\d)  \) d  \mu^{\otimes N} (\XN).
$$
Using  Jensen's inequality as inspired by  \cite{gz}, we may then write
\begin{equation*}
\log  \K(\R^\d)\ge - \frac{\beta}{N^N} \int_{(\R^\d)^N}
  \G(\XN, \R^\d)   
 d\mu^{\otimes N} (X_N).
\end{equation*} 
We next insert the result of ~\eqref{reneman0}  to obtain
\begin{align*}
\lefteqn{\int_{(\R^{\d })^N} \G(\XN, \R^\d)   d \mu^{\otimes N}
} \qquad & 
\\ &
= 
\hal \int_{(\R^{\d })^N} \Bigg( \sum_{i\neq j} \g(x_i- x_j) - 2\sum_{i=1}^N\int_{\R^\d} \g(x_i- y) d\mu(y) 
\\ & \qquad\qquad\qquad\qquad\qquad\qquad
+\iint_{\R^{\d}\times \R^\d} \g(x-y) d\mu(x) d\mu(y)  \Bigg) d\mu^{\otimes N} (\XN)\\
&=\hal \(N(N-1) N^{N-2}- 2N N^{N-1} +N N^{N-1}\)   \iint_{(\R^\d)^2}\g(x-y) d \mu(x) d\mu(y)\\ & = -\hal  N^{N-1}  \iint_{(\R^\d)^2}\g(x-y) d \mu(x) d\mu(y).
\end{align*}
It follows that 
\be \log \K(\R^\d) \ge \frac{\beta}{2 N} \iint_{(\R^\d)^2}\g(x-y) d \mu(x) d\mu(y).
\ee
If $\d \ge 3$, $\g \ge 0$ hence  this yields $\log \K(\R^\d) \ge 0$, which implies the desired result.
If $\d=2$, this yields $\log \K(\R^2) \ge - \beta N \log N$, which is unsufficient  if $\beta $ is not very small.
We will improve this below. 

\smallskip

\noindent
{\bf Step 2: The case of a more general domain}. 
\\
{\bf Substep 2.1: setting up the Green function}.\\
Let $U$ be a general domain with piecewise $C^1$ boundary such that $\mu(U)=\bN$ an integer.
We note that the assumption on $U$ implies that $\pa U$ is a bounded set.

\smallskip

Denote $\hat U := \{x \in U \,:\, \dist (x, \pa U) \le 1\}$ and let $\bar \mu$ be defined   in $\hat U$ by
\be \label{barmuchoix}\bar \mu(x):= 
\begin{cases}
\mu(x) \exp\(- \beta  M \GG(x)\) &\text{if} \ \beta \le 1,\\ 
0  & \text{if} \ \beta > 1,
\end{cases}
\ee
where~$\GG$ is as in~\eqref{defGG} and~$M>0$ is a constant to be selected below. Below (in Substep~2.3) we will extend the definition of~$\bar \mu$ to the rest of~$U$ in such a way that it remains bounded, that~$\mu=\bar \mu$ on~$\{ x\in U\,:\, \dist(x,\partial U) > 2 \}$, and that~$\bar \mu(U)= \mu(U)=\bN$.

First we claim that we have
\begin{multline}\label{reneman}
\G(X_\bN,U)= \hal \iint_{\triangle^c} G_{U}(x,y) d\(\sum_{i=1}^\bN \delta_{x_i}- \mu\)(x)  d\(\sum_{i=1}^\bN \delta_{x_i}- \mu\)(y) \\
+\hal \sum_{i=1}^\bN H_U(x_i)+\sum_{i=1}^\bN \GG(x_i) \end{multline} where $G_U$ is the Neumann Green kernel of $U$, characterized as the solution of $$\left\{\begin{array}{ll}
-\Delta G_U(x,y)= \cd (\delta_y(x)-\frac{1}{\bar \mu(U)}\bar \mu) & \text{in} \ U,\\
  -\frac{\pa G_U}{\pa \nu}=0  & \text{on} \ \pa U,\end{array}\right.$$
  and 
  \be \label{defHU}
  H_U(x) := \lim_{y\to x} G_U(x,y) - \g(x-y).\ee
  We check that $G_U$ and thus $H_U$ exist and are well-defined up to an additive constants.
First, under our assumptions  we claim that 
 $v=\g*(\delta_y- \frac{1}{\bar \mu(U)}\bar \mu)$ is well-defined.
 Indeed, in dimension $\d\ge 3$ the convolution of $\g$ with $\bar \mu$ is well-defined (since $\bar \mu \in \cap_p L^p$) and is in $L^p$ by the Hardy-Littlewood Sobolev inequality. 
 In dimension $\d=2$ we need that $\int_{U} \g(x-z) d\bar \mu(z)<\infty$. If $U$ is bounded then this is immediate from the boundedness of~$\mu $ and~$\bar \mu$. If~$U$ is unbounded, since $\bar \mu$ and $\mu$ differ only near $\pa U$ which is bounded, it follows from~\eqref{asslog}.
  Secondly, we may  solve for $w=G_U-v$ which satisfies 
 $$\left\{\begin{array}{ll}
-\Delta w= 0 & \text{in} \ U\\
  \frac{\pa w}{\pa \nu}=- \frac{\pa v}{\pa \nu} & \text{on } \ \pa U\end{array}\right.$$
which can be done variationally since $\pa U$ is bounded.

We may now observe that $u:=\int_U G_U(x,y) df(y)$ is solution to 
\be \label{equf}
\left\{\begin{array}{ll}  - \Delta u=\cd\( f- \frac{\bar \mu}{\bar \mu(U)} \int_U f\)&\  \text{in} \ U\\ \frac{\pa u}{\pa \nu}=0& \text{on}  \ \pa U.\end{array}\right.\ee The function $u$ of~\eqref{defv} is thus equal to 
$\int_U G_U(x,y) d \( \sum_{i=1}^\bN \delta_{x_i}- \mu\) (y)$.
To obtain the claim ~\eqref{reneman} it then  suffices to integrate by parts from the formula~\eqref{minneum} similarly as in~\eqref{fnmeta}.

\smallskip

\noindent
{\bf Substep 2.2: lower bound.} \\
Starting from~\eqref{defK} we have in both cases $\beta \ge 1$ or $\beta \le 1$,
$$
 \K(U) \ge \bN^{-\bN} \int_{U^\bN} \exp\(- \beta \G(X_\bN, U) -\sum_{i=1}^\bN  \log \frac{\bar \mu}{\mu} (x_i)   \) d \bar \mu^{\otimes \bN} (X_\bN),
$$with the convention that $\bar \mu \log \bar \mu=0$ when $\bar \mu=0$.
Using  Jensen's inequality, we may then write
\begin{equation*}
\log  \K(U)\ge    \frac{1}{\bN^\bN} \int_{U^\bN} 
\( - \beta  \G(X_\bN, U) -\sum_{i=1}^\bN  \log \frac{\bar \mu}{\mu} (x_i)     \)
 d\bar \mu^{\otimes \bN} (X_\bN). 
 \end{equation*} 
We next insert the result of ~\eqref{reneman}  to find
\begin{align*}
&
\int_{U^\bN}\Bigg( \G(X_\bN, U)   + \frac1\beta \sum_{i=1}^\bN \log \frac{\bar \mu}{\mu}(x_i) \Bigg) d\bar \mu^{\otimes \bN}(X_\bN)
\\  &
= 
\hal \int_{U^\bN} \Bigg( \sum_{i\neq j} G_U(x_i, x_j) - 2\sum_{i=1}^\bN\int_U  G_U(x_i, y) d\mu(y) +\iint_{U^2} G_U d\mu d\mu  \Bigg) d\bar \mu^{\otimes \bN}(X_\bN)
\\ & \qquad
+ \int_{U^\bN} \Bigg( \hal \sum_{i=1}^\bN H_U(x_i)+\sum_{i=1}^\bN \GG(x_i) +\frac1\beta   \sum_{i=1}^\bN  \log \frac{\bar \mu}{\mu} (x_i) \Bigg)  d\bar \mu^{\otimes \bN}(X_\bN)\\
&
= 
\hal  \bN(\bN-1) \bN^{\bN-2} \iint_{U^2} G_U(x,y) d\bar \mu(x) d\bar\mu(y) - \bN^{\bN} \iint_{U^2} G_U(x,y) d\bar \mu(x) d \mu(y) 
\\ & \qquad 
+ \hal \bN^\bN \iint_{U^2} G_U(x,y) d\mu(x) d\mu(y)+\hal   \bN^{\bN}\int_U \(H_U(x)+ \frac{2}{\beta} \log \frac{\bar \mu}{\mu}(x) + 2\GG (x)\) d\bar \mu(x) 
\\ &
= \hal \bN^\bN \Bigg[ \iint_{U^2} G_U(x, y) d\( \mu-  \bar \mu\) (x)  d\( \mu-  \bar \mu\) (y) - \frac{1}{\bN } \iint_{U^2} G_U(x,y) d\bar \mu d\bar \mu
\\ & \qquad\qquad 
+  \int_U \( H_U +\frac2{\beta}\log \frac{\bar \mu}{\mu} + 2\GG\) 
 d\bar \mu\Bigg].
\end{align*}
It follows that 
\begin{align} 
\label{3term} 
\log  \K(U)
& \geq
-\beta \hal  \Bigg[\iint_{U^2} G_U(x, y) d\!\( \mu-  \bar \mu\) (x)  d\!\( \mu-  \bar \mu\) (y) - \bN^{-1} \iint_{U^2} G_U(x,y) d\bar \mu (x) d\bar \mu(y)
\\ & \qquad\qquad\qquad \notag 
+  \int_U \( H_U +\frac2{\beta}\log \frac{\bar \mu}{\mu} + 2\GG\) 
 d\bar \mu\Bigg].
 \end{align}

\noindent
{\bf Substep 2.3: discussion of the three terms and end of the definition of $\bar \mu$}.
We now give an upper bound for the three terms in the right-hand side. We observe that by definition \eqref{barmuchoix}
bounding  $\int  \( 2\log \frac{\bar \mu}{\mu} + 2 \beta \GG\) d\bar \mu$  involves evaluating
$\int  \beta \g(r) \exp\(- C \beta ( \g(r)-C)\) \, dr$, while bounding $\int |\mu -\bar \mu|$ involves evaluating
$\int |\exp \( - C \beta ( \g(r)-C)\) - 1| \, dr$, 
 where $r$ is the distance to $\partial U$.
 
With explicit computations using the expression for $\g$ and a change of variables  we observe  that
\be \label{intgeg}
\int_0^1 \beta \g(r) \exp\left( -\beta \g(r) \right) dr \le  C \begin{cases} \min (1, \beta^{\frac{1}{\d-2}} )  & \text{if} \ \d \ge 3, \\  \min(1, \beta) & \text{if} \ \d=2,\end{cases}\ee
and
\be \int_0^1 | \exp\left(-\beta  \g(r)  \right) -1|\, dr\le C \begin{cases}\min(1, \beta^{\frac{1}{\d-2}} ) & \text{if} \ \d \ge 3,
\\ 
\min(1,\beta) & \text{if} \ \d=2.\end{cases}\ee
Thus we find  \be
\label{thus}
\int_{\hat U}| \mu- \bar \mu |    + \int_{\hat U}  \(  \log \frac{\bar \mu}{\mu}+ \beta \GG(x) \) d\bar \mu\le C \mu(\hat U)   
\begin{cases} 
\min (1, \beta^{\frac{1}{\d-2}} )& \text{if} \ \d \ge 3,
\\ \min(1,\beta) & \text{if} \ \d=2.
\end{cases}
\ee
We next claim that we can distribute $\bar \mu-\mu$ in $\{x\in U, \dist(x, \pa U) \le 2\} \backslash \hat U$ so that  
 \be \label{bot1} \iint_{U^2} G_U(x,y) d( \mu-\bar \mu)(x) d(\mu-\bar \mu)(y)\le  C\mu(\hat U) \ee and
 \be  \label{bot2} \int_{U} \log \frac{\bar \mu}{\mu } d\bar \mu \le C  \mu(\hat U)   \begin{cases}\min(1,\beta^{\frac{1}{\d-2}} )& \text{if} \ \d \ge 3,\\ \min(1,\beta) & \text{if} \ \d=2.\end{cases}\ee
This will allow us to extend $\bar \mu-\mu $ by $0$ in $\{x \in U, \dist(x, \pa U) \ge 2\}$  in such a way that $\bar  \mu(U)=\mu(U)$ and $\bar \mu \le C $. 
This is accomplished by partitioning $\{ x\in U, \dist(x, \pa U) \le 2\}$ into disjoint cells $\mathcal C_i$ of bounded size, and then design $\bar \mu$ in $\mathcal C_i$ so that $\bar \mu$ remains bounded  by a constant depending only on $\d, m $ and $\Lambda$, and  $\int_{\mathcal C_i} \mu-\bar \mu=0$. We may then solve for 
 $-\Delta u_i= \mu-\bar \mu$ with zero Neumann data on the boundary of each cell $\mathcal C_i$.
 Letting $E=\sum_i \indic_{\mathcal C_i} \nab u_i$ we have that $-\div E= \mu -\bar \mu$ in $U$ and $E \cdot \nu= 0$ on $\pa U$. Then, by $L^2$ projection argument as in  Lemma~\ref{projlem} we find that 
 $$    \iint_{U^2} G_U(x,y) d( \mu-\bar \mu)(x) d(\mu-\bar \mu)(y)\le \int_U |E|^2 \le \sum_{i} \int_{\mathcal C_i} |\nab u_i|^2 ,$$
and this is bounded by a constant times the number of cells, which is proportional to $|\pa U|$, hence equivalently to $\mu (\hat U)$ since $\mu $ is bounded below in $\Sigma$ and $\pa U$ must be included in $\Sigma$ by assumption.   This proves~\eqref{bot1} and ~\eqref{bot2} is bounded by the number of cells times  the bound of~\eqref{intgeg}.

We then apply Proposition \ref{proappgreen}  of the appendix with the measure $\frac{\bar \mu}{\bar \mu(U)}$,  up to  adding a constant to $G_U$ (hence subtracting it from $H_U$, which has  a null total effect in the above right-hand side)
 we have
 \be\label{intGU}  \int_U G_U(x,y) dx=0\ee
and 
\be \label{ubh} H_U(x) \le  -  \bN^{-1} \int_U \g(x-z) d\bar\mu(z) + C \max  \(\g(\dist(x, \pa U)) ,1\) .\ee
 We then deduce that in view of ~\eqref{barmuchoix}, we have 
 $$ H_U  +\frac{2}{\beta}\log \frac{\bar\mu}{\mu}  + \GG \le    -  \bN^{-1} \int_U \g(x-z) d\bar\mu(z) +C\max \( \g(\dist(x, \pa U)) ,1\) - 2 M\mathsf{h} +\GG \quad \text{in} \ \hat U
.$$ In $U\backslash \hat U$,   since   $\dist (x, \pa U) \ge 1$, thanks to ~\eqref{ubh}, we have  instead 
$$H_U + \GG \le -  \bN^{-1} \int_U \g(x-z) d\bar\mu(z) +C .
$$
 Choosing $M$ so that $2M-1=C$ with $C$ the same constant as in~\eqref{ubh},  and using~\eqref{bot2}, we deduce that 
   \begin{multline}\label{finHU}\int_U \( H_U  +\frac{2}{\beta}\log \frac{\bar\mu}{\mu}+\GG \) d\bar \mu\\  \le  -  \bN^{-1} \iint_{U^2}\g(x-y) d\bar\mu(x)d \bar \mu(y) + C \bN + C \mu(\hat U)\begin{cases}\frac{1}{\beta}\min(1,\beta^{\frac{1}{\d-2}})  & \text{if} \ \d \ge 3,\\ \frac{1}{\beta} \min(1,\beta) & \text{if} \ \d=2.\end{cases} \end{multline}
In view of~\eqref{equf}, we have that 
 $\int_U G_U(x,y) d\bar \mu(y)=cst$, while $\iint_{U^2} G_U(x,y) d\bar \mu(y) dx= 0$  from~\eqref{intGU}, hence $cst=0$.
 It follows that 
   \be \label{330}
   \iint_{U^2} G_U d\bar \mu(x) d\bar \mu(y)=0.\ee
   Finally
   $$\bN^{-1}\iint_{U^2}  \g(x-y) d\bar \mu(x) d\bar\mu(y) \ge\begin{cases}  0 & \text{if } \d \ge 3,\\
-C \bN( \log  R+1) & \text{if $\d=2$, $U=Q_R$,}\\
- C \bN \log N & \text{otherwise}.\end{cases} $$ 
 Inserting this and~\eqref{330},~\eqref{finHU} and~\eqref{bot1}  into~\eqref{3term} and using~\eqref{gmm} we conclude that, for a constant~$C>0$ depending only on $\d$, $m$, and~$\Lambda$, 
\be \label{concjen}
\log \K(U) \ge   - C \mu(\hat U)   \begin{cases}\min(1,\beta^{\frac{1}{\d-2}} )& \text{if} \ \d \ge 3,\\ \min(1,\beta) & \text{if} \ \d=2 \end{cases} \ -  C\beta \begin{cases}  0 & \text{if } \d \ge 3,\\ \bN (\log  R +1)& \text{if $\d=2$, $U=Q_R$,}\\
 \bN \log N & \text{if $\d=2$, $U$ unbounded.}  \end{cases}
\ee
In the case~$\d \ge 3 $ this completes the proof. 

\smallskip

We next treat the case of a cube in $\d=2$ by a superadditivity argument.

\smallskip 

 \noindent
 {\bf Substep 2.4: the superadditivity argument.}
 Let us  now partition $U=Q_R$  into $p$ hyperrectangles in $\mathcal{Q}_{r}$ with  $r=  \max(1,\beta^{-\hal})$. Note that this scale is  roughly equal to $\rb$, the minimal lengthscale at temperature $\beta$, see~\eqref{minoR0}. For each hyperrectangle, we have from \eqref{concjen} a $\log \K$ bounded below by 
 $- C r^{\d-1} \min (1, \beta) - C \beta r^\d (1+\log r)$.

\smallskip
 
 Using~\eqref{superad2} and Stirling's formula (the $\log (N! N^{-N}) $ cancels with 
 $\sum_i \log (N_i ! N_i^{-N_i})$ up to order $\log N$) and since $p= O (\beta \bar N)$ we thus get  
  \begin{align*}
 \log \K(U)
 & \ge   - C p \log r^\d  - Cp + p\( - C\min (1,\beta) r^{\d-1}-C \beta  r^\d(1+ \log r) \)\\
 & \ge   \begin{cases} 
-  C \bN \beta (1+ |\log \beta |)  & \text{if } \ \beta \le 1\\
- C \beta \bN  & \text{if } \ \beta > 1. \end{cases}  \end{align*}
 In view of~\eqref{defchib}, we thus  conclude as desired  that 
 \be \label{logk2}
  \log \K(U)\ge  - C \beta \chi(\beta) \bN  .\ee
This completes the proof in the case $\d=2$ and $U$ is a cube. We can check that the same argument works as well for other nondegenerate Lipschitz cells.
 
 \noindent
 {\bf Step 3: the case of  general $U$.}  We  split $\Sigma \cap U $ (which is a set which a Lipschitz boundary) into nondegenerate cells $Q_i$ of size $\min (1, \beta^{-\hal})$ with $\mu(Q_i)$ integer.  The same superadditivity argument as in the last step provides the bound
 \be \label{logkint}
 \log \K( \Sigma\cap U) \ge - C \beta \chi(\beta) \mu(\Sigma \cap U) \ge - C \beta \chi(\beta) \bN.\ee
On the other hand we may insert ~\eqref{assinter} into~\eqref{concjen} to get 
 $\log \K(\Sigma^c\cap U) \ge - C \beta \bar N$. Another application of  the superadditivity~\eqref{superad2} relative to $\Sigma\cap U$ and $\Sigma^c \cap U$
 concludes the proof of~\eqref{minok}.
\end{proof}

Thanks to the a priori bounds~\eqref{majok} and~\eqref{minok}, we deduce a first control on the exponential moments of the energy.
(In the rest of the paper, we highlight when needed the dependence in $\beta $ of the partition functions, as a superscript.)
\begin{coro}\label{pttproba}
Assume $U$ and $\mu$ are as in Proposition \ref{minoK}, and  $\mu(U)=\bN$. 
There exists a constant $C>0$ depending only on $\d,m$, $\Lambda$ and the constants in~\eqref{assinter} and ~\eqref{gmm} such that 
\begin{equation}
\label{momex}
 \log \Esp_{\Q(U)} \(\exp\(\frac{\beta}{2} \G(X_\bN, U) \) \)\le   \log    \frac{ \K^{\beta/2} (U)}{ \K^{\beta} (U)}  
\le C\beta  \chi(\beta)\bN+ C |\pa U| \xib.\end{equation}
\end{coro}



\section{Comparison of Neumann and Dirichlet problems by screening}\label{sec4}

The screening procedure first introduced in \cite{ssgl} using ideas of \cite{aco}
 consists in taking a configuration $X_n$  in a set whose energy $\H$ or $\F$ is known  and modifying it near the boundary of the set to produce some configurations $Y_\mn$ with a corrected number of points for which the energy $\G$ is controlled by $\H(X_n)$ plus a small, well quantified, error. It has been improved over the years, and we here  provide for the first time a result with optimal errors. An informal description of  the method as well as the proof of the following main result, are postponed to Appendix~\ref{appa}.

 In the following result, two lengthscales $\l $ and $\tilde \l$ will appear,   $\tilde \l$ represents the distance over which one needs to look for a good contour by a mean value argument, then
$\l$ represents the distance needed to screen the configuration away from that good contour. The screening will only be possible if that distance is large enough compared to the boundary energy. In other words, only configurations with  well controlled boundary energy are ``screenable".

For any given configuration, the set $\Old$ (like ``old") represents the interior set in which the configuration and the associated field are left unchanged, while  in the complement  denoted $\New$ (like ``new"), the configuration is discarded and replaced by an arbitrary configuration with the correct number of points. Because we are dealing with statistical mechanics, we need not only to construct one screened configuration, but a whole family of them in order to retrieve a sufficient volume of configurations.  A new feature here is to sample the new points of the screened configuration according to a Coulomb Gibbs measure in the set $\New$ (this done in Proposition \ref{proscreen}). 

By abuse of notation we will also write $Q_{R+t}$ to denote the $t$-neighborhood of $Q_R$ if $t \ge 0$, and the set $\{x\in Q_R , \dist (x, \pa Q_R) \ge |t|\}$ if $t \le 0$.

We have to perform two variants of the screening: an ``outer screening" when $\Omega=Q_R$ and an ``inner screening" when $\Omega=U \backslash Q_R$. Both are entirely  parallel.
The main result is the following 
\begin{prop}[Screening]\label{outscreen} Assume $U$ is either $\R^\d$ or a  finite disjoint union of  hyperrectangles with parallel sides belonging to $ \mathcal Q_R$ for some $R \ge \max(1,\beta^{-\frac1\d})$ all included in $\Sigma$, or the complement of such a set. 
 Assume $\mu$ is a density satisfying $0 <m\le \mu \le  \Lambda$ in 
$\Omega= Q_{R}\cap U $ (outer case), resp. $\Omega= U\backslash Q_R$ (inner case) where $Q_R$ is a hyperrectangle of sidelengths in $[R,2R]$ with sides parallel to those of $U$, and such that 
 $\mu(\Omega)=\mn$, an integer.
There exists   $C>5$ depending only on $\d,m$ and $\Lambda$  such that the following holds.
Let $\ell$ and $\tilde \ell$ be such that  $R\ge \tilde \l \ge \l \ge C$ and in the inner case also assume $Q_R \cap U \subset\{ x\in U, \dist (x,\pa \Sigma \cap U) \ge \tilde \ell\}$.

 Let $X_n$ be a configuration of points in  $ \Omega$ and let  $u$ solve
 \be \label{eqsp}
\left\{\begin{array}{ll}
 -\Delta u= \cd \( \sum_{i=1}^n \delta_{x_i}-\mu\)  & \text{in} \  \Omega \\
 \frac{\partial u}{\partial \nu}= 0 \ & \text{on}\ \partial U \cap \Omega.\end{array}\right.\ee
We denote  if $\Omega= Q_R \cap U$
\be \label{definitionsannex0}
S(X_n) =  \int_{(Q_{R-\tilde \l}\backslash Q_{R-2\tilde \l})\cap U} |\nab u_{ \rrtt}|^2  \qquad S'(X_n)=\sup_x \int_{(Q_{R-\tilde \l}\backslash Q_{R-2\tilde \l}) \cap \car_{\tilde \l}(x)\cap U} |\nab u_{ \rrtt}|^2,\ee 
respectively if $\Omega= U \backslash Q_R$, 
\be \label{definitionsannex2}
S(X_n) =  \int_{(Q_{R+2\tilde \l}\backslash Q_{R+\tilde \l} )\cap U } |\nab u_{\rrtt}|^2  \qquad S'(X_n)=\sup_x \int_{\((Q_{R+2\tilde \l}\backslash Q_{R+\tilde \l}) \cap \car_{\tilde\l}(x)\)\cap U} |\nab u_{\rrtt}|^2,\ee where $\rrtt$ is defined as in~\eqref{rrtt}. 

 Assume the screenability condition 
\be \label{screenab}
 \l^{\d+1}>C\min \(S'(X_n), \frac{S(X_n)}{\tilde \ell}\).\ee
 
  There exists a $T \in  [\tilde \ell, 2 \tilde \ell]$, a set $\Old $ such that $Q_{R-T-1}\cap U \subset \Old \subset Q_{R-T+1}\cap U$ (resp. $   U\backslash Q_{R-T+1}\subset \Old \subset U\backslash Q_{R-T-1}$), a subset~$I_\pa\subseteq \{1,\ldots,n\}$
 and a positive  measure  $\tilde \mu$ in $\New:= \Omega \backslash \Old$ (all depending on $X_n$) such that the following holds:
\begin{itemize}
\item $\N$ being  the number of points of $X_n$ such that $B(x_i, \rrc_i)$ intersects $\Old$, we have 
\be \label{bornimp}
\tilde \mu(\New)= \mn-\N,\qquad |\mu(\New)-\tilde \mu(\New)|\le C\(R^{\d-1}+ \frac{S(X_n)}{\tilde \ell}\)
\ee
\be
\label{mmut2}  \|\mu -\tilde \mu\|_{L^\infty(\New)} \le \frac{m}{2},\qquad 
\int_{\New} (\tilde \mu-\mu)^2 \le C \frac{S(X_n)}{\l\tilde \ell}
\ee
\item we have $\# I_\pa\le C \frac{S(X_n)}{\tilde \ell}$
\item   for any configuration $Z_{\mn-\N}$ of $\mn -\N$ points  in $\New$, the configuration $Y_{\mn}$ in $\Omega$ made by the union of the points $x_i$ of $X_n$ such that $B(x_i, \rrc_i)$ intersects $\Old$  and the points $z_i$ of  $Z_{\mn - \N}$ satisfies
  \begin{multline} \label{nrjy}
\G(Y_\mn, \Omega)\le \H_U(X_n, \Omega)
\\+  C\( \frac{\l S(X_n)}{\tilde \l}  +   R^{\d-1} \tilde \ell+  \G(Z_{\mn-\N} ,\tilde\mu, \New)   + |\mn-n|
+
\sum_{(i,j) \in J} \g(x_i-z_j) \)
\end{multline} 
where the index set $J=J(X_n)$ in the sum is given by
\begin{equation}
J:=\left\{ (i,j) \in I_\partial \times \{1,\ldots, \mn-n_\Old\}  \,:\, |x_i-z_j|\le \rrc_i \right\}. 
\end{equation}
\end{itemize}
\end{prop}

Once this result is established one may tune the parameters $\l, \tilde \l$ to obtain the best results.
For instance, at the beginning we may only know that $\int_{Q_R} |\nab u_\rr|^2 $ is bounded by $O(R^\d)$, we then bound $S(X_n)$ and $S'(X_n)$ by $O(R^\d)$, optimize the right-hand side of~\eqref{nrjy} and choose $\ell \le\tilde \ell$ satisfying the constraints and obtain 
$$\G(Y_\mn,  \Omega) \le \H_U(X_{n} , \Omega) + C (R^{\d-\sigma}+  |n-\mn|),$$
for some $\sigma>0$, i.e. we get an error which is smaller than the order of the energy. The error $|n-\mn|$ can be controlled via the energy on a slightly larger domain, and shown to be negligible as well.

At the end of the bootstrap argument, we will know that the energy and points are well distributed down to say, scale $C$.
This means that we then know that (for good configurations) $S'(X_n)$ is controlled by $\tilde \l^\d$ and $S(X_n)$ by $R^{\d-1}\tilde \l$. The condition~\eqref{screenab} is then automatically satisfied and  we can  thus take $\l=C$, $\tilde \l=C$, and  we may also control $n-\mn$ by $O(R^{\d-1})$ to obtain a bound 
$$\G(Y_\mn, \Omega)\le \H_U(X_{n} , \Omega)+ C R^{\d-1}$$
i.e. with an error only proportional to the surface, the best one can hope to achieve by this approach.

The above  proposition is sufficient when studying energy minimizers, but when studying Gibbs measures, we actually need to show that given a set of configurations with well-controlled energy, we may screen them and sample new points in $\New $  to obtain a set with large enough volume in which~\eqref{nrjy} holds.  This is possible  and  yields comparison   of partition functions (reduced to screenable configurations) as stated in the following.

\begin{prop}\label{proscreen}With the same assumptions and notation as in the previous proposition. Assume in addition that 
$\tilde \ell \ge \beta^{-\hal}$ if $\d=2$. 
Let us define the set $\mathcal D_{s,z}$ to be 
\begin{equation}
\label{defA}
\mathcal D_{s,z}= \left\{X_n \in \Omega^n, S(X_n) \le s \ \text{and} \ S'(X_n) \le z
 \right\}\end{equation}
where $S, S'$ are as in~\eqref{definitionsannex0}, resp.~\eqref{definitionsannex2}. For any  number $s$  such that 
\be\label{constr2}
 \l^{\d+1}> C \min ( \frac{s}{\tilde \l} , z) ,\ee  
 and \be\label{constr3}s< c \tilde \l^2 R^{\d-1}\ee for some $c>0$ small enough (depending only on $\d, m, \Lambda$), 
 there exists $\alpha, \alpha'$ satisfying 
 \be \label{cara} \left|\frac{\alpha'}{\alpha}-1\right|\le C\( \frac{1}{\tilde \ell}+\frac{s}{\tilde \ell^2 R^{\d-1}}\), \qquad  \frac1C \tilde \ell R^{\d-1}\le \alpha \le C \tilde \ell R^{\d-1}\ee
such that letting  
 \be \ep_{e} := C\( \frac{s\l}{\tilde \ell} + R^{\d-1}\tilde \ell \chi(\beta) +|n-\mn|\) \ee 
 and 
 \be\label{epv}
 \ep_{v}:=  C\frac{s}{\ell \tilde \l} +\alpha-\alpha'+ (\mn-n-\alpha) \log \frac{\alpha}{\alpha'} - (\alpha+n-\mn+\hal) \log \(1+\frac{n-\mn}{\alpha}\)  + \hal \log \frac{n}{\mn} 
 \ee   we have 
\begin{equation}\label{resscreen}
n^{-n}  \int_{\mathcal D_{s,z}} \exp\(-\beta \H_U(X_n, \Omega) \)d\mu^{\otimes n} (X_n)
\le  C \K(\Omega) \exp\(  \beta \ep_e+ \ep_v \). \end{equation}
\end{prop}

Here the quantity $\ep_e$ corresponds to the energy error while $\ep_v$ corresponds to the volume error. We want the volume errors to be bounded by $O(\beta)$ times the volume, which is   more difficult to obtain when  $\beta$ is small.

\begin{proof}
For each $X_n\in \mathcal D_{s,z}$ with $s,z $ satisfying~\eqref{constr2}, the screening construction of Proposition~\ref{outscreen} can be applied, providing a number~$\N(X_n)$ and  a set~$\Old (X_n)$ (we emphasize here for a moment their dependence on~$X_n$). When screening we delete~$n-\N$ points in the configuration, those that fell outside of~$\Old$, there are~$\binom{n}{\N}$ ways of choosing the indices of the points that get deleted. In  terms of volume of configurations, this loses at most~$\mu(\New)^{n-\N}$ volume. In addition we glue each~$X_n|_{\Old}$  with~$\mn-\N$ points of~$Z_{\mn-\N}=(z_1, \dots , z_{\mn-\N})$, there are~$\binom{\mn}{\N}$ ways of choosing the indices for the gluing, resulting in configurations~$Y_\mn$ in~$\Omega^\n$ satisfying~\eqref{nrjy}. We integrate the choices of $(z_1, \dots , z_{\mn-\N})$ with respect to the measure $\mu$ restricted to $\New$.
We deduce that 
\begin{align}
\label{depart}
& \int_{\Omega^\mn}\exp\(-\beta \G(Y_\mn, \Omega) \) d\mu^{\otimes \mn}(Y_\mn)
\\ & \quad \notag
\geq
\int_{\mathcal D_{s,z}}\int_{ \New(X_n)^{\mn-\N}} \exp\bigg[-\beta \H_U(X_n , \Omega)- C\beta
\bigg( \frac{s\l}{\tilde \ell} + R^{\d-1}\tilde \ell+ \G(Z_{\mn-\N}, \tilde \mu(X_n), \New(X_n))
\\ & \qquad\qquad\qquad\qquad\qquad\qquad\qquad\qquad\notag
  + |\mn-n| 
+\sum_{(i,j) \in J}   \g(x_i-z_j)  \bigg) \bigg] 
\\ & \qquad \qquad\qquad\qquad\qquad\qquad\qquad\qquad\qquad\notag
\times \frac{\binom{\mn}{\N}}{\binom{n}{\N}}  \frac{1}{\mu(\New)^{n-\N} } \, d\mu|_{\New}^{\otimes (\mn-\N)} (Z_{\mn-\N})\, d\mu^{\otimes n}(X_n) .
\end{align}
Below we will show that, for each $X_n\in \mathcal D_{s,z}$, we have 
\begin{multline} \label{claim44}
\int_{\New^{\mn-\N}} 
\exp\bigg(-C\beta \bigg( \G(Z_{\mn-\N}, \tilde \mu,\New)  +\sum_{(i,j) \in J}   \g(x_i-z_j)\bigg) \bigg) d\mu|_{\New}^{\otimes (\mn-\N)} (Z_{\mn-\N}) \\ 
\geq 
(\mn- \N)^{\mn-\N}
 \exp\(    \int_{\New} \tilde \mu \log \frac{\mu}{\tilde \mu}- C \beta \chi(\beta)  R^{\d-1} \tilde \ell  \).
\end{multline}
Before giving the proof of~\eqref{claim44}, we use it to obtain the proposition. Thanks to~\eqref{mmut2},~\eqref{constr3} and~\eqref{bornimp} we have $|\frac{\mu}{\tilde \mu} -1|< \hal$  if $c$ is chosen small enough and thus by Taylor expansion
\be\label{413}\int_{\New} \tilde \mu \log \frac{\mu} {\tilde \mu} = \int_{\New} \mu - \tilde \mu + O\( \int_{\New} \frac{|  \mu- \tilde \mu|^2}{\tilde \mu} \)  = \mu(\New)- \tilde \mu(\New) + O \( \frac{s}{\ell \tilde \l}\).   \ee
By Stirling's formula, 
\begin{align}
\label{414}
& \log  \( \frac{ \mn! (n-\N)!    }{ n! (\mn-\N)!}  \frac{(\mn-\N)^{\mn-\N}}{\mu(\New)^{n-\N} }\) 
\\ & \qquad \notag
\geq \mn \log \mn -n \log n  + (n-\N) \log \frac{n-\N}{\mu(\New)}   + \hal \log \frac{\mn  (n-\N)}{n  (\mn-\N)} - C. 
\end{align}
Combining~\eqref{claim44}--\eqref{414} and inserting into~\eqref{depart}, we obtain, for a constant~$C$ depending only on~$\d, m$ and $\Lambda$,
\begin{align*}
\lefteqn{
\int_{\Omega^\mn}\exp\(-\beta \G(Y_\mn, \Omega) \) d\mu^{\otimes \mn}(Y_{\mn})
} \quad & 
\\ & 
\geq 
\exp\( -C \beta\( \frac{s\l}{\tilde \ell} +  R^{\d-1} \tilde \l \chi(\beta)+ |\mn-n|\) -C \frac{s}{\ell \tilde \l} \)
\\ & \qquad \times 
\int_{\mathcal D_{s,z}} 
\bigg[ \exp\(-\beta \H_U(X_n, \Omega)  + \mn \log \mn -n \log n  +  \mu(\New)- \tilde \mu(\New)  \)
\\ & \qquad\qquad\qquad \times 
\exp\( (n-\N) \log \frac{n-\N}{\mu(\New)} + \hal \log \frac{\mn  (n-\N)}{n  (\mn-\N)}-C \)\bigg] d\mu^{\otimes n}( X_n).
\end{align*}
We may next use a mean-value argument to obtain, for some  configuration $X_n^0\in \Omega^n$,
\begin{align*}
\lefteqn{
\int_{\Omega^\mn}\exp\(-\beta \G(Y_\mn, \Omega) \) d\mu^{\otimes \mn}(Y_{\mn}) 
} \quad & 
\\ &
\geq 
\exp\bigg[ \mn \log \mn -n \log n +\mu(\New(X_n^0))-\tilde \mu(\New (X_n^0))+  (n-\N(X_n^0)) \log \frac{n-\N(X_n^0)}{\mu(\New(X_n^0))}   
\\ & \qquad\qquad 
+ \hal \log \frac{\mn  (n-\N(X_n^0))}{n  (\mn-\N(X_n^0))}-C 
  -C \beta\( \frac{s\l}{\tilde \ell} +  R^{\d-1} \tilde \l \chi(\beta)+ |\mn-n|\)   -  \frac{Cs}{\ell \tilde \l}     \bigg]  
\\ & \qquad \times \notag
\int_{\mathcal D_{s,z}} \exp\(-\beta \H_U(X_n, \Omega) \) d\mu^{\otimes n} (X_n).
\end{align*}
Letting then $\alpha = \tilde \mu(\New (X_n^0))$ and $\alpha'= \mu(\New (X_n^0))$, we have in view of~\eqref{bornimp} that~\eqref{cara} holds and we may rewrite the second exponential term as 
 $$\exp\(\mn \log \mn -n \log n  + \alpha'-\alpha  +  (n-\mn + \alpha) \log \frac{n- \mn+\alpha}{\alpha'}   + \hal \log \frac{\mn  (n-\mn+\alpha)}{n  \alpha}\).$$
Rearranging terms we obtain the proposition. 

\smallskip 
 
It remains to prove~\eqref{claim44}.
Applying Jensen's inequality, we find
\begin{multline*}
\int_{\New^{\mn-\N}} 
\exp\bigg(-C \beta \bigg(\G(Z_{\mn-\N}, \tilde \mu,\New)   +\sum_{(i,j) \in J}   \g(x_i-z_j)\bigg) \bigg) \, d\mu|_{\New}^{\otimes (\mn-\N)} (Z_{\mn-\N}) 
\\
= \!\int_{\New^{\mn-\N}} \!\exp\bigg[\!\!-\!C\beta \bigg( \G(Z_{\mn-\N}, \tilde \mu,\New) \!+\!\sum_{(i,j) \in J}   \g(x_i-z_j) \bigg)+ \!\sum_{i=1}^{\mn-\N} \!\log \frac{\mu}{\tilde \mu} (z_i)  \bigg] d\tilde\mu^{\otimes (\mn-\N)} (Z_{\mn-\N}) 
\\
\geq 
\tilde \mu(\New)^{\mn-\N} 
\exp\Bigg[  \tilde \mu(\New)^{\N-\mn} 
\int_{\New^{\mn-\N} } \bigg( -C\beta\bigg( \G(Z_{\mn-\N}, \tilde \mu,\New)
 +\sum_{(i,j) \in J}   \g(x_i-z_j)\bigg)
\\+ \sum_{i=1}^{\mn-\N} \log \frac{\mu}{\tilde \mu} (z_i)  \bigg) \,
d\tilde \mu^{\otimes (\mn-\N)} (Z_{\mn-\N}) \Bigg]
\end{multline*}where we recall  that $\tilde \mu(\New)=\mn-\N$.
We then use the same proof as that of Proposition \ref{minoK}. The term $\sum_{(i,j) \in J}   \g(x_i-z_j)$ adds a contribution 
\begin{equation*}- C \beta  (\mn - \N)^{\mn- \N}  \sum_{i\in I_\pa}\int_{|z-x_i|\le \rrc_i}       \g(x_i-z)  d\tilde \mu(z)
    \ge - C \beta  (\mn - \N)^{\mn- \N}   \# I_{\pa}
\end{equation*}
and, by $\#I_\pa \le Cs/\tilde \ell$ and~\eqref{constr3}, we conclude that 
\begin{align*}
&
\int_{\New^{\mn-\N}} \exp\bigg(-C \beta \bigg(\G(Z_{\mn-\N}, \tilde \mu,\New)   +\sum_{(i,j) \in J}   \g(x_i-z_j)\bigg) \bigg) d\mu|_{\New}^{\otimes (\mn-\N)} (Z_{\mn-\N})
\\ & \qquad 
\ge   (\mn-\N)^{\mn-\N} \exp\bigg(    \int_{\New} \tilde \mu \log \frac{\mu}{\tilde \mu} - C \beta R^{\d-1}\tilde \l (1+(\log R )\indic_{\d=2}) \bigg).
\end{align*}
In the case $\d=2$, in view of the fact that~$\tilde \ell \ge \beta^{-\hal}$, we see from its construction (in Appendix~\ref{appa}) that~$\New $ can be partitioned into disjoint nondegenerate cells of size $\max(1,\beta^{-\hal})$ in which $\tilde\mu$ integrates to an integer. Using superadditivity as in the proof of Proposition \ref{minoK}, 
we conclude that~\eqref{claim44} holds.

\end{proof}

\begin{coro}\label{coro43}With the same assumptions and notation as in the previous proposition,
there exists $C>0$ depending only on $\d, m, \Lambda$ such that the following holds. Let 
$$\mathcal B_n=\left\{ X_n \in \Omega^n, \sup_x \int_{  \{ 
(\pa \Omega)_{-2\tilde\l}\cap \car_L(x)}|\nab u_{\rrtt}|^2\le \chi(\beta)ML^\d\right\}$$
where $(\pa \Omega)_{-2\tilde\l}$ denotes $Q_{R-\tilde \l}\backslash Q_{R-2\tilde \l}\cap U $ if $\Omega =Q_R\cap U$ and $Q_{R+2\tilde \l}\setminus Q_{R+\tilde \l}\cap U $ if $\Omega= U\setminus Q_R$.
If \be\label{condsurL} R>L > CM \max(1, \beta^{-\hal}\indic_{\d=2}) ,\ee  and
$\dist(Q_R, \pa\Sigma \cap U)\ge L$, we have
\begin{multline}\label{eqsulk}
n^{-n}  \int_{\mathcal B_n} \exp\(-\beta \H_U(X_n, \Omega) \)d\mu^{\otimes n} (X_n) 
\\ \le  C \K(\Omega) \exp\Big(  \beta \(  C R^{\d-1}L  \chi(\beta)M +|n-\mn|   \)   +  \frac{C M \chi(\beta) R^{\d-1}}{  L}    \\
+ \alpha-\alpha'+ (\mn-n-\alpha) \log \frac{\alpha}{\alpha'} - (\alpha+n-\mn+\hal) \log \(1+\frac{n-\mn}{\alpha}\)  + \hal \log \frac{n}{\mn} 
\Big),\end{multline}
with $\alpha, \alpha'$ satisfying 
 $$  \left|\frac{\alpha'}{\alpha}-1\right|\le C \frac{   \chi(\beta) }{L } , \qquad  \frac1C L  R^{\d-1}\le \alpha \le C L R^{\d-1}.$$
 \end{coro} 
\begin{proof} 
If $X_n$ in $\mathcal B_n$ then 
 $$S (X_n) \le \frac{R^{\d-1}}{L^{\d-1}} M  \chi(\beta) L^\d, \quad S'(X_n) \le M  \chi(\beta) L^\d .$$
 using the definition~\eqref{definitionsannex0} or~\eqref{definitionsannex2}.
  We check that 
   setting  $\ell=\tilde \ell = L$ and $s= M \frac{R^{\d-1}}{L^{\d-1}}  \chi(\beta)L^\d$ and $z= M \chi(\beta) L^\d$ we  have  that if~\eqref{condsurL} holds, then up to making the constant larger in~\eqref{condsurL},  ~\eqref{constr2} and~\eqref{constr3} hold  and  the result follows by applying the result of Proposition~\ref{proscreen}.
\end{proof}

\begin{remark} \label{remerr}
When summing the contributions over $\Omega$ where $n$ points fall and $U\backslash \Omega $ where $N-n$ points fall, the errors of~\eqref{epv} compensate and add up to a well bounded error. More precisely, 
if  $\alpha, \alpha'$, respectively $\gamma, \gamma'$ satisfy ~\eqref{cara} then for every $n$ we have
\begin{multline}\label{425}
\alpha-\alpha'+ (\mn-n-\alpha) \log \frac{\alpha}{\alpha'} - (\alpha+n-\mn+\hal) \log \(1+\frac{n-\mn}{\alpha}\)  + \hal \log \frac{n}{\mn} 
\\+ \gamma'-\gamma+ (n-\mn-\gamma) \log \frac{\gamma}{\gamma'} - (\gamma+\mn-n+\hal) \log \(1+\frac{\mn-n}{\gamma}\)  + \hal \log \frac{N-n}{N-\mn} \\ \le  C\( \frac{R^{\d-1}}{\tilde \ell} + \frac{s^2}{\tilde \ell^3 R^{\d-1}}\)\end{multline}
\end{remark}
\begin{proof}
First we notice that since the expressions arising here  originate in Stirling's formula,  they can be restricted to the case of $\alpha+ n-\mn \ge 1$, $\gamma+\mn-n\ge 1$, $n\ge 1$ and $N-n\ge 1$ (all the quantities involved are integers). 

We then study the expression in the left-hand side of~\eqref{425} as a function of the real variable $n$ (with the above constraints).  Differentiating in $n$, we find that it achieves its maximum when  
\begin{multline*} \log \frac{\gamma \alpha'}{\gamma'\alpha}- \log \( 1+\frac{n-\mn}{\alpha}\)  + \frac{1}{2(\alpha+ n-\mn)} + \log \( 1+ \frac{\mn-n}{\gamma}\)  -\frac{1}{2(\gamma+\mn-n)}\\+ \frac{1}{2n}-\frac{1}{2(N-n)} 
  =0. \end{multline*}
Using $\alpha + n-\mn\ge 1$, $\gamma+\mn-n\ge 1$, $n \ge 1$, $N-n\ge 1$ and~\eqref{cara} we deduce that 
$$\left| \log \( 1+ \frac{\mn-n}{\gamma}\)- \log \( 1+\frac{n-\mn}{\alpha}\)\right|\le C$$
and thus 
$$ \frac{  1+ \frac{\mn-n}{\gamma}}{1+\frac{n-\mn}{\alpha}}   \ \text{is bounded above and below}$$
and it follows easily in view of~\eqref{cara} that $|n-\mn|\le C  \tilde \ell R^{\d-1}$. 
To find the maximum of~\eqref{425} it thus suffices to maximize it for such $n$'s. But for such $n$'s we may check that 
$\hal \log \( 1+\frac{n-\mn}{\alpha}\) $, $\hal \log \(1+ \frac{\mn-n}{\gamma}\)$, $\log \frac{n}{\mn} $ and $\log \frac{N-n}{N-\mn}$
 are all bounded by a constant depending only on $\d, m, \Lambda$, hence it suffices to obtain a bound for 
 \begin{multline}\label{426}
\alpha-\alpha'+ (\mn-n-\alpha) \log \frac{\alpha}{\alpha'} - (\alpha+n-\mn) \log \(1+\frac{n-\mn}{\alpha}\)  
\\+ \gamma'-\gamma+ (n-\mn-\gamma) \log \frac{\gamma}{\gamma'} - (\gamma+\mn-n) \log \(1+\frac{\mn-n}{\gamma}\)  \end{multline}
 Differentiating in $n$, we find that this expression is maximal exactly for 
 $$  1+\frac{\mn-n}{\gamma}  =\frac{\gamma'\alpha}{\gamma\alpha'}  \(1+\frac{ n-\mn}{\alpha}\)  \Leftrightarrow n=\mn+ \frac{\frac{\gamma}{\gamma'}- \frac{\alpha}{\alpha'}} {\frac{1}{\gamma'}+\frac{1}{\alpha'} }$$
 Inserting this into~\eqref{426} we find  that the expression is then equal to 
 \begin{multline*}\alpha-\alpha' - \alpha \log \frac{\alpha}{\alpha'}-\alpha \log  \( 1+\frac{n-\mn}{\alpha}\)-\gamma \log    \( 1+  \frac{\mn -n }{\gamma}\) +(\mn-n) \log \frac{\gamma'}{\gamma}  \\= O\( \frac{R^{\d-1}}{\tilde \ell} + \frac{s^2}{\tilde \ell^3  R^{\d-1}}\)  \end{multline*}
 where we used a Taylor expansion and ~\eqref{cara}.
\end{proof}

The next goal is to select $s, \l,\tilde \l$ to optimize the errors made in Proposition \ref{proscreen}. This way
we  obtain the main result of this section, which shows that if one has good energy controls at some scale, one can deduce  some  control at slightly smaller scales.

In all the rest of the paper, we will denote the event that $X_\bN$ has $n$ points in $\Omega$ by
 \be 
\label{Bndef} \mathcal A_n:=\{ X_\bN \in U^\bN, \# (\{X_\bN\} \cap \Omega) = n\}.
\ee

\begin{prop}\label{proboot}
Assume $U$ is either $\R^\d$  or a  finite disjoint union of disjoint hyperrectangles all included in $\Sigma$ with parallel sides belonging  to $ \mathcal Q_\rho$ for some $\rho \ge \max(1,\beta^{-\frac12})$, or the complement of such a set. 
Let $\mu$ be a density such that $0<m\le \mu \le  \Lambda$  in the set $\Sigma$ and $\mu(U)=\bN$ is an integer.
Let $C_0=  \frac{2C}{4\cd}$ for the constant $C$ of 
~\eqref{14}. 
  
There exists a   constant  $C>0$  depending only on $\d, m$ and  $\Lambda$ such that the following holds.
Assume  that $Q_R$ is a hyperrectangle of sidelengths in $[R,2R]$ with sides parallel to those of $U$, that $\mu (Q_R \cap U)=\mn$ and $Q_R \cap U \subset \Sigma$.
  Assume that  there exists a cube 
$\car_L$ of size $L$ such that 
  \begin{equation}\label{locallaw}
  \left|\log \Esp_{\Q (U)} \( \exp \(\frac{\beta}{2} \( \G^{ \car_L} (\cdot , U)+ C_0 \#( \{X_\bN \}\cap \car_L)  \)    \) \)\right|\le \mathcal C \beta \chi(\beta) L^\d\end{equation} with $\mathcal C>1$, and such that 
    $\car_L$ contains  $ Q_{R+\tilde \l} \cap U$ 
 with $$   L\ge R \ge \hal L,$$  
\be \label{condRL}
    R>C '\max(1, \beta^{-\hal} \chi(\beta)^{\ed{\frac12}}) 
\ee and
 \be\label{deftl} \tilde \l=   
  \ed{ C''  \max  \( \(\frac{R }{ \max(1, \beta^{-\hal} \chi(\beta)^{\frac12})}\)^{-\frac23}   R,   R^{1-\frac{2}{\d+2}}\) },\ee 
\ed{for  some $C', C''$ large enough},  both depending only on $\d$, $m,\Lambda$  and $\mathcal C$.
 Assume in addition that    \be\label{assumpdist}
 \dist (Q_R\cap U , \partial \Sigma \cap U) \ge \tilde \l. \ee
Then there exists a sequence $\gamma_n$ satisfying
\be\label{boundgamma}
\sum_{n=0}^\bN  
\gamma_n\le \exp\(- \mathcal C \beta \chi(\beta)R^\d\).
\ee such that we have
\begin{multline}
\label{resLK}
\Esp_{\Q(U)} \(  \exp \(  \frac{\beta}{2} \G^{Q_{R-2\tilde \ell}}(X_\bN, U)  \indic_{\mathcal A_n}\) \)
\\ \le
\gamma_n + \frac{\K^{\frac\beta2} (Q_R) }{\K^\beta (Q_R)} \exp\(
\beta\( \frac{ \mathcal C }{ 4} \chi(\beta) R^\d  
 +|n-\mn| + \frac{C_0}{2}n\)  \).
\end{multline}
\end{prop}
Once one has obtained local laws down to the minimal scale $\rb$,  Corollary \ref{coro43} will allow to improve the error term and bound it by $R^{\d-1}$.

\begin{proof}
{\bf Step 1: the case of excess energy}.
Recalling the definition of~$\mathcal{A}_n$ in~\eqref{Bndef}, and letting $S$ be as in~\eqref{definitionsannex0} and $M>0$ be a constant to be determined below, we define 
\begin{equation*}
\mathcal B_n := \left\{X_\bN \in \mathcal{A}_n, S(X_\bN|_{\Omega^c}) \le M \mathcal{C} \chi(\beta) L^\d  , \ S(X_\bN|_{\Omega}) \le  M \mathcal{C} \chi(\beta) L^\d 
 \right\}.
\end{equation*}
We also define 
\begin{align*}
\mathcal B_n^+ 
& := \left\{X_{\bN-n} \in (U\backslash \Omega)^{\bN-n}, S(X_{\bN-n}) \le M \mathcal{C} \chi(\beta) L^\d \right\}, \\
\mathcal B_n^- 
& := \left\{ X_n \in \Omega^n, S(X_n) \le  M \mathcal{C} \chi(\beta) L^\d 
 \right\}.
 \end{align*}
 It is clear that if $X_\bN \in \mathcal B_n$ then $X_\bN|_{\Omega^c} \in \mathcal B_n^+$ and $X_\bN|_{\Omega} \in \mathcal B_n^-$. Also, if $X_\bN \in \mathcal B_n^c$ then, in view of~\eqref{14} and the definition of $S$, we have
 $$\G^{\car_L \backslash \omc} (X_\bN, U) + C_0 \#(\{X_\bN\}\cap \car_L) \ge\frac{M \mathcal{C}  \chi(\beta)L^\d} {C}$$ hence 
\begin{align*}
 & 
\Esp_{\Q (U)} \( \exp \(\frac{\beta}{2} \( \G^{ \car_L} (\cdot , U)+ C_0 \# (\{X_\bN\} \cap \car_L)  \) \indic_{\mathcal B_n^c}   \) \)
\\ & \qquad 
\ge \  \exp\(\frac{\beta}{2} \frac{M \mathcal{C}  \chi(\beta)L^\d} {C}\) \Esp_{\Q (U)} \( \exp \(\frac{\beta}{2}  \G^{ \omc} (\cdot , U)     \) \indic_{\mathcal B_n^c}\).
\end{align*}
It follows  that 
\be \label{controlbad}
\Esp_{\Q (U)} \( \exp \(\frac{\beta}{2}  \G^{ \omc} (\cdot , U)     \) \indic_{\mathcal B_n^c}\)\le  \gamma_n  ,
\ee
with 
$\sum_{n=0}^N
\gamma_n \le   \exp\(- \mathcal C \beta \chi(\beta)R^\d\)$
in view of ~\eqref{locallaw}, provided~$M$ is chosen large enough, depending only on $\d, m, \Lambda$.
We henceforth fix~$M$.

\smallskip

\noindent
{\bf Step 2: the case of good energy bounds.}\\
We now wish to estimate the same quantity in the event $\mathcal B_n$. 
\ed{We wish to choose   $\tilde \l=\tilde \eta R$ with $\tilde \eta <\frac14 $} to be determined later, and \ed{$\ell=\eta\tilde \ell$ with $0\le \eta\le 1$, satisfying}  
\be\label{defl}
 \ed{ \l^{\d+1} \ge  \frac{CM\mathcal C \chi(\beta) L^\d }{ \tilde \l}  }
  \ee with $C$ as in~\eqref{constr2}.
This way, choosing $s= M\mathcal C \chi(\beta) L^\d$, the screenability condition~\eqref{constr2} will be verified.
To apply Proposition \ref{proscreen}, \ed{we also need 
\be  \label{condtl}\max \(\beta^{-\hal} \indic_{\d=2},1\) \le \tilde \eta R
\ee }and 
\be \label{435b} 
CM \mathcal C \chi(\beta)  L^\d<\tilde \ell^2 R^{\d-1}.\ee
Using~\eqref{add3},~\eqref{restri1} and~\eqref{14}, we have 
\begin{align*}
  \frac{\beta}{2}  \G^{ \omc} (X_{\bN} , U)   -\beta \G(X_{\bN}, U) & \le 
 \frac{\beta}{2}\G^{\omc} (X_{\bN}, U)   -\beta \G^{\Omega}(X_{\bN}, U) 
 - \beta\G^{U\backslash \Omega} (X_{\bN}, U)\\
  & \le  \frac{\beta}{2} \G^{\omc} (X_{\bN} , U) -\frac{\beta}{2} \G^{\Omega}(X_{\bN}, U)  -\frac{\beta}{2} \G^{\Omega}(X_{\bN},U) - \beta\G^{U\backslash \Omega} (X_{\bN}, U)
  \\  & \le  -\frac{\beta}{2} \G^{\Omega \backslash  \omc} (X_{\bN}, U)    -\frac{\beta}{2} \G^{\Omega}(X_{\bN},U) - \beta\G^{U\backslash \Omega} (X_{\bN}, U)
  \\
& \le - \frac{\beta}{2} \G^{\Omega}(X_{\bN},U)  - \beta \G^{U\backslash \Omega} (X_{\bN}, U)  + \frac{\beta}{2} C_0 n
\\ & \le  - \frac{\beta}{2} \H_U(X_{\bN}|_{\Omega} , \Omega) - \beta \H_U (X_{\bN}|_{U\backslash \Omega} , U\backslash \Omega) + \frac{\beta}{2}C_0n \end{align*}
Thus,
\begin{align*}
\lefteqn{
\Esp_{ \Q(U)}  \( \exp \(\frac{\beta}{2}  \G^{ \omc} (\cdot , U)     \) \indic_{\mathcal B_n}\)
} \quad & 
\\ & 
= \frac{1}{\bN^\bN \K(U)} \int_{\mathcal B_n} \exp\(  \frac{\beta}{2} \G^{\omc}(X_\bN,U) -\beta \G (X_\bN,U) \) d\mu^{\otimes \bN} 
\\ & 
\le  \frac{1}{\bN^\bN  \K(U)}   \frac{\bN!}{n!(\bN-n)!}  \int_{\Omega^n \cap \mathcal B_n^-} \exp\( - \frac{\beta}{2} \H_U(\cdot , \Omega)    + \frac{\beta}{2} C_0 n\) d\mu^{\otimes n}
\\ & \qquad 
\times \int_{(U\backslash \Omega)^{\bN-n}\cap \mathcal B_n^+ } \exp\( - \beta \H_U(\cdot , U\backslash \Omega) \)d\mu^{\otimes (\bN-n)}.\end{align*}
Inserting~\eqref{resscreen} applied in $\Omega $ (with $\beta/2$  instead of $\beta$) and in $U\backslash \Omega$ and using Remark \ref{remerr}, we deduce that 
\begin{align*}
\lefteqn{
\Esp_{\Q(U)}  \( \exp \(\frac{\beta}{2}  \G^{ \omc} (\cdot , U)     \) \indic_{\mathcal B_n}\)
} \quad & 
\\ & 
\leq
  \frac{1}{\bN^\bN  \K(U)}   \frac{\bN!}{n!(\bN-n)!}    C n^n (\bN-n)^{\bN-n} \K^{\frac{\beta}{2}}(\Omega)  \K^{\beta}(U\backslash \Omega) \exp\(  \beta \ep_e+ \ep_v + \frac{\beta}{2} C_0 n\)
  \end{align*} 
with  \be \ep_{e} := C\(\ell \frac{ M \mathcal C \chi(\beta) L^\d }{\tilde \ell}+ R^{\d-1} \tilde \ell \chi(\beta)+ |n-\mn|\) 
\ee and 
\be \ep_v := C\(\frac{ M \mathcal C \chi(\beta) L^\d  }{\ell \tilde \ell}+ \frac{R^{\d-1}}{\tilde \ell}+ \frac{ ( M \mathcal C \chi(\beta) L^\d )^2}{\tilde \ell^3 R^{\d-1}}\),\ee where we used the choice $s:=M\mathcal C \chi(\beta) L^\d$.
   We may also bound from below  $\K(U)$  using ~\eqref{superad2} applied with $\Omega $ and $U\backslash \Omega$, which yields
   $$ \frac{\bN! \mn^{\mn} (\bN-\mn)^{N-\mn} }{\bN^\bN  \K(U) \mn! (\bN-\mn)!}  \K^{\frac\beta2}(\Omega)\K^{\beta} (U \backslash\Omega) 
   \le  \frac{\K^{\frac \beta2}(\Omega)}{\K^{\beta} (\Omega)}.$$ 
   Inserting into the above, we obtain that 
\begin{align*}
\lefteqn{
\Esp_{\Q(U)}  \( \exp \(\frac{\beta}{2}  \G^{ \omc} (\cdot , U)     \) \indic_{\mathcal B_n}\)
} \quad & 
\\ & 
\leq
  C  \frac{ \mn! (\bN-\mn)!    n^n (\bN-n)^{\bN-n} }{n!(\bN-n)!   \mn^\mn (\bN-\mn)^{\bN-\mn} }      \frac{\K^{\frac \beta2}(\Omega)}{\K^{\beta} (\Omega)} \exp\(  \beta \ep_e+ \ep_v + \frac{\beta}{2} C_0 n\)
\end{align*} 
By Stirling's formula, for every $n \le \bN$, we have 
$$
\frac{ \mn! (\bN-\mn)!    n^n (\bN-n)^{\bN-n} }{n!(\bN-n)!   \mn^\mn (\bN-\mn)^{N-\mn} } \sim \sqrt{ \frac{ \mn (\bN-\mn)}{n(\bN-n)}} \le C.
$$
We may therefore absorb the log of this quantity into $\ep_v$, and conclude that 
    \begin{equation}\label{438}
\Esp_{\Q(U)}  \( \exp \(\frac{\beta}{2}  \G^{ \omc} (\cdot , U)     \) \indic_{\mathcal B_n}\)
\\ 
\leq C   \frac{\K^{\frac \beta2}(\Omega)}{\K^{\beta} (\Omega)} \exp\(  \beta \ep_e+ \ep_v + \frac{\beta}{2} C_0 n\).
\end{equation} 
   
We now search for the smallest~$\tilde \l$ such that the terms of~$\beta \ep_e+\ep_v$ (except those involving~$n$ and $\mn$) are $\le  \beta \ed{\chi(\beta)} \frac{\mathcal C}{4}R^\d$ 
that is 
 \begin{align*} 
\left\{   
\begin{aligned}
& C\frac{ M\mathcal C \chi(\beta) L^\d \ell}{\tilde \ell} \le \frac{\mathcal C }{4}\ed{\chi(\beta)}  R^\d
\\&
 CR^{\d-1}\tilde \ell \chi(\beta) \le \frac{\mathcal C }{4} \ed{\chi(\beta)} R^\d\\
 &  C \frac{M\mathcal C \chi(\beta) L^\d}{\ell \tilde \ell} \le \frac{\mathcal C }{4}\beta  \ed{\chi(\beta)} R^\d
 \\
& C \frac{R^{\d-1}}{\tilde \ell}\le \frac{\mathcal C }{4} \beta \ed{\chi(\beta)} R^\d
 \\ &C \frac{(M \mathcal C \chi(\beta) L^\d )^2}{\tilde \ell^3 R^{\d-1}} \le \frac{\mathcal C }{4}\beta \ed{\chi(\beta)} R^\d
\end{aligned}\right.
\end{align*}
and also~\ed{\eqref{defl},}~\eqref{condtl},~\eqref{435b} are satisfied.
 \ed{With our choice  $R \le L \le 2R$ and $\ell =\eta \tilde \ell$, $\tilde \ell=\tilde \eta R$, after direct computations we find that this reduces to the conditions:
  \begin{align*} 
\left\{   
\begin{aligned}
&  
C \eta M \le \frac14
\\ &
C \tilde \eta \le \frac{\mathcal C}{4} 
\\ &
\frac{ C M }{\eta \tilde \eta^2} \le \frac{\beta}{4} R^2,
\\ &
\frac{C}{\tilde \eta} \le \frac{\mathcal C}{4} \beta \chi(\beta) R^2,
\\ & CM^2\mathcal C\chi(\beta) \le  \frac{\beta}{4} R^2 \tilde \eta^3,
\\ &
\eta^{\d+1}\tilde \eta^{\d+2} R^2 \ge C M \mathcal C \chi(\beta) ,
\\ &
\tilde \eta R \ge \max\(\beta^{-\hal}\indic_{\d=2}, 1\) ,
\\ &
CM \mathcal C \chi(\beta) <\tilde \eta^2  R
\end{aligned} 
\right. 
\end{align*}
    for some constant $C>0$ large enough, and depending only on $\d,m$ and $\Lambda$.
    We can  take $\eta$ small enough that the first condition is realized.
    We can then make the other conditions realized by requiring 
    $$\left\{ \begin{aligned} &R^2\min( \beta \chi(\beta)^{-1},1) \ge C'\\
    & \tilde \eta= C''\max\((R^2 \beta \chi(\beta)^{-1})^{-1/3},  R^{-\frac{2}{\d+2}}\)
    \end{aligned}\right.$$
    where $C', C''$ are large enough constants 
     depending on the other parameters.}
    Combining ~\eqref{438} with~\eqref{controlbad}, we obtain the result. 
\end{proof}

\section{Main bootstrap and first conclusions}\label{sec5}
This section contains the core of the proof, i.e. the bootstrap procedure that allows to show that 
 if local laws hold down to a certain scale, they  hold at slightly smaller scales.
We note that the local laws are valid up to the boundary as long as one remains in the set where  $\mu\ge m>0$.

\begin{prop}\label{proloclaw} Assume $\mu$ and $U$ are as in Proposition \ref{minoK}. 
  Let $\mu$ be a density such that $0<m \le \mu \le  \Lambda $  in  the set $\Sigma$. Assume that   $\mu(U)=\bN$ is  an integer and that 
 \be \label{asspu}\text{if } \ \d \ge 3, \quad
 |\partial U| \min(1, \beta^{\frac{1}{\d-2}} ) \le \beta \bN.\ee 
 There exists $C>2$ depending only on $\d,m$ and $\Lambda$ such that the following holds.
Let
\be \label{minoR}   \rb= C \max\(1, \beta^{-\hal} \chi(\beta)^{\frac12} , \beta^{\frac{1}{\d-2}-1} \indic_{\d\ge 5} \)
.\ee
Let $\car_R(x)$ be a cube of size $R\ge \rb $ centered at $x$ included  in  $\Sigma$ satisfying  \begin{equation}
\label{conddist}
\dist(\car_R(x) , \p \Sigma\cap U) 
\ge  d_0:= 
 \ed{C \max\( \(\frac{N^{\frac1\d} }{ \max(1, \beta^{-\hal} \chi(\beta)^{\frac12})}\)^{-\frac23}   N^{\frac1\d} ,   N^{\frac{1}{ \d+2}}\) }.
\end{equation}
Then  we have, for $C_0 := \frac{2C}{4\cd}$ with~$C$ the constant in~\eqref{14},
\begin{equation}\label{locallawint}
\log \Esp_{\Q(U)} \bigg( \exp\bigg( \hal \beta\bigg( \G^{\car_{R}(x)\cap U}(\cdot, U)+ C_0 \#(\{X_\bN \} \cap \car_R(x)\bigg)  \bigg)\bigg)
\leq 
C\beta\chi(\beta  )R^\d.
\end{equation}
\end{prop}

\begin{proof}[Proof of Proposition \ref{proloclaw}] We first note that it is enough to prove the result in  hyperrectangles $Q_R\in \mathcal{Q}_R$, with sides parallel to those of $U$ and even more generally in $Q_{R-2\tilde \ell}$ if $\tilde \ell<\frac14R$, with $R\ge \rb$ as in~\eqref{minoR}. Indeed, thanks to the lower bound on $\mu$,  general cubes  of size satisfying~\eqref{minoR} can be covered by a finite number of such hyperrectangles. 
 The proof then proceeds by a bootstrap on the scales:
we wish to show that if 
 \be\label{hyprec}
 \log \Esp_{\Q(U)} \( \exp \(\frac{\beta}{2} \( \G^{\car_L(x) }(\cdot, U)+ C_0 \#(\{X_\bN \} \cap \car_L(x))\) \)\) \le \mathcal C \beta\chi(\beta) L^\d ,\ee
 for any $\car_L(x)$  sufficiently far from $\partial \Sigma$, then  if $ \frac34 L \ge R \ge \hal L$, and as long as $R$ is large enough, we have
 \be\label{hyprecfin}
 \log \Esp_{\Q(U)} \( \exp \(\frac{ \beta}{2} \( \G^{Q_{R-2\tilde \ell}}(\cdot, U)+C_0 \#(\{X_\bN \} \cap Q_R) \) \)\) \le \mathcal C \beta\chi(\beta) R^\d.  \ee 
   By iteration, this will clearly imply the result: indeed 
  in view of Corollary \ref{pttproba} and \eqref{asspu}  and up to changing $\mathcal C$ if necessary,  we have that~\eqref{hyprec} holds for  $L\ge \hal N^{\frac1\d}$.  
Without loss of generality, we may now  assume for the rest of the proof that  $L \le \hal N^{\frac1\d}$.

 To make sure that the constants are independent of~$\beta $ and~$R$,  we have used the notation~$\mathcal C$, and we wish to prove~\eqref{hyprecfin} with the same constant~$\mathcal C$ as in~\eqref{hyprec}. 
 In the sequel, unless specified, {\it all constants~$C>0$ will be independent of~$\mathcal C $}, i.e they may depend only on~$\d,m$ and~$\Lambda$.
 
   \smallskip

 Let us now consider~$Q_R\in \mathcal Q_R$, denote~$\mn=\mu(Q_R\cap U)$  and as previously, denote by~$\mathcal A_n$ the event that~$X_\bN$, a configuration of~$\bN$ points in~$U$,  has~$n$ points in~$Q_R\cap U$.
We wish to control 
\begin{align*}
&
\Esp_{\Q(U)} \( \exp \(\frac{\beta}{2} ( \G^{Q_{R-2\tilde \ell} }(\cdot, U)+C_0 n) \)\)
\\ & \qquad
=\sum_{n=0}^\bN  \exp\left( \frac{\beta}{2} C_0 n\right)
 \Esp_{\Q(U)} \( \exp \(\frac{\beta}{2} ( \G^{Q_{R-2\tilde \ell} }(\cdot, U) )\indic_{\mathcal A_n} \)\).
\end{align*}
The terms in the sum for which $n$ is close to $\mn$, more precisely $|n-\mn| \le K R^{\d-\hal}$ are easily treated using~\eqref{resLK}. The terms for which $|n-\mn|> KR^{\d-\hal}$  will be handled separately and controlled by energy-excess considerations. 

To apply Proposition \ref{proboot} we need $Q_{R+\tilde \ell}$ to be included in a cube $\car_L$ in which the local laws hold and at distance $\ge \tilde \ell $ as in~\eqref{deftl} from $\partial \Sigma$.
At the first iteration 
$L$ is of order $N^{\frac1\d}$ and $ R \ge \hal L$ so we need 
$$
\dist(Q_R,\partial \Sigma) \geq 
\ed{ C''  \max  \( \(\frac{N^{\frac1\d} }{ \max(1, \beta^{-\hal} \chi(\beta)^{\frac12})}\)^{-\frac23}   N^{\frac1\d} ,   N^{\frac{1}{ \d+2}}\) }.
 $$
At further iterations, to have $Q_{R+\tilde \ell} $ be included in $\car_L$,  we need a further distance of  
\ed{$ C''  \max  \( \(\frac{R }{ \max(1, \beta^{-\hal} \chi(\beta)^{\frac12})}\)^{-\frac23}   R,   R^{1-\frac{2}{\d+2}}\) .$  Since $R $ is multiplied by a factor in $[\hal , \frac34]$ at each step,  summing  the series over the iterations gives a total distance 
$$C''' \max  \( \(\frac{N^{\frac1\d} }{ \max(1, \beta^{-\hal} \chi(\beta)^{\frac12})}\)^{-\frac23}   N^{\frac1\d} ,   N^{\frac{1}{ \d+2}}\) $$ hence  a condition of the form~\eqref{conddist} suffices. }

 \smallskip

\noindent
{\bf Step 1: the bad event}. We claim that  in the bad event  $|\mn -n| \ge KR^{\d-\hal} $, we have 
\be\label{claim80}
 \G^{Q_{R+3}}(\XN, U) - \G^{Q_R}(\XN, U) \ge C R^{1-\d} |\mn - n|^2 - C \mathcal{N}_{Q_{R+3}}\ee
where  $ \mathcal{N}_{Q_{R+3}}$  denotes the number of points in $Q_{R+3}$ and $C>0$ depends only on $\Lambda$ and $\d$.
Assuming this, and changing $C_0$ to the larger constant in~\eqref{claim80} if necessary,  we then write 
\begin{multline}\label{510}\Esp_{\Q(U) }\( \exp \(\frac\beta2 (\G^{Q_R }(\cdot, U)+  C_0 n) \)   \indic_{\mathcal A_n} \)
\\
\le \Esp_{\Q(U) }\( \exp \( \frac \beta2 (\G^{Q_{R+3} }(\cdot, U)  +   C_0 \mathcal{N}_{Q_{R+3} }) \) \indic_{\mathcal A_n}   \) \exp\(- \beta C R^{1-\d}|\mn -n|^2 +\beta C_0 n\). 
\end{multline} Since $ L \le 2R$ and $|\mn-n| \ge K R^{\d-\hal}$, we now see that if we choose $K = C \sqrt{\mathcal{C} \chi(\beta)}$ where $C>0$ is large enough and depends only on $C, C_0$ and $\d$,
the exponent in the second term in the right-hand side is at most $-  \mathcal C \beta \chi(\beta) L^\d$.

Using ~\eqref{restri1},~\eqref{add3} and~\eqref{14} we may check that 
$$\G^{Q_{R+3}}(\cdot, U)+C_0 \mathcal N_{Q_{R+3}}\le \G^{\car_L} (\cdot,U)+C_0 \mathcal N_{\car_L}$$ hence in view of 
~\eqref{510} and the assumption that~\eqref{hyprec} satisfied in a cube $\car_L$ containing $Q_{R+3}$, we may  bound 
\begin{multline}\label{stp1}\sum_{n=KR^\d}^\bN \log \Esp_{\Q(U) }\( \exp \(\frac \beta2 (\G^{Q_R }(\cdot, U)+ C_0 n) \) \indic_{\mathcal A_n}    \)
\\ \le \exp\(-\mathcal C \beta \chi(\beta) L^\d\)  \sum_{n=0}^\bN \Esp_{\Q(U) }\( \exp \( \frac\beta2 (\G^{\car_L}(\cdot, U)  + \beta  C_0 \mathcal N_{\car_L } )\) \indic_{\mathcal A_n}   \)\le 1.
\end{multline}

To prove the claim, in view of~\eqref{disc1} we may  write 
\be\label{bf45}C \int_{Q_{R+2}\backslash Q_{R+1}} |\nab u_{\rrtt}|^2 \ge CR^{1-\d} \( |n-\mn| - C(1+ \|\mu\|_{L^\infty})R^{\d-1}\)^2\ge c R^{1-\d} |\mn -n|^2\ee if $K$ is chosen large enough (depending on $\d$ and $\Lambda$),
where $c>0$ is a constant depending only on $\d,m$ and $\Lambda$.
In view of~\eqref{add3} we have 
\be\label{fqr3}
\G^{Q_{R+3}}(X_\bN, U) - \G^{Q_R}(X_\bN, U) \ge \G^{Q_{R+3}\backslash Q_R}(X_\bN, U).\ee
By~\eqref{14}, we may write that 
$$ C\G^{Q_{R+3}\backslash Q_R}(X_\bN, U) \ge \int_{Q_{R+3} \backslash Q_R} |\nab u_{\rrtt}|^2 - C \mathcal N_{Q_{R+3}}$$ where $u_{\rrtt}$ is computed with respect to $Q_{R+3}\backslash Q_R$. But by definition $ \int_{Q_{R+3} \backslash Q_R} |\nab u_{\rrtt}|^2$ is larger than 
$ \int_{Q_{R+2} \backslash Q_{R-1}} |\nab u_{\rrtt}|^2$ with this time $\rrtt$ computed with respect to $U$, which is bounded below by~\eqref{bf45}. 
Inserting into~\eqref{fqr3} we thus conclude~\eqref{claim80}.

 \smallskip

\noindent 
{\bf Step 2: the good event.}
We next  consider the terms for which $|\mn -n |\le KR^{\d-\hal}$.
For those, we may apply 
Proposition \ref{proboot} (at least if $R >C$ with $C$ made large enough). We need to assume~\eqref{condRL}. In view of~\eqref{resLK} we  may thus write  
\begin{align*}
\lefteqn{
\sum_{|n-\mn|\le KR^{\d-\hal}}  
 \Esp_{\Q(U)} \( \exp \(\frac{\beta}{2} \( \G^{Q_{R-2\tilde \l} }(\cdot, U) +C_0n \) \)\indic_{\mathcal A_n} \)
 } \qquad \  & 
 \\ & 
 \le  \sum_{n=\mn-KR^{\d-\hal}}^{\mn+ KR^{\d-\hal}} \exp\left( \beta
 \( \frac{ \mathcal C }{ 4} \chi(\beta) R^\d   +|n-\mn|
 + C_0 n\)    \)   \frac{\K^{\frac\beta2}(Q_R)}{\K^\beta(Q_R)}  
+ \gamma_n\exp\left( \frac{\beta}{2} C_0 n \) .\end{align*} 
Recalling the choice of $K$ as $C\sqrt{ \mathcal C \chi(\beta)}$ and using that $\mn= \mu(Q_R) \le \Lambda R^\d$,  we have that if 
$|n-\mn|\le KR^{\d-\hal}$, then if $R \ge \mathcal C \chi(\beta)$, we have $KR^{\d-\hal} \le C R^\d$ and  $n\le CR^\d$, with $C$ depending only on $\d, m, \Lambda$. 

\smallskip

Using~\eqref{boundgamma}, and the fact that, by~\eqref{majok} and~\eqref{minok},
\begin{equation}
\log\left( \frac{ \K^{\beta/2} (Q_R)}{ \K^\beta(Q_R)} \right) 
\leq 
C\beta \chi(\beta) R^\d  + C R^{\d-1} \min(1, \beta^{\frac{1}{\d-2}}),
\end{equation}
we deduce that, for every $R \ge \mathcal C \chi(\beta)$,
\begin{align*}
& \sum_{n=\mn-K R^{\d-1/2}}^{\mn+ KR^{\d-1/2}}   \Esp_{\Q(U)} \( \exp \(\frac{\beta}{2} \( \G^{Q_{R-2\tilde \l} }(\cdot, U) +C_0n \)\)\indic_{\mathcal A_n} \)
 \quad
\\ & \quad 
\leq 2 R^\d \exp\(\beta \( \frac{ 3 \mathcal C }{ 8} \chi(\beta) R^\d 
  + C_0 CR^\d \) \)   \exp\( C\beta \chi(\beta)R^\d + C \min (\beta^{\frac{1}{\d-2}},1) 
R^{\d-1}\)  
\\ & \quad \qquad 
+\exp\(\beta C_0 C R^\d -  \mathcal C \beta \chi(\beta) R^\d\).
\end{align*} 
Making $\mathcal C$ larger if necessary (compared to the constants $C_0$, $C$ appearing here) we deduce 
\begin{align}
\label{bonterm}
& \sum_{n=\mn-K R^{\d-1}}^{\mn+KR^{\d-1}} \Esp_{\Q(U)} \( \exp \(\frac{\beta}{2} \( \G^{Q_{R-2\tilde \l} }(\cdot, U) +C_0n\) \)\indic_{\mathcal A_n} \) 
\\  & \qquad \notag
\le  \exp\( \beta \frac{\mathcal C}{2} \chi(\beta) R^\d   +   C \min (\beta^{\frac{1}{\d-2}}, 1) R^{\d-1}   + C \log R    \).
\end{align} The term in $  \min (\beta^{\frac{1}{\d-2}},1)R^{\d-1}$ can be absorbed into $\beta \chi(\beta)R^\d$ if 
we assume in addition that $R\ge C \beta^{\frac{1}{\d-2}-1}$  (for dimension $\d\ge 3$), this condition itself is implied by $ R>C \beta^{-\hal}$ if $\d=3,4$.
The logarithmic term can then also be absorbed  using ~\eqref{minoR}.
\smallskip

\noindent
 {\bf Step 3: Conclusion}.
Combining~\eqref{stp1} and~\eqref{bonterm}, we conclude that 
~\eqref{hyprecfin} holds and this finishes the proof.    
\end{proof}

\begin{coro}\label{coro61}
Assume the hypotheses of Proposition \ref{proloclaw} for $\car_R(x)$ with $R\ge \rb$ as in~\eqref{minoR} and let $B$ be a ball such that $2B \subseteq \car_R(x)$. 
There exists $C>0$ depending only on $\d,m$ and $\Lambda$ such that
\be \label{result2}
 \log \Esp_{\Q(U)} \Bigg(       \exp \Bigg( \frac{\beta}{C } R^{2(1-\d)} \rb^{\d-1}  \Bigg(\int_{\car_R} \sum_{i=1}^\bN \delta_{x_i}- d\mu\Bigg)^2  \  \Bigg) \Bigg)  \le    C  \beta \chi(\beta )  \rb^{\d}  ,
 \ee
and letting $$D:= \int_B \( \sum_{i=1}^N \delta_{x_i}- d\mu\right) $$
we have
\be\label{result3}\log \Esp_{\Q(U)} \( \exp\(  \frac{\beta}{C} \frac{D^2 }{R^{\d-2}} \min \(1, \frac{|D|}{R^\d}\) \) \) \le C \beta \chi(\beta ) R^\d.  
\ee
\end{coro}
\begin{proof}
We may suppose $x=0$. 
First, we observe that by choice of $C_0$ and~\eqref{14} we  have for any  $R\ge \rb$
\be \label{fini} \log \Esp_{\Q(U)} \( \exp \( \frac{1}{2C} \beta \int_{\car_R}| \nab u_{\rrtt}|^2 \) \) 
\le C  \beta \chi(\beta)R^\d  \ee
where $\rrtt$ is computed with respect to $\partial \car_R$.
We next may use either first~\eqref{disc10}--\eqref{disc1} or second~\eqref{disc30}--\eqref{disc3} to deduce from this a control of the discrepancy.

In the first way we cover  $\car_{R+2}\backslash \car_{R-2}$ by at most $O( (R/\rb)^{\d-1} )$ cubes  $Q_k$ of size $\rb$.  Applying~\eqref{fini} for the cubes $Q_k$ and using 
the generalized H\"older inequality 
\be\label{gholder}
\Esp(f_1 \dots f_k) \le \prod_{i=1}^k \Esp(f_i^k)^{\frac{1}{k}},\ee which can be proved by induction, 
 we find 
\be \label{tryu}
\log \Esp_{\Q(U)} \( \exp\( C^{-1} \beta (\frac{R}{\rb})^{1-\d}\int_{\car_{R+\rb}\backslash \car_{R-\rb}} |\nab u_{\rrtt}|^2 \) \) \le C\beta \chi(\beta)\rb^\d,
\ee for some constants $C$ depending only on $\d ,m$ and $\Lambda$.
In view of~\eqref{disc10}--\eqref{disc1}, we  then bound 
$$\left|\int_{\car_R} \sum_{i=1}^N \delta_{x_i}- d\mu\right|^2  \le C \|\mu\|^2_{L^\infty} R^{2(\d-1)} + C R^{\d-1}  \int_{\car_{R+1}\backslash \car_{R-1}} |\nab u_{\rrtt}|^2  .$$
Inserting into~\eqref{tryu}, we find ~\eqref{result2}.

In the second way, we simply bound $\int_{B_{2R}}|\nab u_{\rrtt}|^2$ using~\eqref{fini}.
Inserting into~\eqref{disc30}--\eqref{disc3} directly yields~\eqref{result3}.
\end{proof}

\subsection{Conclusion: proof of Theorem \ref{th3}}

We apply Proposition \ref{proloclaw}  in $U=\R^\d$, since~\eqref{asspu} is then automatically satisfied, it yields that for any $ \car_R(x)$ satisfying  ~\eqref{conddist0},  the estimate~\eqref{locallawint} holds.
Then
~\eqref{loclawpoints0} and~\eqref{loclawpoints00} follow
 from  Corollary \ref{coro61}. The bound~\eqref{loclawphi} follows from the combination of~\eqref{locallawint} and
~\eqref{fluctuationsbu} applied in $\R^\d$.
Finally,~\eqref{loclawdistmin0} is a consequence of~\eqref{15} and ~\eqref{locallawint} applied with $R=\rb$.

\begin{remark}We note that similarly, all the results of Theorem \ref{th3} hold for the Neumann Gibbs measure $\Q(U)$ of~\eqref{defQ} for any $U$ and they can also be proven to hold for the  Dirichlet Gibbs measure $\P_N(U)$ of~\eqref{defP} away from the boundary.
 \end{remark}

\subsection{Proof of Corollary \ref{coro1}} 
Let us recall the setup for point processes, following~\cite{lebles}. We denote by $\mathcal X(A)$  the set of local finite point configurations on $A \subset \R^\d$ or equivalently the set of non-negative, purely atomic Radon measures on A giving an
integer mass to singletons.
 We use $\mathcal C$ for denoting a point
configuration and we will write $\mathcal C$ for $\sum_{p\in \mathcal C} \delta_p$ and $|\mathcal C|(A)$ for the number of points of the configuration in $A$.
We endow  $\mathcal X(\R^\d)$ 
with the topology induced by the topology of weak convergence of Radon measure (also known as vague convergence or convergence against compactly
supported continuous functions), and we define the following distance on $\mathcal X$:
\begin{equation}
\label{517} 
d_{\mathcal X}(\mathcal C, \mathcal C')=\sum_{k=1}^\infty \frac1{2^k}\frac{\sup \left\{\int_{\car_k} f d(\mathcal C-\mathcal C') , \|\nab f\|_{L^\infty(\R^\d)} \le 1\right\}}{ |\mathcal C|(\car_k) + |\mathcal C'|(\car_k)}.
\end{equation}
The subsets $\mathcal X(A)$
inherit the induced topology and distance. As seen in~\cite[Lemma 2.1]{lebles}, the space $\mathcal X(A)$ is then a Polish space.

\smallskip

Now let $\beta$ be fixed and let $x$ be a point as in the statement of the corollary. Let $P_N$ denote the the push-forward of $\PNbeta$ under the map from $(\R^\d)^N$ to $\mathcal X(\R^\d)$ given be 
 $$(x_1, \dots, x_N) \mapsto  \sum_{i=1}^N \delta_{x_i-x}.$$
We wish to show that $P_N$ is tight. Indeed, since $\mathcal X(A)$ is Polish, Prokhorov's theorem will then imply the existence of a convergent subsequence for the topology on $\mathcal X(A)$. 
Let now $\mathcal N_k$ denote the map $\mathcal C \mapsto |\mathcal C|(\car_k)$, i.e. $\mathcal N_k(\mathcal C)$ gives the number of points of $\C$ in $\car_k$.
By~\eqref{loclawpoints00}, we have that  for any $k$,  if $M$ is large enough,
$$\PNbeta\(  \mathcal N_k( \{x_1-x, \dots, x_N-x\}) \ge Mk^\d\) \le \exp\(- C_\beta  M^2 k^{\d+2}\)$$
or in other words, by definition of $P_N$, 
$$P_N\(\mathcal N_k(\mathcal C)  \ge MR^\d\) \le \exp\(- C_\beta  M^2 k^{\d+2}\).$$
It follows that letting $K_M= \cap_{k=1}^\infty \{ \mathcal C, \mathcal N_k(\mathcal C)   \le M R^\d \} $,
$$P_N( K_M ) \ge 1-\frac1M,$$
hence to conclude that $P_N$ is tight, it suffices to justify  that $K_M$ is compact in $\mathcal X(\R^\d)$. Let $(\mathcal C_n)_n$ be a sequence  of point configurations in $K_M$. By definition $|\mathcal C_n|(\car_k)$ is bounded uniformly by some $p_k$ independent of $n$  for each $k$, hence by compactness of $\car_k^{p_k}$, we may find a subsequence such that $\mathcal C_n $ converges in $\mathcal X(\car_1)$, and by diagonal extraction we may find a subsequence of $n$ such that $\mathcal C_n$ converges in $\mathcal X (\car_k)$ for each $k$. By definition of the distance~\eqref{517}, this implies that (after extraction) $\mathcal C_n$ converges in $\mathcal  X(\R^\d)$. This proves that $K_M$ is compact and finishes the proof of convergence of $P_N$ up to extraction.

The fact that the points are simple under the limiting process is a consequence of~\eqref{loclawdistmin0}. The finiteness of the moments of all order then follows in view of the bound of all moments of the number of points, provided by~\eqref{loclawpoints00}.

\section{Leveraging on the local laws: free energy estimates}\label{sec6}

\subsection{An almost additivity result}
     We next prove a general subadditivity result that makes use of the local laws.
    Comparing it with the a priori superadditivity result of~\eqref{superad2} gives additivity up to an error.
     
     \begin{prop}   Assume that $0<m\le \mu \le \Lambda$ in $\Sigma$.
      Assume $\hat U$ is a subset of $ \Sigma$ at distance $d\ge d_0$ from $\pa \Sigma$ with $d_0$ as in~\eqref{conddist}, and is a disjoint union of $p$ hyperrectangles $Q_i$ belonging to $\mathcal{Q}_{R}$, with $R \ge \rb$  satisfying 
\be 
\label{rrb} 
R \ge \rb+ \( \frac{1}{\beta\chi(\beta)} \log \frac{R^{\d-1}}{\rb^{\d-1}}\)^{\frac1\d}
\ee
and in addition, if $\d \ge 4$,  \be\label{rrb2}
R\ge \max(\beta^{\frac{1}{\d-2}-1},1) N^{\frac1\d} d^{-1}.\ee
Then there exists $C$, depending only on $\d, m$ and $\Lambda$, such that 
\begin{align}
\label{subad3}
& \left| \log \K (\R^\d) 
- \left( \log \K (\R^\d\backslash \hat U)+ \sum_{i=1}^p \log \K (Q_i) \right) \right|
\\ & \qquad \notag
\leq C p  \(       \beta R^{\d-1}     \rb   \chi(\beta) +\beta^{1-\frac1\d} \chi(\beta)^{1-\frac1\d} \(\log  \frac{R}{\rb}\)^{\frac1\d}   R^{\d-1}  \).
\end{align}
     If $U$ is a subset of $\Sigma$ equal to a disjoint union of $p$  hyperrectangles $Q_i$ belonging to $\mathcal{Q}_{R}$, with $R \ge \rb$ satisfying \eqref{rrb}, $N_i= \mu(Q_i)$, then we have, with $C$ as above, 
     \be
     \label{subad4} 
     \left| \log  \K(U) -
     \sum_{i=1}^p \log \K (Q_i) \right| \leq C p \Bigg(    \beta R^{\d-1}      \chi(\beta) \rb  +
            \beta^{1-\frac1\d} \chi(\beta)^{1-\frac1\d} \(\log  \frac{R}{\rb}\)^{\frac1\d} R^{\d-1}  \Bigg).
     \ee
      \end{prop}
\begin{proof}
We will only prove upper bounds for $\log \K(\R^\d)$ and $\log \K(U)$, since the matching lower bounds are direct consequences of~\eqref{superad2}, Stirling's formula and the control~\eqref{lrob} below.
  
\smallskip

We recall that the local laws hold down to scale $\rb$ in $U=\R^\d$. In particular, for any cube~$\car$ in~$\hat U$ of size~$r \ge \rb$,  we have 
     \be \label{lolocd} \log \Esp_{\Q(\R^\d)} \( \exp \( \frac{1}{2C} \beta \int_{\car}| \nab u_{\rrtt}|^2 \) \) 
\le C r^\d \beta \chi(\beta).\ee 
      Let $Q_1$ be the first rectangle in the list, and let us denote by $n$ the number of points a configuration has in $Q_1$ and by $\mn = \mu(Q_1)$. Let us also define 
      $$\hat Q_1:=\{x\in Q_1, \dist (x, \partial Q_1) \le r \}$$ and 
     $$\mathcal B:= \left\{X_N\in (\R^\d)^N \,:\,   |n- \mn|\le \ep , \quad  \sup_x \int_{\hat Q_1\cap \car_{r}(x)} |\nab u_{\rrtt}|^2 \le M \chi(\beta) r^\d
        \right\}$$ where we let
     $$\ep:=M\( R^{\d-1} \sqrt{\chi(\beta) \rb} \)$$ 
     and $M>0$ is to be selected below.  The first condition~$|n- \mn|\le \ep$ in the definition of~$\mathcal{B}$ has large probability in view of~\eqref{result2}.
    For the second condition, by a covering argument we have $ \frac{R^{\d-1}}{r^{\d-1}}$ conditions to satisfy and each of them has probability
     of the complement bounded by $\exp\(- M \beta \chi(\beta)r^\d\)$ if $M$ is large enough, in view of
              ~\eqref{lolocd}.
               Using a union bound  we thus  have 
               $$\Q(\R^\d) [\mathcal B^c]\le   \frac{R^{\d-1}}{r^{\d-1}}  \exp\(- M \beta \chi(\beta)r^\d\)$$
               and this is $\le \hal$ if 
               $$  \frac{R^{\d-1}}{r^{\d-1}}  \exp\(- M \beta \chi(\beta)r^\d\)\le \hal$$
               so we choose 
               \be \label{defir}
               r= M \rb+ \( \frac{1}{\beta\chi(\beta)} \log \frac{R^{\d-1}}{\rb^{\d-1}}\)^{\frac1\d}\ee
               which satisfies the condition if $M$ is large enough.
               It follows that 
     $$N^{-N} \int_{\mathcal B^c} \exp\left( -\beta \G( \cdot)  \right)\, d\mu^{\otimes N} 
     = \Q(\R^\d)[ \mathcal B^c] \K(\R^\d) \le 
      \hal\K (\R^\d). $$
          We thus have
\begin{multline*}
\frac{N^N }{2} \K(\R^\d)
\leq 
\int_{\mathcal B}\exp\(-\beta \G(\cdot)\)d\mu^{\otimes N}
\\ 
\leq 
\sum_{n=\mn-\ep  } ^{\mn+\ep} \frac{N!}{n!(N-n)!}     
\int_{\mathcal B} \exp\(-\beta \H_{\R^\d}(\cdot, Q_1) \)d\mu^{\otimes n}   \int_{\mathcal B} \exp\(-\beta \H_{\R^\d}(\cdot, \R^\d\backslash Q_1) \)d\mu^{\otimes (N-n)},
\end{multline*}
where for the second line we subdivided the event over the possible values of $n$ and  applied~\eqref{restri1}.

We now apply the results of Corollary \ref{coro43} with $L=r$ to $Q_1$ and $\R^\d \backslash Q_1$, combined with Remark \ref{remerr}. For that we check that~\eqref{condsurL} is satisfied since $ r \ge \rb$, and obtain
\begin{align*}
\K(\R^\d)  
& \leq  
2 \K (Q_1)\K(\R^\d \backslash Q_1) \sum_{n=\mn-\ep}^{\mn+\ep} \frac{N! N^{-N}}{n!(N-n)!} n^n  (N-n)^{N-n}  
\\ & \qquad 
\times\exp\( C \beta \(R^{\d-1} r  \chi(\beta)M+\ep  \)   +  \frac{M^2 \chi(\beta)^2 R^{\d-1}}{  r} \).
\end{align*}
Next, using Stirling's formula we have 
      $$ \frac{N! N^{-N} n^{n} (N-n)^{N-n}  }{n!(N-n)!   } \le C  \sqrt{ \frac{ N}{ 2\pi n  (N-n)}}\le C$$
and we deduce
\begin{align*} 
& \log  \K(\R^\d) 
\\ & \quad 
\le \log   \K(Q_1)+\log \K(\R^\d\backslash Q_1)+C+ \log \ep
+  \beta \(M R^{\d-1} r  \chi(\beta) +\ep\)
+  \frac{M^2 \chi(\beta)^2 R^{\d-1}}{  r} .
\end{align*}
Since \be\label{eqrf}
r \ge \rb \ge \max(1, \chi(\beta)^{\hal}\beta^{-\hal} ) \ge 1\ee we have $\frac{\chi(\beta)}{r}\le \beta r$ so   we may absorb the last term. Also, 
 since $ r \ge \rb \ge 1$ and $\chi(\beta)\ge 1$,  by definition of $\ep $ we may absorb  $ \ep$ into $M R^{\d-1} r \chi(\beta)   $.
       Since $R \ge \rb \ge \sqrt{\chi(\beta)}$, we  have $ R^{\d-1} \sqrt{\chi(\beta) \rb}\le C R^\d$, so
inserting the definition of $r$, 
we find 
\begin{multline*}
\log  \K(\R^\d)  
\\  
\le  \log   \K(Q_1)+\log \K(U\backslash Q_1)  + C \log  R       +  C\beta R^{\d-1}     \rb \chi(\beta) + C
            \beta^{1-\frac1\d} R^{\d-1} \(\log  \frac{R}{\rb}\)^{\frac1\d}  \chi(\beta)^{1-\frac1\d} .
\end{multline*}
Finally, 
 since $R \ge \rb \ge C\chi(\beta)^{\hal} \beta^{-\hal} $ we have that, for every $R \ge \rb$,
\be \label{lrob}
\log R \le  C \beta \chi(\beta)\rb R^{\d-1},
\ee 
which allows us to absorb the $\log R$ term into the others. 

\smallskip

We may now iterate this by bounding $\log   \K(\R^\d \backslash \cup_{i=1}^j Q_i)   $ in the same way thanks to the local laws up to the boundary 
 of Proposition  \ref{proloclaw}.  For this we check that  for every $j\le p$, $\R^\d \backslash \cup_{i=1}^j Q_i$ is a set satisfying the assumptions of the proposition, in particular \eqref{asspu}: indeed, $$\mu(\R^\d\backslash \cup_{i=1}^j Q_i) \ge \frac{ m d  N^{1-\frac1\d}}{C} $$ where we recall $d\ge d_0$,
 while $$|\pa ( \R^\d \backslash \cup_{i=1}^j Q_i)| \le C p R^{\d-1}\le C \frac{N}{R^\d} R^{\d-1}=C \frac{N}{R} $$ hence the condition \eqref{asspu} if \eqref{rrb2} holds. This yields~\eqref{subad3}. 
 
 The proof of~\eqref{subad4} is analogous, using that the local laws hold up to the boundary for $\Q(U)$ and that for any union of hyperrectangles in $\mathcal Q_R$ with $R \ge \rb$ we have $\min(1,\beta^{\frac{1}{\d-2}}) |\pa U | \le C \beta \mu (U)$, for some $C>0$ depending only on $\d, m, \Lambda$, hence \eqref{asspu} is also satisfied.
\end{proof}
     
     \subsection{Free energy for uniform densities on hyperrectangles: proof of Theorem \ref{th1}}
     We are now ready to compute~$\log \K(Q_R)$  when the density is constant in a rectangle~$Q_R$, taking advantage of the  superaddivity of~$\log \K$  and the almost additivity provided by~\eqref{subad4}.
     We reintroduce the $\mu$ dependence in the notation $\K(U, \mu)$.
     

     \begin{prop}\label{pro62}
There exists a function $\mf$  on $(0, +\infty]$ and a constant $C>0$, depending only on~$\d$,
such that the following hold: 
\begin{itemize}
\item for every $\beta>0$, 
\be -C \le \mf(\beta)\le  C \chi(\beta). \ee

\item $\mf$ is locally Lipschitz in $(0, +\infty)$ with 
\be\label{boundfp}
|\mf'(\beta)|\le \frac{C\chi(\beta)}{\beta}.\ee 
\item
 if $Q_R \in \mathcal{Q}_R$ and $R \ge \rb$ satisfies
 $$R \ge \rb + \( \frac{1}{\beta \chi(\beta) } \log \frac{R^{\d-1}}{\rb^{\d-1}}\)^{\frac1\d},$$ then
 \be\label{formprop}\left|\frac{\log \K(Q_R, 1)}{\beta |Q_R|}+\mf(\beta)\right|\le C  \(  \frac{\chi(\beta) }{R} \( \rb + \beta^{-\frac1\d} \chi(\beta)^{-\frac1\d}\log^{\frac1\d} \frac{R}{\rb} \) \).
\ee
\end{itemize}
\end{prop}

\begin{proof}We first start by treating the case of a cube $\car_R$ with  $R^\d$ integer.
In view of~\eqref{superad2}, we have 
$$\frac{1}{\beta} \log \K(\car_{2R}, 1)\ge  O\(\frac{\log N}{\beta}\) + \frac{2^\d}{\beta} \log \K(\car_R, 1).$$
Thus, denoting $\phi(R)= \frac{\log\K(\car_R,1)}{\beta R^\d}$, this means that 
$$\phi(2R)  \ge \phi(R)+ O\(\frac {\log R}{\beta R^\d}\) $$
and summing these relations we have
$$\phi(\infty) \ge \phi(R)+ O\( \sum_{k=1}^\infty\frac{\log R }{\beta 2^k R}\), $$
that is, 
\be \label{ubb}\phi(R) \le \phi(\infty)+ O\(\frac{\log R}{\beta R^\d}\).\ee
On the other hand, 
in view of~\eqref{subad4}, 
we have
\be\label{subad23}
\K(\car_{2R},1)  \le    2^\d   \log \K (\car_R, 1) + C\beta R^\d \(
 \frac{\chi(\beta) }{R} \( \rb + \beta^{-\frac1\d} \chi(\beta)^{-\frac1\d}\log^{\frac1\d} \frac{R}{\rb}\) \), \ee
that is,
 \begin{equation*}
  \phi(2R) \le  \phi(R) + C \(  \frac{\chi(\beta) }{R} \( \rb + \beta^{-\frac1\d} \chi(\beta)^{-\frac1\d}\log^{\frac1\d} \frac{R}{\rb} \)\)  .
\end{equation*} Summing these relations, we conclude just as above that 
  \be\label{lbb}\phi(\infty) \le \phi(R) + O\(  
 \frac{\chi(\beta) }{R} \( \rb + \beta^{-\frac1\d} \chi(\beta)^{-\frac1\d}\log^{\frac1\d} \frac{R}{\rb}\)   \).\ee 
Denoting by 
$-\mf(\beta)$ the value $\phi(\infty)$ and recalling~\eqref{lrob}, we have the desired bounds for $Q_R=\car_R$  by 
 combining~\eqref{ubb} and \eqref{lbb}.
 We may then generalize to $Q_R\in \mathcal{Q}_R$ by another application of the sub/superaddivity  of~\eqref{superad2} and~\eqref{subad4} and the a priori bounds~\eqref{majok} and~\eqref{minok}.

 \smallskip
 
 In view of~\eqref{majok} and~\eqref{minok} applied with $\mu=1$ and $U=\car_R$, we also have 
 $- C \chi(\beta) \le  \phi(R) \le C $ with $C$ independent of $\beta$.
 This implies that $  -C \le \mf(\beta)\le C\chi(\beta)$.

\smallskip
 
To prove that $\mf$ is locally Lipschitz, let us temporarily highlight the $\beta$-dependence and compute 
\begin{align*}
\log \frac{\K^{\beta+\delta}(\car_R)}{\K^\beta(\car_R)}
&
= \log\Esp_{\Q(\car_R)} \( \exp\left( -\delta \G(\cdot,\car_R) \right) \)
\\ & 
\leq
\frac{2 |\delta|}{\beta} \log \Esp_{  \Q(\car_R)} \( \exp\left( \hal \beta \G(\cdot, \car_R) \right) \) 
\\ &
\le C |\delta |\chi(\beta) R^\d,
\end{align*}
using H\"older's inequality and~\eqref{locallawint}. 
Dividing by $\beta R^\d$ and sending $R\to\infty$ yields~\eqref{boundfp}.  
\end{proof}

The proof of Theorem \ref{th1} is now complete.

\smallskip

We may scale the formula~\eqref{formprop} to obtain the limit for any uniform density: 
we have  if $Q \in \mathcal{Q}_R$  and $Q'=m^{\frac{1}{\d}}Q$.
$$\G(\XN, m,  Q)= \begin{cases}
m^{1-\frac2\d}\G(m^{\frac1\d}\XN, 1, Q')  & \text{if} \ \d\ge 3, \\
\G(m^{\frac1\d}\XN, 1, Q') - \frac{m|Q|}{4} \log m & \text{if} \ \d=2.\end{cases}$$
Thus, highlighting the $\beta$ dependence, we have that 
$$\K^\beta(Q, m)= m^{-m|Q|} \K^{\beta m^{1-\frac{2}{\d}} }( Q', 1) \exp\(\frac{\beta}{4}|Q|m\log m \indic_{\d=2}\).
$$
It follows that 
$$ \frac{ \log \K^\beta(Q, m)}{\beta |Q|}= m^{2-\frac2\d} \frac{\log  \K^{\beta m^{1-\frac{2}{\d}  } }( Q', 1)}{\beta m^{1-\frac2\d} |Q'|} + \frac{1}{\beta }\(  -m\log m + \frac{\beta}{4}m\log m \indic_{\d=2}\).$$
Using the result~\eqref{formprop}, we deduce that 
\begin{align}
\label{limcasm}
\frac{ \log \K^\beta(Q_R, m)}{\beta |Q_R|}
&
= -m^{2-\frac2\d} \mf(\beta m^{1-\frac{2}{\d}}) -\frac{m}{\beta} \log m + \frac{1}{4}m(\log m) \indic_{\d=2} 
\\ & \qquad \notag 
+ O\(  
 \frac{\chi(\beta) }{R} \( \rb + \beta^{-\frac1\d} \chi(\beta)^{-\frac1\d}\log^{\frac1\d} \frac{R}{\rb}\)  \),
\end{align} 
where the implicit constant in the~$O(\cdot)$ depends only on $\d$ and $m$.
  
 \subsection{Case of a varying $\mu$}
 In \cite{s2} we will obtain precise expansions for the expansion of $\log \K$ when $\mu$ varies, however  
 in preparation for Theorem \ref{theoldp},   we give  a first  rougher estimate  that we deduce from~\eqref{limcasm} combined with~\eqref{subad3}. For this, we will need to assume some regularity of $\mu$.
 
 \begin{lem} Assume  $\mu(Q_R)$ is an integer. 
 Let $\overline\mu$ be another measure with~$\overline\mu(Q_R) = \mu(Q_R)$ and assume that both $\mu$ and $\overline\mu$ have densities bounded below by $m$ and above by $\Lambda $. Then there exists $C>0$ depending only on $\d, m$ and $\Lambda$ such that  
 \begin{align}
 \label{logkmm} 
 \left| \log \frac{\K(Q_R, \mu)}{\K(Q_R, \bar \mu)}\right|
&
\leq 
C \beta R^{\d+2} \|\mu-\bar \mu\|_{L^\infty(Q_R)}^2 
+C  \|\mu-\bar \mu\|_{L^\infty(Q_R)}
\(  \beta\sqrt{\chi(\beta)} R^{\d+1}  +   R^\d\)
\\ & \qquad \notag  + C \|\mu \|_{C^{0, \kappa}} \sqrt{\chi(\beta) } R^\d + \frac{C}{\beta}
.  
\end{align}
 \end{lem}
 \begin{proof} Let us denote $\bN=\mu(Q_R)$. 
Let $\Q(Q_R)$ denote the Gibbs measure for the density~$\bar \mu$. We have
 $$  \frac{\K(Q_R, \mu)}{\K(Q_R, \bar \mu)}=\Esp_{\Q(Q_R)} \( \exp \(\beta( \G(X_\bN, \bar \mu)- \G(X_\bN, \mu))+\sum_{i=1}^{\bN} (\log \mu-\log \bar \mu)(x_i)\)\).$$
 Then from~\eqref{minneum} we have  
\begin{equation*}
| \G(X_\bN, \bar \mu, Q_R)- \G(X_\bN, \mu, Q_R) |
\le
\int_{Q_R} |\nab w|^2 + 2 \int_{Q_R} |\nab w||\nab u_{\rrh}|+ \|\mu-\bar \mu\|_{L^\infty} \sum_{i=1}^{\bN} \int |\f_{\rrh_i}|,
\end{equation*}
 where $u$ is the solution to~\eqref{defv} with $\bar \mu$, and $w$ is the mean-zero solution to 
 \be \label{eqw}\left\{\begin{array}{ll}
 -\Delta w= \mu - \bar \mu & \text{in}  \ Q_R\\
 \frac{\partial w}{\partial \nu}=0 & \text{on} \ \partial Q_R.\end{array}\right.
 \ee
 Using~\eqref{14} and since $\rrtt=\rrh$ in this instance, we have 
 $$\int_{Q_R} |\nab u_{\rrh}|^2 \le C \( \G(X_\bN, Q_R, \bar \mu) + C R^\d\)$$
 while testing ~\eqref{eqw} against $w$ and using Poincar\'e's inequality allows to show that 
 $$\int_{Q_R} |\nab w|^2 \le C R^{\d+2}  \|\mu-\bar \mu\|_{L^\infty(Q_R)}^2,$$
 and the third term can be bounded by $ R^\d \|\mu-\bar\mu\|_{L^\infty(Q_R)}$ using that
  \be \label{intf} \int_{\R^\d} |\f_\alpha|\le C \alpha^{2} \ee  with $ C$ depending only on $\d$.
For the $\log \mu$ terms we write 
$$\sum_{i=1}^\bN (\log \mu-\log \bar \mu)(x_i)= \int_{Q_R}(\log \mu-\log \bar \mu)d\bar \mu+ \int_{Q_R} (\log \mu-\log \bar \mu) d\( \sum_{i=1}^\bN \delta_{x_i}- \bar\mu\).$$
Let us denote $\omega_\bN$ for $\sum_{i=1}^\bN \delta_{x_i}- \bar\mu.$
By interpolation between H\"older spaces we have $$\|\omega_\bN\|_{(C^{0, \kappa})^*} \le \|\omega_\bN\|_{(C^{0})^*}^{1-\kappa} \|\omega_\bN\|_{(C^{0, 1})^*}^\kappa \le  C\bN^{1-\kappa}  \|\omega_\bN\|_{(C^{0, 1})^*}^\kappa \le C R^{\d(1-\kappa)}  \|\omega_\bN\|_{(C^{0, 1})^*}^\kappa ,$$ hence 
using the local law~\eqref{loclawphi} we have 
\be\label{prehold}\left|\log \Esp_{\Q(Q_R)}\( \exp \frac{\beta}{C} \|\omega_\bN\|_{(C^{0, \kappa})^*}^2 \)\right|\le C \beta \chi(\beta) R^{2\d}.\ee
Using now that $x\le \ep \beta x^2 + \frac{1}{4\beta \ep}$ we  deduce that  for every $\ep \le \frac1C$ with $C$ from the above inequality, we have 
\begin{multline*}\left|\log \Esp_{\Q(Q_R)}\( \exp\( \int_{Q_R} \varphi \omega_\bN \)\)\right|
\le \log  \Esp_{\Q(Q_R)}\( \exp \( \ep \beta \( \int_{Q_R} \varphi \omega_\bN \)^2 + \frac{1}{4\beta \ep} \) \)\\
\le \frac{1}{4\beta \ep} + C \ep \beta \chi(\beta)R^{2\d} \|\varphi\|_{C^{0, \kappa}}^2  \end{multline*}
where we have used H\"older's inequality and~\eqref{prehold}.
Optimizing over $\ep \le 1$ and applying to $\varphi = \log \mu - \log \bar \mu$, we deduce that 
$$\left|\log \Esp_{\Q(Q_R)}\( \exp\( \int_{Q_R} (\log \mu - \log \bar \mu)  \omega_\bN \)\)\right|\le C \|\mu \|_{C^{0, \kappa}} \sqrt{\chi(\beta) } R^\d + \frac{C}{\beta}.$$
It follows that 
$$
\left|\log \Esp_{\Q(Q_R)}\( \exp\( \sum_{i=1}^N (\log \mu-\log \bar \mu)(x_i)\)\)\right|
\le CR^\d \|\mu-\bar \mu\|_{L^\infty} + 
 C \|\mu \|_{C^{0, \kappa}} \sqrt{\chi(\beta) } R^\d + \frac{C}{\beta}.$$
 Combining these estimates with the local law~\eqref{locallawint} we deduce that for every $\lambda $, we have 
\begin{align*}
\left|\log \frac{\K(Q_R, \mu)}{\K(Q_R, \bar \mu)}\right|
& \leq  
C\lambda \beta \chi(\beta) R^\d +( \frac{C}{\lambda}+1) \beta R^{\d+2}   \|\mu-\bar \mu\|_{L^\infty(Q_R)}^2+  C\beta R^\d \|\mu-\bar \mu\|_{L^\infty} 
\\ & \qquad 
+CR^\d \|\mu-\bar \mu\|_{L^\infty} + C \|\mu \|_{C^{0, \kappa}} \sqrt{\chi(\beta) } R^\d + \frac{C}{\beta} .
\end{align*}
Optimizing over $\lambda $ yields the result. 
\end{proof}

 \begin{prop}\label{pro65} 
 Assume $\|\mu\|_{ C^{0,\kappa}} \le C N^{-\frac \kappa\d}$ for some $\kappa>0$,  and $R \gg \rb$ as $N \to \infty$.
 Then, as $N\to \infty$, 
   \begin{equation}\label{619}
\log  \K(Q_R, \mu) = - \beta \int_{Q_R}\mu^{2-\frac2\d} \mf(\beta \mu^{1-\frac2\d})+ \( \frac{\beta}{4} \indic_{\d=2} -1\) \int_{Q_R} \mu \log \mu  +o\((1+\beta)R^\d\),
\end{equation} 
where the term~$o(\cdot)$ on the right side is independent of $\beta$.
 \end{prop}
 \begin{proof}
For any $r \in [\rb, R]$,  we may partition $Q_R$ into cubes $Q_i$ belonging to $\mathcal Q_r$.
 In view of~\eqref{subad4} we obtain 
 $$\log \K(Q_R,\mu)= \sum_{i=1}^p\log \K(Q_i, \mu) + O\(\beta \chi(\beta) \frac{R^\d}{r} \(\rb +  \beta^{-\frac1\d}\chi(\beta)^{-\frac1\d} \log^{\frac1\d} \frac{r}{\rb} \) \).$$  Using~\eqref{logkmm}, letting $\bar \mu_i$ denote the average of $\mu$ in $Q_i$, we then obtain
 \begin{multline*}
 \log \K(Q_R, \mu) = \sum_{i=1}^p \log \K(Q_i, \bar \mu_i) +  O
 \(\beta \chi(\beta) \frac{R^\d}{r} \(\rb +  \beta^{-\frac1\d}\chi(\beta)^{-\frac1\d} \log^{\frac1\d} \frac{r}{\rb} \)   \)
 \\ + 
O\(  \beta R^{\d} \( r^{2+2\kappa} N^{-\frac{2\kappa}{\d}} 
+C   N^{-\frac{\kappa}{\d}} 
\(  \sqrt{\chi(\beta)} r^{\kappa+1}  +   \frac{r^{\kappa} N^{-\frac{\kappa}{\d}} }{\beta}\) \)+ \frac{1}{\beta}\)
\end{multline*}
For $ R \gg \rb$ we have $\frac{1}{\beta} \ll R^{\d}$ hence we check that we may choose $ \rb\ll r \ll R$ such that the right-hand side errors are $o((1+\beta)R^\d)$.
 Inserting  also~\eqref{limcasm}  and using the Lipschitz bound on $\mf$~\eqref{boundfp}, we obtain ~\eqref{619}.
\end{proof}

\section{The Large Deviations Principle: the proof of Theorem \ref{theoldp}}\label{sec9}
First we note that the assumption 
$\dist(x, \p \, \supp \, \muv) \ge C \theta^{-1/2}$ and the fact that $\mu_V\in C^{0,\kappa}$ ensure in view of \cite{ascomp} that $\mut$ is also in $C^{0, \kappa}$ in $\car_R(x)$. Translated to the blown up scale, this gives us a bound by 
$CN^{-\kappa/\d}$ for the $C^{0, \kappa}$ norm of $\mu=\mut'$  so that we may apply Proposition \ref{pro65}.
Since we assumed that $R\ll N^{\frac1\d}$ this also implies that, as $N\to \infty$,
 \be\label{mmuv0}
\|\mu-\mu_V(x_0) \|_{L^\infty(\car_{2R}(N^{1/\d}x_0))}\le o(1).\ee

We consider $P$ a probability measure on infinite point configurations, stationary, with intensity $\mu_V(x_0)$, and $B(P, \ep)$ a ball for some distance that metrizes the weak topology.
By exponential tightness (see \cite[Sec. 4]{lebles}) it suffices to prove a weak LDP, i.e. relative to balls $B(P, \ep)$.

We thus focus on  proving upper and lower bounds on $\log \mathfrak{P}_{N,\beta}^{x_0,R}(B(P,\ep)) $.
For simplicity, let us denote $\car_R$ for $\car_R(N^{1/\d}x_0)$.

\smallskip

\noindent
{\bf Step 1: reducing to good number of points and good energy.}
Since $R$ is large enough, we may include $\car_R$ in a hyperrectangle $Q_R\in \mathcal{Q}_R$ such that 
$|Q_R|-|\car_R|=O(R^{\d-1})=o(R^\d)$.

Let us denote by $n$ the number of points a configuration has in $Q_R$ and by $\mathrm{n}=\mu(Q_R)$ which is an integer. 
Since we assume $R \gg \rb \ge C\max\( \beta^{-\hal} \chi(\beta)^{\hal},1\)$, for $\sigma$ small enough we have $R^{2-3\sigma}\ge \chi(\beta)$ hence in view of the local law~\eqref{loclawpoints00} and~\eqref{fini}    we may write  that  for some $\sigma>0$
\be 
\PNbeta \( |n-\mathrm{n}|\ge R^{\d-\sigma}\) \le \exp\left( -C\beta  R^{\d+1} \right)
\ee
and 
\be \PNbeta\( \sup_x  \int_{\car_{R^{1+\sigma/\d }} } |\nab u_{\rrtt}|^2 \ge C \chi(\beta) R^{\d+\sigma} \) \le \exp\left( -\chi(\beta) \beta R^{\d+\sigma} \right)
\ee
for some $C$  large enough independent of $R$ and $\beta$.
Hence we may restrict the study to the event 
$$\mathcal B=\left\{ |n-\mathrm{n}|\le R^{\d-\sigma},\  \sup_x  \int_{\car_{R^{1+\sigma/\d}}} |\nab u_{\rrtt}|^2 \le  \chi(\beta)R^{\d+\sigma}\right\},$$ since the complement has a probability which is negligible in the speed we are interested in.

\smallskip

\noindent
{\bf Step 2: upper bound.} We recall that $i_N^{x_0,R}$ is defined in~\eqref{defii}. Using~\eqref{add3} and~\eqref{restri1} (recall $\F=\H_{\R^\d}$), we have
\begin{align*}
\lefteqn{ 
 \mathfrak{P}_{N,\beta}^{x_0,R}(B(P, \ep)\cap i_N^{x_0,R}(\mathcal B) )
 } \quad & 
\\ &
= \frac{1}{N^N \K(\R^\d)} \int_{i_N^{x_0,R}  (\XN)\in B(P, \ep) \cap \mathcal B} \exp\left( -\beta \F(\XN,\R^\d) \right)\, d\mu^{\otimes N}(\XN)
\\ & 
\leq  
\frac{1}{N^N \K(\R^\d)}  \int_{i_N^{x_0, R}  (\XN)\in B(P, \ep)  \cap \mathcal B} \exp\left(-\beta \F(X_N|_{Q_R}, Q_R)
-\beta \F(\XN|_{Q_R^c}, Q_R^c ) \right)\,  d\mu^{\otimes N}(\XN).
\end{align*}
Splitting up the events as in the proof of Proposition \ref{proloclaw} with $n$ being the number of points of the configuration which belong to $Q_R$, and using that   $i_N^{x_0,R}(\XN)$ depends only on the configuration in $\car_R$ hence in $Q_R$, we may then write 
\begin{align}
\label{fac1}
\lefteqn{
\mathfrak{P}_{N,\beta}^{x_0, R}(B(P,\ep) \cap i_N^{x_0,R}(\mathcal B))
} \quad & 
\\ & \notag
\le \frac{1}{N^N \K(\R^\d)} \sum_{n=\mathrm{n}-R^{\d-\sigma} }^{\mathrm{n}+ R^{\d-\sigma}} \frac{N!}{n! (N-n)!} 
 \int_{\mathcal B_{n} \cap  (Q_R^c)^{N-n} } \exp\left( -\beta \F (\cdot, {Q_R^c} ) \right)\, d\mu^{\otimes (N-n)}   
\\ & \qquad \notag
\times 
\int_{i_N^{x_0, R}  (X_n)\in B(P, \ep) } \exp\left( -\beta \F(X_n, Q_R) \right)\,  d\mu^{\otimes n}(X_n),
\end{align}
where $\mathcal B_{n}$ is $\mathcal B$ intersected with the event that $X_N$ has $n$ points in  $Q_R$.

On the one hand,  noting that $\H_{\R^\d}=\F$,~\eqref{eqsulk} applied with $L$ such that $ R \gg L \gg \rb$ and combined with Remark \ref{remerr} yields 
\begin{align*}
\lefteqn{
 \int_{\mathcal B_{n} \cap  (Q_R^c)^{N-n} }\exp\(-\beta \F (\cdot, {Q_R^c} ) \) d\mu^{\otimes (N-n)}   
} \qquad & 
\\ & 
\leq  \frac{(N-n)! (N-\mn)^{N-\mn}}{(N-\mn)!} \K(Q_R^c) \exp\( C (\beta \chi(\beta) +1) o(R^{\d})  \),
\end{align*}
with $C$ independent of $\beta$.

On the other hand, Proposition 2.4 in \cite{loiloc} (stated there for dimension 2 but extends with no change to general dimension) itself relying on 
\cite[Theorem 3.1]{GZ93}
states that\footnote{in fact the factor $\frac{1}{n!} $ was missing in  \cite{lebles,loiloc}} if $m = \lim_{R\to \infty}\frac{n}{R^\d}$ then 
$$\lim_{\ep\to 0} \lim_{R\to \infty} \frac{1}{R^\d } \log \frac{1}{n!} \mathcal{L}^{\otimes n}  \{ i_N^{x_0, R}(X_n) \in B(P, \ep)\}
= - (\mathsf{ent}[P|\Pi^m]-m+m\log m).$$
Therefore, we have 
$$\lim_{\ep\to 0} \lim_{R\to \infty} \frac{1}{R^\d } \log \frac{1}{n!} \mu^{\otimes n}  \{ i_N^{x_0, R}(X_n) \in B(P, \ep)\}
= - (\mathsf{ent}[P|\Pi^{m}]- m ) $$
with $m= \mu_V(x_0)$,  in view of~\eqref{mmuv0}, the fact that $\mn= \mu(Q_R)$ and $|n-\mn|=o(R^\d)$.
In what follows we continue to denote $m$ for either $\mu_V(x_0)$ or a generic point density (not to be confused with the lower bound for $\mu$ that we had been  using so far in  the paper).

Moreover, the lower semi-continuity of the energy and the characterization of $\F$ by~\eqref{mindiri}, the fact that $|Q_R|\backslash |\car_R|=o(R^\d)$  ensure,  see for instance \cite[Prop. 5.2]{ps} or proof of Proposition 4.6 in \cite{loiloc},  that  if $ i_N^{x_0, R}(X_n) \in B(P, \ep)$ then 
$$\liminf_{N\to \infty}\frac{1}{R^\d} \F(X_n, Q_R) \ge \hal\widetilde{ \mathbb{W}}^{m} (P) -o_\ep(1).$$ 
Combining these facts and inserting them into~\eqref{fac1} leads to 
\begin{multline}\label{lad2}
 \mathfrak{P}_{N,\beta}^{x_0,R}(B(P,\ep)\cap i_N^{x_0, R}(\mathcal B))\\
 \le \exp\(  - R^\d  \(\frac{\beta}{2} \widetilde{\mathbb{W}}^{m}  (P) +   \mathsf{ent}[P|\Pi^m]-m+(1+\beta) o_{\ep,N}(1)+ C \beta \chi(\beta) R^{-\sigma}     \)\)    \\ 
 \times \frac{1}{N^N \K(\R^\d)}   \sum_{n=\mathrm{n}-R^{\d-\sigma} }^{\mathrm{n}+ R^{\d-\sigma}} \frac{N!}{ (N-\mn)!} 
     (N-\mn)^{N-\mn} \K(Q_R^c)
.\end{multline}
On the other hand using~\eqref{superad2}, we have 
$$\K(\R^\d) \ge \frac{N! N^{-N}}{\mathrm{n}!(N-\mathrm{n})!  \mn^{-\mn}(N-\mn)^{-(N-\mn)}} \K(Q_R)\K (Q_R^c),$$
 and inserting this into~\eqref{lad2},
we find 
\begin{align*}
\lefteqn{
\mathfrak{P}_{N,\beta}^{x_0,R}(B(P,\ep) \cap i_N^{x_0,R}(\mathcal B)) 
} \quad & 
\\ & 
\leq
R^{\d-\sigma}  \exp\(  - R^\d  \(\frac{\beta}{2} \widetilde{\mathbb{W}}^{m}  (P) +   \mathsf{ent}[P|\Pi^m]-m+(1+\beta)o_{\ep,N}(1)\)\)     \K(Q_R)^{-1} \mn! \mn^{-\mn}
\\ & 
\leq
\exp\(   - R^\d  \(\frac{\beta}{2} \widetilde{\mathbb{W}}^{m}  (P) +   \mathsf{ent}[P|\Pi^m]-m+(1+\beta) o_{\ep, N} (1) \)\)\K(Q_R)^{-1} \exp\left( -\mn+o(\mn) \right),
\end{align*}
where we used Stirling's formula and $ R \gg \rb$.
To estimate $\K(Q_R)$ we use~\eqref{619} and the Lipschitz bound on $\mf$ to write, using again~\eqref{mmuv0}
 \begin{multline}\label{620}
 \log \K(Q_R) = - \beta |Q_R|m^{2-\frac2\d} \mf(\beta m^{1-\frac2\d})  +\(\frac{\beta}{4} \indic_{\d=2}-1\)|Q_R| m \log m +o((1+\beta)R^\d).\end{multline}
 Since $\mn= m |Q_R|+o(R^\d)$,  we obtain \begin{multline}\label{bornesupldp}
\log \mathfrak{P}_{N,\beta}^{x_0, R}(B(P,\ep))
\\ \le   - R^\d  \(\frac{\beta}{2} \widetilde{\mathbb{W}}^{m}  (P) +   \mathsf{ent}[P|\Pi^m] -  m^{2-\frac2\d} \beta  \mf(\beta m^{1-\frac{2}{\d}} )+ \(\frac\beta{4}  \indic_{\d=2} -1\) m\log m \)\\
  +(1+\beta) o_{\ep, N} (R^\d) ,
\end{multline}with $m=\mu_V(x_0)$,
which concludes the upper bound.

\smallskip

\noindent
{\bf Step 3: lower bound.}
Retranscribed in our notation,  \cite[Lemma 5.1]{loiloc} shows that given any $P $ such that 
$ \widetilde{\mathbb{W}}^{m}  (P) +   \mathsf{ent}[P|\Pi^m]$ is finite, we can construct a family $A$ of  configurations $X_{\mn}$ of $\mn$ points in $Q_R$ such that $i_N^{x_0,R}(X_{\mn}) \in B(P, \ep)$ and 
\be \G(X_{\mn},Q_R) \le R^\d \frac{\widetilde{ \mathbb{W}}^{m}(P) }{2}+o(R^\d)\ee
uniformly on $A$, and 
\be \log \(\frac{1}{\mn!}\mathcal{L}^{\otimes \mn} (A)\) \ge - R^\d\( \mathsf{ent}[P|\Pi^{m} ] -  m+ m\log m\) +o(R^\d)\ee
Applying this with $m =\mu_V(x_0)$, we may thus write with the help of~\eqref{subad1}
\begin{align*}
\lefteqn{
\mathfrak{P}_{N,\beta}^{x_0,R}(B(P,\ep))
} \quad & 
\\ & 
= \frac{1}{N^N \K(\R^\d)} \int_{i_N^{x_0,R}(\XN) \in B(P, \ep) } \exp\left( -\beta \G(\XN,\R^\d) \right)\, d\mu^{\otimes N}(X_N)
\\ & 
\ge \frac{1}{N^N \K(\R^\d) } \frac{N!}{\mn! (N-\mn)! }\int_{(Q_R^c)^{N-\mn}}  \exp\left( -\beta \G(\cdot,Q_R^c) \right) d\mu^{\otimes (N-\mn)}(X_N)
\\ & \qquad \times \int_{A } \exp\left( -\beta \G(X_\mn, Q_R) \right)\, d\mu^{\otimes {\mn}}(X_\mn)
 \\ & 
 = \frac{1}{N^N \K(\R^\d) } \frac{N!}{\mn! (N-\mn)!(N-\mn)^{-(N-\mn)} }\K(Q_R^c) \int_{A } \exp\left( -\beta \G(X_\mn, Q_R) \right)\, d\mu^{\otimes {\mn}}(X_\mn).
\end{align*}
But in view of  ~\eqref{subad3} we have 
$$\log \K(\R^\d)= \log \K (Q_R)+ \log \K(Q_R^c)+o((1+\beta \chi(\beta) ) R^{\d})$$
so we find, using also Stirling's formula, that 
\begin{multline*}
\log \mathfrak{P}_{N,\beta}^{x_0,R}(B(P,\ep))
\\
\ge  -\mn - \log  \K (Q_R) -\frac{ \beta}{2} R^\d \widetilde{\mathbb{W}}^{m}(P) +\beta o(R^\d)
-  R^\d\( \mathsf{ent}[P|\Pi^{m} ] -  m\) +o(R^\d).\end{multline*}
Inserting~\eqref{620} to estimate $\K (Q_R)$, we obtain
\begin{align}
\label{borneinfldp}
\lefteqn{
\log \mathfrak{P}_{N,\beta}^{x_0,R}(B(P,\ep))
} \quad & 
\\ & \notag
\geq 
- R^\d\(\frac{   \beta}{2} \widetilde{\mathbb{W}}^{m}(P)  + \mathsf{ent}[P|\Pi^{m} ]  -  \beta m^{2-\frac2\d} \mf(\beta m^{1-\frac{2}{\d}}) + \(\frac{\beta}{4}\indic_{\d=2}- 1\) m\log m\) 
\\ & \qquad  \notag
+(1+\beta) o_{\ep, N}(R^\d).
\end{align}
Applying this to $P$ a minimizer of $\beta  \widetilde{\mathbb{W}}^{m}(P)  + \mathsf{ent}[P|\Pi^{m} ] $ we deduce that 
$$ \inf_P \( \frac{\beta }{2}\widetilde{\mathbb{W}}^{m}(P)  + \mathsf{ent}[P|\Pi^{m} ]  \) \ge   \beta m^{2-\frac2\d} \mf(\beta m^{1-\frac{2}{\d}}) +\( 1-\frac{\beta}{4}\indic_{\d=2} \)m\log m,$$ with $m=\mu_V(x_0)$.
We may write this for any $m$,  thus deducing that 
\be\label{inff}
 \inf \mathcal{F}_\beta^1 \ge \beta  \mf(\beta).\ee

\smallskip

\noindent
{\bf Step 4: conclusion}.
By exponential tightness (see  \cite{lebles}), we then upgrade the conclusions of the previous steps to a strong LDP result: for any Borel set $E$, it holds that, as~$N\to \infty$, 
\begin{align}
\label{concldp1}
\lefteqn{
\log \mathfrak{P}_{N,\beta}^{x_0,R}(E)
} \ \ & 
\\ & \notag
\leq 
-\inf_{\bar E}  R^\d\(   \frac{\beta }{2} \widetilde{\mathbb{W}}^{m}(P)  + \mathsf{ent}[P|\Pi^{m} ]  -  \beta m^{2-\frac2\d} \mf(\beta m^{1-\frac{2}{\d}}) +\( \frac{\beta}{4}\indic_{\d=2}-1\) m\log m\) 
\\ & \notag 
\quad  +(1+\beta) o(R^\d)
\end{align}
and
\begin{align}
\label{concldp2}
\lefteqn{
\log \mathfrak{P}_{N,\beta}^{x_0,R}(E)
} \ \ & 
\\ & \notag
\geq 
-\inf_{\overset{\circ}{E}}  R^\d\(   \frac{\beta}{2} \widetilde{\mathbb{W}}^{m}(P)  + \mathsf{ent}[P|\Pi^{m} ]  -  \beta m^{2-\frac2\d} \mf(\beta m^{1-\frac{2}{\d}}) 
+ \( \frac{\beta}{4}\indic_{\d=2}-1\) m\log m\)
\\ & \notag 
\quad  +(1+\beta) o(R^\d)
\end{align}
Applying this relation to $E$ equal the whole space, we find 
$$ -  \inf \(  \frac{ \beta}{2} \widetilde{\mathbb{W}}^{m}(P)  + \mathsf{ent}[P|\Pi^{m} ]  -  \beta m^{2-\frac2\d} \mf(\beta m^{1-\frac{2}{\d}}) + \frac{\beta}{4} m\log m\indic_{\d=2} -m\log m\)\ge 0.$$
When  $m=1$ we find $\inf \mathcal{F}_\beta^1 \le \beta  \mf(\beta)$, which with~\eqref{inff} allows to prove the claim (which already follows from the result of \cite{lebles}) that $\min \mathcal{F}_\beta^1=\beta \mf(\beta)$.
With the scaling properties of $\widetilde{\mathbb{W}}^m$ and $\mathsf{ent}[ \cdot|\Pi^m]$ with respect to $m$ (see \cite{lebles}), we deduce that 
\be\inf \mathcal{F}_\beta^m= \beta m^{2-\frac2\d} \mf(\beta m^{1-\frac{2}{\d}}) +\(1-\frac{\beta}{4} \indic_{\d=2} \) m\log m.\ee
Inserting into~\eqref{concldp1} and~\eqref{concldp2}, the stated LDP result follows if $\beta$ is fixed. The generalization to $\beta \to 0$ or $\beta\to \infty$ is straightforward from  ~\eqref{bornesupldp} and~\eqref{borneinfldp}.
 This concludes the proof of Theorem \ref{theoldp}.

\section{The case of energy minimizers}
\label{secmini}

To consider energy minimizers, we  define an  analogous quantity to the partition function
\be\label{kmini}
\K^\infty(U,\mu)=\min_{X_\bN} \G(X_\bN, U, \mu),\ee with $\bN=\mu(U)$.
In view of ~\eqref{subad1} we have that if $U$ is partitioned into regions $Q_i$, 
with  $\mu(Q_i)=N_i$   integer, then
\be\label{subad6} 
\K^\infty(U,\mu) \le \sum_{i=1}^p \K^\infty(Q_i,\mu).\ee

We have easy a priori bounds: if $\bN=\mu(U)$ 
\be\label{apriorib}
 - C \bN \le  \K^\infty(U,\mu)\le C\bN,\ee with $C>0$ depending only on  $\d,m$ and $\Lambda$.
 Indeed, the lower bound follows from~\eqref{14}, while for the upper bound, we may deduce from ~\eqref{minok} applied with $\beta =1$ that there exists  at least one $X_\bN\in U^{\bN}$ such that $\G(X_\bN, U) \le C \bN$ for some $C$ large enough.

\begin{theo}
\label{th4}
\begin{enumerate}
\item  (Neumann problems in cubes)
Let $\car_R$ be a cube of size $R$ with $R^\d=\bN$ an integer.
We have 
\be\label{mn1}
\left|\frac{\K^\infty(\car_R,1)}{R^\d}-\mf(\infty)\right|\le \frac{C}{R}\ee
where $\mf(\infty)=\hal\min\widetilde{ \mathbb{W}^1}= \lim_{\beta \to \infty} \mf(\beta)$ and $C>0$ depends only on $\d$.
Moreover, if $X_\bN$ is a minimizer for $\K^\infty(\car_R,1)$, for any cube $\car_\ell(x)\subset \car_R$, we have 
\be \label{mn2}
\left|\int_{\car_\ell(x)}( \sum_{i=1}^{\bN} \delta_{x_i} - 1)\right|\le C \l^{\d-1}\ee
and  the energy is uniformly distributed in the sense that 
\be \label{mn3}\G^{\car_\ell(x)}(X_\bN, \car_R,1)= \ell^\d \mf(\infty)+ O(\ell^{\d-1}).\ee

\item
 (Minimizers of the Coulomb gas energy). 
Assume that the equilibrium measure $\mu_V$  satisfies $m\le \mu_V\le \Lambda $ on its support and $\mu_V \in C^{0,\kappa}$ on its support, for some $\kappa>0$. If $\XN$ minimizes $\HN$,  
if  $\car_R(x)$ is a cube of size $R$ centered at $x$ satisfying $$\dist(\car_R(x) , \p \,\supp\, \mu_V) \ge     C N^{ \frac{-2}{\d(\d+2)} }    $$
we have
\begin{equation}\label{loclawpointsmin}
\left| \int_{\car_R(x)}  \sum_{i=1}^N \delta_{x_i}-  N \int_{\car_R(x)} d\mu_V \right|\le C\(N^{\frac{1}{\d}} R\)^{\d-1},\end{equation}
and 
\be\label{mn4} \G^{\car_R(x)}(\XN',  \mu_V')=  \mf(\infty) \int_{\car_R(x)}(\mu_V')^{2-\frac2\d} -  \frac{1}{4} \indic_{\d=2} \int_{\car_R(x)} \mu_V' \log \mu_V'+o(R^{\d}),\ee
where  $C$ and $o$ depend only on $\d,m$ and $\Lambda$.
\end{enumerate}
\end{theo}
\begin{remark} The explicit rate in~\eqref{mn3} is the improvement compared to \cite{rns,PRN}, in the same way~\eqref{mn4} can be improved, see \cite{s2}. As in \cite{rns}, we can also prove with the same method the same results on minimizers and the minimum of the renormalized energy ${\mathbb{W}}^1$ of \cite{rs,ps}. For instance the limit as $R \to \infty$ that defines ${\mathbb{W}}^1$ can be shown to be $\mf(\infty)$ with rate $1/R$: the upper bound is by periodization of a minimizer for $\K^\infty$  while the lower bound is obtained as in ~\eqref{minofk} to be combined with~\eqref{mn1}.\end{remark}
\begin{proof}
{\bf Step 1: bootstrap.}
Let $\mu $ satisfy $0<m\le \mu \le \Lambda $ in $\Sigma$, let $X_\bN^0$ be a minimizer of  $\G(\cdot, U)$ among configurations with $\bN$ points.
We claim that if $\car_R(x)$ satisfies the same assumptions as in Proposition \ref{proloclaw}, in particular~\eqref{conddist} with $\beta=\infty$    and if  $R$ is large enough, then
\be \label{aprouve} \G^{\car_R(x)} (X_\bN^0, U)+ C_0 \#( \{X_\bN^0\} \cap \car_R(x)) \le C R^\d \ee  for some $C$ depending only on $\d$ and $\mu$.
This is proven by a bootstrap: assume this is true for some $L$ i.e. 
assume 
\be\label{hypboot}
 \G^{\car_L(x)} (X_\bN^0, U)+ C_0 \#( \{X_\bN^0\} \cap \car_L(x)) \le \mathcal{C} L^\d ,\ee 
 we need to show it is true for $R \ge L/2$.
Let us proceed as in the proof of Proposition \ref{proloclaw}, reducing to $Q_R\in \mathcal{Q}_R$ and denoting by 
$n=  \#( \{X_\bN^0\} \cap \car_R(x))$ and $\mn= \mu(Q_R\cap U)$.
First by~\eqref{hypboot} and the choice of $C_0$, we have from
\eqref{14} and~\eqref{disc10}--\eqref{disc1} that 
\be \label{mn}  |n-\mn| \le C R^{\d-1} + C\sqrt{\mathcal{C}} R^{\d-\hal} .\ee
We then apply  Proposition \ref{outscreen} with $S(X_\bN) \le \mathcal C L^\d$, $\tilde \ell= M  L^\d R^{- \frac{\d(\d+1)}{\d+2}} $, $\ell =   R^{ \frac{\d}{\d+2}}  $ and $Z_{\mn-\N}$ minimizing $\G(\cdot, \tilde \mu, \N)$ (recall that that minimum is bounded by the order of the volume, see~\eqref{apriorib}).
We may check that as soon as  $M$ is large enough and $R$ is larger than some constant depending only on $\d$ and $M$, $\ell \le \tilde \ell \le R$ and~\eqref{constr2} is satisfied. 
 The proposition yields  in view of~\eqref{mn} and~\eqref{hypboot}
\be \label{111}\K^\infty(Q_R^c) \le \H_U(X_\bN^0|_{Q_R^c},Q_R^c)  +  C\( \frac{  \mathcal C}{M} R^\d 
+  R^{\d-1} + \sqrt{\mathcal{C}} R^{\d-\hal}  \)
.\ee 
Choosing $M$ large enough, combining ~\eqref{restri1},~\eqref{subad6},~\eqref{apriorib} and~\eqref{111}, it follows that \begin{multline*}
\G^{Q_R} (X_\bN^0, U)+ \H_U(X_\bN^0|_{Q_R^c}, Q_R^c)\le \G(X_\bN^0, U)= \K^\infty(U)  \le \K^\infty(Q_R)+ \K^\infty(Q_R^c) \\
\le \K^\infty(Q_R)+  \H_U(X_\bN^0|_{Q_R^c}, Q_R^c)   +\hal \mathcal C R^\d.
\end{multline*} 
Hence if $R $ is large enough (depending on $\mathcal{C}$), we have 
$$ \G^{Q_R} (X_\bN^0, U)\le  \K^\infty(Q_R)+ \hal \mathcal{C} R^\d.$$
In view of~\eqref{mn} we have as well  $n\le \hal \mathcal{C} R^\d$.
With~\eqref{apriorib} this concludes the proof of~\eqref{aprouve}.
\smallskip

\noindent
{\bf Step 2: local laws.}
Now that we know~\eqref{aprouve} down to scale $C$, we can use it to control $|\mn-n|$ by $ CR^{\d-1}$ with~\eqref{disc10}--\eqref{disc1} and then return to~\eqref{111} and upgrade it to have an error $R^{\d-1}$, i.e. we find 
$$  \G^{Q_R} (X_\bN^0, U) \le \K^\infty(Q_R)+ CR^{\d-1}$$
and $|n-\mn|\le C R^{\d-1}$.
By Proposition \ref{outscreen} we also have
\be \label{minofk}
\K^\infty(Q_R)\le \G^{Q_R} (X_\bN^0, U)+ CR^{\d-1},\ee  so 
\be \label{815}
\G^{Q_R} (X_\bN^0, U,\mu)= \K^\infty(Q_R,\mu)+O(R^{\d-1}),\ee
with the $O$ depending only on $\d,m$ and $\Lambda$.
\smallskip 

\noindent
{\bf Step 3:  Energy expansion.} 
We may use the well-known characterization  
$$- \frac{ \log \K^\beta(Q_R)}{\beta}= \min_{P \in \mathcal{P}(Q_R^\bN)}  \int\G(X_\bN, Q_R) dP(X_\bN) + \frac{1}{\beta} \int P(X_\bN) \log P(X_\bN)dX_\bN$$
to write that  for each fixed $\bN$
$$ \lim_{\beta \to \infty} - \frac{\log \K^\beta(Q_R)}{\beta}= \min_{X_\bN}\G(X_\bN, Q_R)= \K^\infty(Q_R).$$
We may thus  compute $\K^\infty(Q_R,1)$ via ~\eqref{formprop} and find 
$$\K^\infty(Q_R,1)= |Q_R| \lim_{\beta\to \infty} \mf(\beta)+O(R^{\d-1})$$ where the limit exists in view of the form~\eqref{fb}, and is equal to $\hal \min \widetilde {\mathbb{W}^1}$. In the case of general $\mu$ we find from~\eqref{619} that  if $\|\mu \|_{C^{0,\kappa}}\le C N^{-\frac{\kappa}{\d}}$ then
 \begin{equation}\label{kinftyvar}
  \K^\infty(Q_R, \mu) =   \mf(\infty) \int_{Q_R}\mu^{2-\frac2\d} -  \frac{1}{4} \indic_{\d=2} \int_{Q_R} \mu \log \mu  +o(R^\d)\end{equation} as $N \to \infty$.

\noindent
{\bf Step 4: conclusion}.
The relation~\eqref{mn1} has been proved.~\eqref{mn3} follows from~\eqref{815} applied with $U=Q_R$  and $\mu=1$, and~\eqref{mn2} follows from~\eqref{disc10} and~\eqref{disc1} combined with~\eqref{14}. 

We now turn to the proof of~(2). ~\eqref{mn4} is  a consequence of~\eqref{kinftyvar} and~\eqref{mn1} applied with $U=\R^\d$, $\mu= \mu_V'$, and then a blow-down,  and~\eqref{loclawpointsmin} follows from~\eqref{disc10}--\eqref{disc1} combined with~\eqref{14}.
\end{proof}

\appendix

\section{Estimates on Green functions}
\label{appendixgreen}
In this appendix we prove the following estimate on the Neumann Green functions of a domain. (It may be known but we were not able to locate it in the literature.)
\begin{prop}\label{proappgreen}
Let $U$ be a Lipschitz domain (bounded or unbounded).
Let $G_U$ be the Neumann Green function relative to $U$ with background $\mu$ ($\int_U \mu=1$) i.e. solving 
$$\left\{\begin{array}{ll}
-\Delta G_U(x,y) = \cd \( \delta_y - \mu \)& \text{in}  \ U\\
\frac{\pa G_U}{\pa \nu} = 0 & \text{on} \ \pa  U \end{array}\right.$$
 Then if 
$\int \g(x-y) d\mu(y) <\infty$, up to addition of a constant to $G_U$ we have  $\int_U G_U(x,y) dx=0$ and 
\be \label{borneGU}
\sup_{x\in U}\left|G_U(x,y)-\g(x-y) + \int_U \g(x-z) d\mu(z) \right|\le C \min \( \max \(\g(\dist(y, \pa U)),1\) , \g(x-y)\)
\ee
where $C$ depends only on $\d$ and the Lipschitz type of $U$.
\end{prop}
\begin{proof}
First the upper bound by $ C\g(x-y)$ is standard (one can also deduce it from integrating in time~\eqref{estpt} below), so there remains to prove the other one.
Let $\Phi_t$ denote the heat kernel in dimension $\d$
$$\Phi_t(x)=  \frac{1}{(4\pi t)^{\frac{\d}{2}} } \exp\( - \frac{|x-y|^2}{4t}\).$$ 
First we claim that  
\be \label{duha} 
G_U(x,y)=\int_0^\infty w(t,x) dt\ee
where 
$w$ solves 
$$\left\{\begin{array}{ll}
\pa_t w -\Delta w  = 0 & \text{in}  \ U\\
w(0,x)= \cd(\delta_y- \mu) & \text{in} \ U\\
\frac{\pa w}{\pa \nu} = 0 & \text{on} \ \pa  U. \end{array}\right.$$
To prove this, it suffices to write that
$$\Delta_x \int_0^\infty w = \int_0^\infty\pa_t w\, dt=-w(0,x)= -\cd( \delta_y- \mu).$$
Thus the Laplacian of both quantities in~\eqref{duha} is the same and so is their normal derivative on the boundary.  The two functions must then coincide up to a constant, which we choose to be $0$.
Let us then set 
\be \label{stu}
u(x,y)= \cd\( \g(x-y)- \int_U \g(x-z) d\mu(z) \).\ee
 Similarly as the previous claim, we may write 
$$ u(x,y)= \g * \( \delta_y - \mu\) = \int_0^\infty  \tilde w(t,x) dt$$ where 
\be\label{deftw}
\tilde w(t,x):= \cd\int_U \Phi_t(x-z)  \( \delta_y - \mu\)(z).\ee
 We thus turn to bounding 
 \be\label{tbd}
 G_U(x,y)- u(x,y)= \int_0^\infty (w(t,x)-\tilde w(t,x))\, dt.\ee
 For that we note that $w-\tilde w:= f $ solves
 \be\label{eqf}\left\{\begin{array}{ll}
\pa_t f -\Delta f  = 0 & \text{in}  \ U\\
f(0,x)= 0& \text{in} \ U\\
\frac{\pa f}{\pa \nu} = - \frac{\pa \tilde w}{\pa \nu} & \text{on} \ \pa  U \end{array}\right.\ee
  We are going to break the integral in~\eqref{tbd} into two pieces, from $0$ to $t_*:= \min (1, \dist^2(y, \pa U))$ and from $t_*$ to  $+\infty$. 
  
\smallskip

\noindent     
{\bf Step 1: the bound on $[0,t_*)$.} Let $T>0$.
We consider the solution of the adjoint equation to~\eqref{eqf}, that is,
\be\label{eqh}
\left\{\begin{array}{ll}
\pa_t h -\Delta h  = 0 & \text{in}  \ U\\
h(0,x)= f(T,x) \eta(x) & \text{in} \ U\\
\frac{\pa h}{\pa \nu} =0 & \text{on} \ \pa  U \end{array}\right.
\ee
 where $\eta$ is a  smooth cutoff function to be specified later satisfying $\int \eta=1$. We may write that 
 \be \label{eqfh}
 h(t,x)= p_t* (f(T,x) \eta)\ee
 where $p_t$ is the Neumann  heat kernel relative to $U$. As can be found in \cite{saloff}, in Lipschitz  domains, we have estimates of the form 
 \be \label{estpt} p_t(x) \le C   t^{-\frac{\d}{2}} \exp\(-C \frac{|x-y|^2}{t}\) \ee
 so that 
 \be\label{56}
 |h(T-t,x)| \le \|  f(T,\cdot) \|_{L^\infty(\supp\, \eta)} \sup_{z\in \supp \, \eta}  \exp\(- C \frac{|x-z|^2}{T-t} \) (T-t)^{-\frac{\d}{2}}.\ee
We then compute  using~\eqref{eqf} and~\eqref{eqh}
\begin{multline*}\pa_t \int_U f(t) h(T-t) = \int_U \Delta f (t) h(T-t) -\int_U f(t) \Delta h(T-t) 
= \int_{\pa U} \frac{\pa f}{\pa \nu}(t) h(T-t) \\=- \int_{\pa U} \frac{\pa \tilde w}{\pa \nu}(t) h(T-t) .  \end{multline*}
Integrating between $t=0$ and $t=T$ and then using~\eqref{56}, it follows that 
\begin{multline*}
\left|\int_U f^2(T,x)\eta(x)\right| =\left| \int_0^T  \int_{\pa U} \frac{\pa \tilde w}{\pa \nu}(t) h(T-t) \, dt\right| \\
\le  \int_0^T  \int_{\pa U}\left| \frac{\pa \tilde w}{\pa \nu}(t)\right|  \|  f(T,\cdot) \|_{L^\infty(\supp\, \eta)} \sup_{z\in \supp \, \eta}  \exp\(- C \frac{|x-z|^2}{T-t} \) (T-t)^{-\frac{\d}{2}}\,dx\, dt.
\end{multline*}
The quantity $\frac{\pa \tilde w}{\pa \nu}(t)$ can be computed explicitly from~\eqref{deftw} which yields
 $$\left\|\frac{\pa \tilde w}{\pa \nu}\right\|_{L^\infty(\pa U)}\le C t^{-\frac{\d}{2} }\frac{\dist(y, \pa U) }{t}\exp\(- \frac{\dist^2(y, \pa U) }{4t}\)  
\le C t^{-\frac{\d}{2} -\hal}\exp\(- \frac{\dist^2(y, \pa U) }{8t}\).$$
 We thus obtain 
\begin{multline*}
\left|\int_U f^2(T,x)\eta(x)\right| 
\le C \|  f(T,\cdot) \|_{L^\infty(\supp\, \eta)} \\
\times   \int_0^T     t^{-\frac{\d}{2} -\hal}\exp\(- \frac{\dist^2(y, \pa U) }{8t}\)   (T-t)^{-\frac{\d}{2}} \int_{\pa U} \sup_{z\in \supp \, \eta}  \exp\(- C \frac{|x-z|^2}{T-t} \) \,dx\, dt.
\end{multline*}
Using the change of variables $x'= x (T-t)^{-\hal}$  and then $s= \frac{\dist^2(y, \pa U)}{t}$ we obtain
\begin{multline*}
\left|\int_U f^2(T,x)\eta(x)\right| \\
\le C \|  f(T,\cdot) \|_{L^\infty(\supp\, \eta)}   \int_0^T     t^{-\frac{\d}{2} -\hal}\exp\(- \frac{\dist^2(y, \pa U) }{8t}\)  dt \ \sup_{z\in \supp\, \eta} 
\int e^{-C \left|x'-\frac{z}{\sqrt{T-t}} \right|^2} dx'
\\ \le  C  \|  f(T,\cdot) \|_{L^\infty(\supp\, \eta)}  \dist(y, \pa U)^{1-\d} \int_{ \frac{\dist^2(y, \pa U)}{T} }^\infty e^{-\frac{s}{8}} s^{\frac{\d-3}{2}}\, ds .
\end{multline*}   
We may write that for some constant $C_2 \ge C_1 >0$ 
\begin{align}
\label{145}
& 
C_1 e^{-\frac{\dist^2(y, \pa U) }{8  T}}  \(\frac{\dist^2(y, \pa U)}{T}\)^{\frac{\d-3}{2}}
\\ & \qquad \notag
\leq
\int_{ \frac{\dist^2(y, \pa U)}{T} }^{ 2\frac{\dist^2(y, \pa U)}{T} }e^{-\frac{s}{8}} s^{\frac{\d-3}{2}}\, ds \le C_2e^{-\frac{\dist^2(y, \pa U) }{8  T}}  \(\frac{\dist^2(y, \pa U)}{T}\)^{\frac{\d-3}{2}}
\end{align}
and, by an integration by parts, 
\begin{multline*}
\int_{ 2\frac{\dist^2(y, \pa U)}{T} }^\infty e^{-\frac{s}{8}} s^{\frac{\d-3}{2}}\, ds =8 e^{-  \frac{\dist^2(y, \pa U)}{4T} } \(2 \frac{\dist^2(y, \pa U)}{T}\)^{\frac{\d-3}{2}} +4(\d-3)\int_{2\frac{\dist^2(y, \pa U)}{T} }^\infty e^{-\frac{s}{8}} s^{\frac{\d-5}{2}}\, ds.
\end{multline*}

If we consider only $T \le \dist^2 (y , \pa U)$, then the last term in the right-hand side can be absorbed into the quantity of~\eqref{145}
and we conclude that 
$$
  \int_{ \frac{\dist^2(y, \pa U)}{T} }^{ \infty }e^{-\frac{s}{8}} s^{\frac{\d-3}{2}}\, ds \le C e^{-\frac{\dist^2(y, \pa U) }{8  T}}  \(\frac{\dist^2(y, \pa U)}{T}\)^{\frac{\d-3}{2}}.$$ Inserting into the above, this implies that  for $T \le t_*$, 
$$
\left|\int_U f^2(T,x)\eta(x)\right| 
\\ \le  C  \|  f(T,\cdot) \|_{L^\infty(\supp\, \eta)}  \dist(y, \pa U)^{-2} e^{-\frac{\dist^2(y, \pa U) }{8  T}}  T^{\frac{3-\d}{2}}
.
$$
Choosing $\eta $ to converge to $\delta_{x_0}$ we deduce that 
$$ |f(T,x_0) |\le C   \dist(y, \pa U)^{-2} e^{-\frac{\dist^2(y, \pa U) }{8  T}}  T^{\frac{3-\d}{2}}.$$
Since this is true for every $t \le t_*$ and every $x_0 \in U$ it follows that 
$$\int_0^{t_*} \|f(t,\cdot)\|_{L^\infty(U)} \, dt \le C\dist(y, \pa U)^{-2} \int_0^{\min(1, \dist^2(y , \pa U)) } 
e^{-\frac{\dist^2(y, \pa U) }{8  t}}  t^{\frac{3-\d}{2}}\, dt.$$
With the change of variables $s= t/\dist^2(y, \pa U)$, we are led to 
$$\int_0^{t_*} \|f(t,\cdot)\|_{L^\infty} \, dt \le  C \dist(y, \pa U)^{3-\d}  .$$
This is $\le \g(\dist(y, \pa U))$ if $\dist(y, \pa U) \le 1$. If $\dist (y, \pa U) \ge 1$ we do not perform the change of variables but instead bound the integral by 
$\int_0^1 \exp\(- \frac{1}{8t}\) t^{\frac{3-\d}{2}}\,dt\le C$ and find 
$\dist(y, \pa  U)^{-2}\le C$.
 We conclude that 
 $$\int_0^{t_*} \|f\|_{L^\infty(U)} \le C \( \max( \g(\dist(y, \pa U)  , 1) \).$$
 {\bf Step 2: Bound on $[t_*,+\infty)$.}
 We use that $\tilde w= \cd\Phi_t * (\delta_y-\mu)$ and $w=\cd p_t * (\delta_y-\mu) $ with $p_t$ the Neumann heat kernel as above that satisfies~\eqref{estpt}. 
 It follows that 
 $$\left|\int_{t_*}^1 \|\tilde w-w\|_{L^\infty(U)} \, dt\right|\le C \int_{t_*}^1 t^{-\frac{\d}{2}}\,dt \le C  \begin{cases} t_* ^{1-\frac\d2} & \text{if } \ \d\ge 3\\ -\log  t_* & \text{if} \ \d=2.\end{cases}$$
On the other hand we may write, with $u$ as in~\eqref{stu}, 
\begin{multline*}
\int_1^\infty \tilde w\, dt = \cd \int_1^\infty \int_U \Phi_t(x-z)\( \delta_y- \mu\) (z) =- \int_1^\infty  \int_{\R^\d} \Phi_t(x-z) \Delta u(z,y)\, dt\\
=- \int_1^\infty \int_{\R^\d} \Delta \Phi_t(x-z) u(z,y) \, dt= -\int_1^\infty \int_{\R^\d} \pa_t \Phi_t(x-z) u(z,y) dt \, dz\\
=\int_{\R^\d}  \Phi_1(x-z) u(z,y) \, dz 
= \frac{1}{(4\pi)^{\frac\d2}} \int_{\R^\d} \exp\(- \frac{|x-z|^2}{4}\)  u(z,y) dz\le C .\end{multline*}
In the same way, we find 
\begin{multline*}
\int_1^\infty  w\, dt =  - \int_1^\infty  \int_{U} p_t(x-z) \Delta G_U (z,y)\, dt= \int_{U}  p_1(x-z) G_U(z,y) dz\\  \le 
C  \int_{U} \exp\(- \frac{|x-z|^2}{4}\)  G_U(z,y) dz\le C,\end{multline*}
by using the bound $G_U(z,y) \le C \g(z-y)$.

Combining all these results and using the definition of $t_*$, 
  it follows that 
  $$\sup_{x\in U} |G_U(x,y)- u(x,y)|\le C ( \max (\g(\dist(y,\pa U)) ,1)) $$ from which we deduce the result.
\end{proof}

\section{Auxiliary results on the energies}\label{appendixa}

We gather in this appendix some results that are similar to \cite{ps,ls2}. The notation is as in Section \ref{sec2}.

\subsection{Monotonicity results}

We need the following result adapted from \cite{ps,ls2} which expresses a monotonicity with respect to the truncation parameter.
\begin{lem}\label{monoto}
Let $u$ solve 
\be\label{eqsu}
-\Delta u= \cd\(\sum_{i=1}^N \delta_{x_i} - \mu \) \quad \text{in} \ U,\ee  and let $u_{\vec{\alpha}}, u_{\vec{\eta}}$ be as in~\eqref{formu2}. Assume 
$\alpha_i \le \eta_i$ for each $i$.  Letting $I_N$ denote $\{i, \alpha_i\neq \eta_i\}$, assume that 
for each $i \in I_N$   we have  $B(x_i ,\eta_i) \subset U$ (or 
$\frac{\partial u}{\partial \nu}=0$ on $\pa U\cap B(x_i, \eta_i) $ and $U$ is convex).
Then 
\begin{multline}
\label{premono}
\int_U|\nab u_{\vec{\eta}}|^2 - \cd \sum_{i=1}^N \g(\eta_i) -2\cd \sum_{i=1}^N\int_U \f_{\eta_i}(x-x_i)d\mu\\- \( \int_U |\nab u_{\vec{\alpha}}|^2 - \cd \sum_{i=1}^N \g(\alpha_i) -2\cd\sum_{i=1}^N\int_U \f_{\alpha_i}(x-x_i)d\mu\) \le 0,\end{multline}
with equality if the $B(x_i,\eta_i)$'s are disjoint from all the other $B(x_j, \eta_j)$'s for each $i \in I_N$. Moreover,  if $\eta_i \ge \rrc_i$ for each $i$, and 
$\eta_i=\rrc_i=\frac14 $ if $\dist (x_i, \pa \Omega \setminus \pa U) \le \hal $,  we have
\begin{multline}\label{supplem}
\sum_{\substack{i,x_i , x_j \in \Omega, i\neq j, \\
\dist (x_i , \pa \Omega \setminus \pa U) \ge 1, 
\dist(x_j, \pa \Omega \setminus \pa U ) \ge 1}} \( \g(x_i-x_j) - \g(\eta_i)\)_+ \\ \le \G^\Omega(X_N, U) -  \frac{1}{2\cd}\( \int_\Omega |\nab u_{\vec{\eta}}|^2 -\cd \sum_{i , x_i \in \Omega} \g(\eta_i)  -2\cd\sum_{i, x_i \in \Omega }\int_U \f_{\eta_i}(x-x_i)d\mu\) .\end{multline} 
\end{lem}
\begin{proof}
For any $\alpha \le \eta$, let us  denote $\f_{\alpha, \eta}$ for  $\f_{\alpha}-\f_{\eta}$ and  note that $\f_{\alpha,\eta}$ vanishes outside $B(0, \eta)$ and 
$$\g(\eta)-\g(\alpha) \le \f_{\alpha ,\eta}\le 0$$ 
while, in view of~\eqref{deltaf}
\begin{equation}\label{eqfae}
- \Delta \f_{\alpha, \eta}= \cd (\delta_0^{(\eta)}- \delta_0^{(\alpha)}).\end{equation}
Using the fact that from~\eqref{formu2} we have 
$$u_{\vec{\eta}}(x) - u_{\vec{\alpha}}(x) = \sum_{i\in I_N} \fae(x-x_i),$$
we may compute
\begin{align*}
T:= & \ 
\int_{U}| \nab u_{\vec{\eta}}|^2 - \int_{U}  |\nab u_{\vec{\alpha}}|^2
= 2 \int_U  (\nab u_{\vec{\eta}}- \nab u_{\vec{\alpha}} ) \cdot \nab u_{\vec{\alpha}}+ \int_U  |\nab u_{\vec{\eta}}- \nab u_{\vec{\alpha}}|^2 \\
= & \  2\sum_{i\in I_N } \int_U \nab \fae(x-x_i)  \cdot \nab u_{\vec{\alpha}}+ \sum_{i,j \in I_N} \int_U  \nab \fae(x-x_i) \cdot \nab \faej (x-x_j).
\end{align*}
If $B(x_i, \eta_i)\subset U$ the function $\fae(x-x_i) $ vanishes on $\partial U$, and we can integrate by parts without getting any boundary contribution. If not but we instead assume $\frac{\partial u}{\partial \nu}=0$ and $U$ convex, then in view of ~\eqref{formu2} and the definition of $\f_{\alpha,\eta}$,   the boundary term contributions are
$$\sum_{i\in I_N} \int_{\partial U} \fae(x-x_i) \sum_{j\in I_N}\( -\frac{\partial \f_{\alpha_j}}{\partial \nu} (x-x_j)- \frac{\partial \f_{\eta_j}}{\partial \nu}(x-x_j)\).$$ Since $\f_\alpha$ is always radial nonincreasing and since we consider $U$ which is convex, the  outer normal derivatives involved are always nonpositive, and since $\f_{\alpha, \eta}\le 0$, these boundary contributions are $\le 0$. 

With the help of ~\eqref{eqsu} and~\eqref{eqfae} we thus obtain in all cases \begin{align}
\label{ii1i2}
T 
&
\le
2 \cd\sum_{i\in I_N}\int_U  \fae(x-x_i) \Big( \sum_{j=1}^N \delta_{x_j}^{(\alpha_j)} -  d\mu\Big) 
 +\cd \sum_{i,j\in I_N } \int_U  \fae(x-x_i) \( \delta_{x_j}^{(\eta_j)}-\delta_{x_j}^{(\alpha_j)} \)
\\ & \notag
= 
\cd\sum_{i\in I_N} \int_U  \fae(x-x_i) \( \sum_{j=1}^N \delta_{x_j}^{(\alpha_j)} + \delta_{x_j}^{(\eta_j)}\)  -2 \cd\sum_{i\in I_N} \int  \fae(x-x_i)  d\mu
\\ & \notag
=
 \sum_{j=1}^N\sum_{i\in I_N, i\neq j} \cd \int_{\R^{\d}} \fae (x-x_i) d(  \delta_{x_j}^{(\alpha_j)}+ \delta_{x_j}^{(\eta_j)}) 
+ \cd\sum_{i\in I_N}  \int_{\R^{\d}} \fae(x-x_i) d( \delta_{x_i}^{(\alpha_i)} + \delta_{x_i}^{(\eta_i)} )
\\ & \notag \qquad 
- 2\cd \sum_{i\in I_N} \int_U\fae(x-x_i) d \mu.
\end{align}
Since $\f_{\alpha_i, \eta_i} \le 0$, the first term in the right-hand side is nonpositive, and is zero if the $B(x_i, \eta_i)$'s with $i\in I_N$ are disjoint from the other balls. For the diagonal terms, we note that 
  $$  \int_U  \fae(x-x_i)   \(  \delta_{x_i}^{(\alpha_i)} + \delta_{x_i}^{(\eta_i)}\)     =  - (\g(\alpha_i)-\g(\eta_i))$$ by definition of $\f_{\alpha,\eta}$ and the fact that $\delta_0^{(\alpha)}$ is a measure of mass $1$ on $\partial B(0,\alpha)$.
Since $\fae=\f_{\alpha_i}-\f_{\eta_i}$ this finishes the proof of~\eqref{premono}.
 
We may next apply this  in $\Omega$  to $u$ solution of~\eqref{defv} with $\eta_i=\rrc_i$ and $\alpha_i \le \rrc_i$  with  $\alpha_i = \rrc_i$ when $ \dist (x_i, \pa \Omega \setminus \pa U) \le 1$.
With this choice, by definition of $\rrc_i$ in \eqref{rrc} we are sure that $B(x_i, \eta_i)$ does not intersect any $B(x_j ,\eta_j)$ if $i \in I_N$ and $j \neq i$.
We are thus in the equality case and in view of the definition~\eqref{Glocal}  we find that 
\be \label{b3} \G^\Omega (X_N, U) = \frac{1}{2\cd}\(\int_{\Omega} |\nab u_\alpha|^2 -\cd \sum_{i, x_i \in \Omega} \g(\alpha_i) \)  - 
\sum_{i, x_i \in \Omega} \int_U \f_{\alpha_i} (x-x_i) d\mu(x) +\sum_{i, x_i\in \Omega} \GG(x_i) .\ee

We now define $\g_\eta =\min(\g, \g(\eta))$ and note that $\f_{\alpha, \eta}= \g_\eta-\g_\alpha$.
To prove~\eqref{supplem} we apply again  the previous result in $\Omega$ to the same $u$ with $\alpha_i$ as above, and this time $\eta_i \ge \rrc_i$ with equality if $\dist(x_i, \p \Omega \setminus \pa U) \le \hal$. We  return to the nonpositive first term in the right-hand side of~\eqref{ii1i2} and 
 bound it above and below  by 
\begin{align*}& \cd \sum_{i\neq j}\( \g_{\eta_i} (|x_i-x_j|+\alpha_j)- \g(|x_i-x_j|-\alpha_j)\)_- \le 
\sum_{i\neq j} \cd\int_{\R^{\d}} \f_{\alpha_i, \eta_i} (x-x_i) d\delta_{x_j}^{(\alpha_j)}\\ &
\leq 
\sum_{i\neq j}\cd \int_{\R^{\d}}\( \g_{\eta_i}(x-x_i)- \g_{\alpha_i}(x-x_i) \) d\delta_{x_j}^{(\alpha_j)} 
\leq \sum_{i\neq j}\cd \int_{\R^{\d}}
\( \g(\eta_i) - \g_{\alpha_i}(|x_i-x_j|+\alpha_j) \)_-,\end{align*}
where we used the fact that $\g_\alpha$ is radial decreasing.     
Combining the previous relations, we find 
\begin{align*}  
\lefteqn{
 \cd \sum_{x_i , x_j \in \Omega, i\neq j}\( \g_{\alpha_i}(|x_i-x_j|+\alpha_j) -\g(\eta_i)\)_+
 } \qquad & 
 \\ & 
\leq 
\(\int_{\Omega} |\nab u_{\vec{\alpha}}|^2  -\cd\sum_{i,x_i\in \Omega}  \g(\alpha_i)-2\cd\sum_{i, x_i \in \Omega}\int_U \f_{\alpha_i}(x-x_i)d\mu\) 
\\ & \qquad
-\( \int_{\Omega}|\nab u_{\vec{\eta}}|^2 - \cd\sum_{i, x_i \in \Omega} \g(\eta_i)   -2\cd\sum_{i, x_i \in \Omega} \int_U \f_{\eta_i}(x-x_i)d\mu\).
\end{align*}
Letting all $\alpha_i \to 0$  if $\dist (x_i, \pa \Omega\setminus \pa U) \ge 1$ we find $\G^\Omega(X_N, U)$ in the right-hand side in view of~\eqref{b3} (up to $\sum \GG(x_i)$) and $\cd \sum (\g(|x_i-x_j|)-\g(\eta_i))_+ $ in the left-hand side.  This finishes the proof.
\end{proof}

\subsection{Local energy controls}
We now show how the quantities based on $\G$ control the energy and the number of points locally. We will state all the results for $\G^\Omega$ and $\rrtt$, of course it implies them also for $\G$ and $\rrh$.

The following result shows that despite the cancellation between the two possibly very large terms $ \int_{\R^\d}|\nab u_{\vec{\eta}}|^2$ and $\cd \sum_{i=1}^N  \g(\eta_i)   $, when choosing $\eta_i=\rr_i$ we may control each of these two terms by the energy.
It is adapted from \cite[Lemma 2.7]{ls2}.

\begin{lem}
 \label{lem:contrdist1}There exist $C>0$ depending only on $\d$ and $\|\mu\|_{L^\infty}$  such that 
for any configuration $\XN$ in $U$  and $u$ corresponding via~\eqref{defv}, and for any $\Omega \subset U$,
\be\label{15}\sum_{i, x_i \in \Omega} \g(\rrtt_i)\le 2  \G^{\Omega}(X_N, U)  +   C \# (\{\XN\} \cap \Omega)   
\ee
and
\be\label{14}
 \int_{\Omega} |\nab u_{\rrtt_i}|^2 \le 4\cd   \G^{\Omega} (\XN, U) + C\#\(\{X_N \} \cap \Omega \)  \ee
with $\rrtt$ as in \eqref{rrtt} and computed with respect to $ \Omega$.
  \end{lem}
  \begin{remark} With the same proof, we can prove analogous results for  $\H_U$ and $\F$.\end{remark} 
  \begin{proof}Let us proceed as in Lemma \ref{monoto} with $\eta_i=\frac14\min (1, \dist(x_i, \pa U \cap \Omega))$
  and $\alpha_i=\rrc_i$.
We note that the assumptions of the lemma are verified  in $\Omega$  since the size of the balls intersecting $\partial \Omega$ is  not changed and $\alpha_i\le \eta_i$ for each $i$. We obtain  as in~\eqref{ii1i2} that 
\begin{multline*}
T:=\int_\Omega |\nab u_{\vec{\eta}}|^2 - \cd \sum_{i, x_i \in \Omega} \g(\eta_i) -2\cd \int_\Omega \f_{\eta_i}(x-x_i)d\mu\\- \( \int_\Omega |\nab u_{\vec{\alpha}}|^2 - \cd \sum_{i=1}^N \g(\alpha_i) -2\cd \int_\Omega \f_{\alpha_i}(x-x_i)d\mu\) 
\\ \le \cd\sum_{i, j\neq  i} \int_\Omega  \f_{\alpha_i, \eta_i}(x-x_i) \( \sum_j \delta_{x_j}^{(\alpha_j)}+\delta_{x_j}^{(\eta_j)}\).\end{multline*}
Assume  first that $x_i$ is such that  $\dist(x_i, \pa U\cap \Omega) \ge  1$ and  $\rrtt_i <1/20$.  Then $\dist(x_i,\pa \Omega \setminus \pa U) \ge 1$ and $\rrtt_i=\rrc_i=\rr_i= \frac14 \min_{j\neq i} |x_i-x_j|$, in view of the definitions of $\rrtt_i$ and $\rrc_i$. Using that $\f_{\alpha_i, \eta_i}\le 0$ we may bound 
\begin{equation*} \int  \fae(x-x_i)  \sum_{j\neq i } ( \delta_{x_j}^{(\alpha_j)} + \delta_{x_j}^{(\eta_j)}) 
\le  \int  \fae(x-x_i) \sum_{j, x_j\  \text{nearest neighbor to}\  x_i}  \delta_{x_j}^{(\alpha_j)}.
\end{equation*}
We then note that 
 $\f_{\alpha_i,\eta_i} (x-x_i)= \g_{\eta_i}(x-x_i) - \g_{\alpha_i}(x-x_i) \le \g(\eta_i) - \g_{\rr_i }(x-x_i) $ (with the notation as in the previous proof) using the definition of $\alpha_i$.
For $x_j$ nearest neighbor to $x_i$ we have $|x_i-x_j|= 4\rr_i< 1/5$,  hence also $\dist (x_j, \pa \Omega \setminus \pa U) \ge \hal$ by the triangle inequality, which implies by definition of $\rrc_j$ that $\rrc_j\le \frac14\min_{k\neq j} |x_k-x_j|\le  \rr_i <1/20$.
The support of $\delta_{x_j}^{(\alpha_j)}= \delta_{x_j}^{(\rrc_j)}$ is thus contained in $B(x_i, 5\rr_i)$, where $\g_{\rr_i}(x-x_i) \ge   \g(5\rr_i)$ by monotonicity of $\g$.
We  thus find that the right-hand side is bounded above in this case by
\begin{equation*}    \g(\eta_i)- \g( 5 \rr_i) = \g(\eta_i) - \g(5\rrtt_i).\end{equation*}  
On the other hand, if $\rrtt_i \ge 1/20$ then $5\rrtt_i \ge \eta_i$ and the same bound is true as well since the left-hand side is nonpositive.
If $\dist(x_i, \pa  U\cap \Omega) \le 1$ and $\rrtt_i= \rr_i \le 1/20 $ then the same reasoning as above applies. If on the contrary $\dist(x_i, \pa U\cap \Omega) \le 1$ and $\rrtt_i<\rr_i$, then $\rrtt_i = \frac14 \dist (x_i, \pa U\cap \Omega)=\eta_i$ and 
$\g(\eta_i)-\g(5\rrtt_i) \ge 0$, so the result holds as well.

Summing over $i$, we have thus obtained that 
$$T\le \cd   \sum_{i, x_i \in \Omega} ( \g(\eta_i)-\g(5 \rrtt_i)).$$
 
On the other hand, by definition of $T$ and choice of $\alpha_i$ and $\eta_i$,  we may also write 
\begin{multline*}
T\ge -\int_{\Omega} |\nab u_{\rrc_i}|^2 + \cd \sum_{i, x_i \in \Omega} \g(\rrc_i)
+2\cd  \sum_{i, x_i \in \Omega} \int_{\Omega} \f_{\rrc_i} (x-x_i) d\mu
\\-2\cd  \sum_{i, x_i \in \Omega} \int_{\Omega} \f_{\eta_i} (x-x_i) d\mu   -\cd  \sum_{i,x_i \in \Omega} \g(\eta_i) 
\\\ge
-2 \cd \( \G^\Omega(\XN, U)  - \sum_{i,x_i \in \Omega} \GG(x_i)  \)  - \cd\sum_{i, x_i \in \Omega} \g(\eta_i) .\end{multline*} 
Combining the two relations we deduce 
$$
\G^\Omega(\XN, U)  - \sum_{i, x_i \in \Omega} \GG(x_i) \ge - \sum_{i, x_i\in \Omega} \g(\eta_i)+ \hal \sum_{i, x_i \in \Omega} \g( 5\rrtt_i)
$$ By definition of $\GG$  \eqref{defGG} and choice of $\eta_i$ we have 
 $\sum_{i, x_i \in \Omega} \GG(x_i)- \g(\eta_i)\ge - C\#( \{ X_N\} \cap \Omega)$ and 
we deduce $$\sum_{i , x_i \in \Omega} \g(\rrtt_i)\le 2\G^\Omega(\XN, U) + C \#( \{ X_N\} \cap \Omega)$$ which proves \eqref{15}. 
In addition, applying \eqref{supplem} (with simply zero left-hand side) with $\eta_i = \rrtt_i$ we have 
$$\G^{\Omega} (X_N, U) \ge \frac{1}{2\cd}\( \int_{\Omega} |\nab u_{\rrtt_i}|^2 - \cd \sum_{i, x_i \in \Omega} \g(\rrtt_i)
- 2\cd \sum_{i, x_i \in \Omega} \int_{\Omega} \f_{\rrtt_i} (x-x_i) d\mu\)$$
hence~\eqref{14} follows after rearranging terms and using \eqref{intf}.
\end{proof}

We turn to showing how the energy  controls the  fluctuations.  The next lemma is adapted from previous results, such as \cite{rs}.
The first result~\eqref{disc1} allows to treat the case of an excess of points and control it using only the energy outside the set, while
 ~\eqref{disc10} allows to treat the case of a deficit of points and control it using only the energy inside the  set. The last two results provide improvements when considering balls  and using the energy in a larger set.
 
\begin{lem}[Control of charge discrepancies]\label{coronp}
 Let $\XN$ be a configuration in $U$, let $u$ be associated via~\eqref{defv}, and let $\Omega $ be a set of finite perimeter included in $U$.
 We have 
 \begin{equation}
\label{disc10}
\left| \min\( \int_{\Omega} \sum_{i=1}^N \delta_{x_i}-  \int_{\Omega} d\mu, 0\)\right| \le C\|\mu\|_{L^\infty}  |\partial \Omega|+ C|\partial \Omega|^{\hal} 
 \|\nab u_{\rrtt}\|_{L^2 (\{x\in \Omega, \dist (x, \partial \Omega) \le 1\})},\ee with $\rrtt$ computed with respect to any set containing $\Omega$, 
and if in addition $\Omega $ is at distance $\ge 1$ from $\partial U$, 
 \be \label{disc1}
 \max \(  \int_{\Omega} \sum_{i=1}^N \delta_{x_i}-  \int_{\Omega} d\mu, 0\) \le C\|\mu\|_{L^\infty}  |\partial \Omega|+ C|\partial \Omega|^{\hal} 
 \|\nab u_{\rrtt}\|_{L^2 (\{x\notin \Omega, \dist (x, \partial \Omega) \le 1\})},\ee 
 where  $C$ depends only on $\d$.

 Let $B_R\subset U$ be a ball of sidelength $R>2$ and let 
 \begin{equation*}
 D= \int_{B_R} \sum_{i=1}^N \delta_{x_i}-  \int_{B_R} d\mu.
 \end{equation*}
 If $D\le 0$ then 
 \be\label{disc30}
\frac{D^2}{R^{\d-2}}\left|\min \(1, \frac{D}{\|\mu\|_{L^\infty(B_R)}R^\d}\) \right|\le C \int_{B_{R}}|\nab u_{\rrtt}|^2, \ee
and if $D\ge 0$  and $B_{2R}\subset U$ 
\be\label{disc3}
\frac{D^2}{R^{\d-2}}\min \(1, \frac{D}{\|\mu\|_{L^\infty(B_{2R}) } R^\d}\) \le C \int_{B_{2R}}|\nab u_{\rrtt}|^2 \ee
where $C$ depends only on $\d$.

\end{lem}
\begin{proof}
Let $\chi$ be a  smooth nonnegative function equal to $1$ at distance $\le \hal$ from  $\Omega$ and vanishing at distance $\ge 1$ from $\Omega $ outside that set.  Let $\xi$ be a smooth nonnegative function equal to $1$ for points  in $\Omega $ at distance $\ge 1$ from $\partial \Omega$, and vanishing outside $\Omega$. Their gradient can be bounded by $C$ and $\|\nab \chi\|_{L^2}$ and $\|\nab\xi\|_{L^2}$ can be bounded by $ C |\partial \Omega|^{\hal}$.
Since $\rrtt_i \le \frac14$ for each $i$, we have
\be\label{clear1}\int \xi \sum_{i=1}^N \delta_{x_i}^{(\rrtt_i)}  \le \int_{\Omega} \sum_{i=1}^N \delta_{x_i}\le \int \chi \sum_{i=1}^N \delta_{x_i}^{(\rrtt_i)}.\ee
Using~\eqref{defv}, integrating by parts, using the fact that $\partial_{\nu} u_{\rrtt}=0$ on $\partial U$ and the Cauchy-Schwarz inequality,   we find
\begin{equation*}\left|\int_\Omega \chi\, d \( \sum_{i=1}^N \delta_{x_i}^{(\rrtt_i)}-  \mu\)\right|
\le \frac{1}{\cd}\|\nab \chi \|_{L^2}   \|\nab u_{\rrtt} \|_{L^2(\supp \nab \chi )}\le C   |\partial \Omega|^{\hal}  \|\nab u_{\rrtt} \|_{L^2(\supp \nab \chi )}\end{equation*}
and the same for $\xi$. 
Meanwhile
$$\left|\int_{\R^\d} (\indic_{\Omega}- \chi) d \mu\right|\le C|\partial \Omega|\|\mu\|_{L^\infty}$$ and the same for $\xi$.
Let us now first assume that  $\int_{\Omega} \sum_{i=1}^N \delta_{x_i}- \int_{\Omega} d\mu  \ge 0$.
Then in view of~\eqref{clear1}  and the above, we have 
\begin{multline*}0 \le \int_{\Omega} \sum_{i=1}^N \delta_{x_i}- \int_{\Omega} d\mu\le \int_\Omega \chi \( \sum_{i=1}^N \delta_{x_i}^{(\rrtt_i)}  - d\mu\) + O(|\partial \Omega| \|\mu\|_{L^\infty})\\
\le  C |\partial \Omega|^{\hal}  \|\nab u_{\rrtt} \|_{L^2(\supp \nab \chi )}+ C|\partial \Omega| \|\mu\|_{L^\infty}.\end{multline*}
In all cases, the result~\eqref{disc1} follows.  The proof of~\eqref{disc10} is similar.
 
 Let us now turn to~\eqref{disc30} and~\eqref{disc3}, following \cite[Lemma 4.6]{rs}.
We first consider the case that $D>0$ and  note that if 
\begin{equation}\label{defLe}
R+\eta\le t \le T:= \min \left(  2R, \Big((R+\eta)^\d + \frac{D}{C\|\mu\|_{L^\infty(B_{2R}) }}\Big)^{\frac{1}{\d}}\right)\end{equation}
with $C$ well-chosen,  we have 
\begin{eqnarray*}
 -\int_{\pa  B_t  } \frac{\partial u_{\rrtt} }{\pa \nu} &=&-\int_{ B_t}\Delta u_{\rrtt}  = \cd \int_{B_t}\Big(\sum_{i=1}^N\delta_{x_i}^{(\rrtt_i)} - d\mu\Big)\\
 &\ge &\cd\left(D  -   \int_{B_t\backslash B_R }d \mu\right)\ge\cd D- C\|\mu\|_{L^\infty} \left( t^{\d}-R^{\d}\right)\ge \frac{\cd}{2} D   ,
\end{eqnarray*}  if we choose the same $C$ in~\eqref{defLe} depending only on $\d$. 
By the Cauchy-Schwarz inequality, the previous estimate, and explicit integration, there holds 
\begin{eqnarray*}
 \int_{B_{2R}}|\nab u_{\rrtt}|^2&\ge&\int_{R+\eta}^{T}\frac{1}{|\pa B_t|}\left(\int_{\pa  B_t} \frac{\pa u_{\rrtt}}{\pa \nu}\right)^2dt\\
 &\ge &C D^2 \int_{R+\eta}^{T}t^{-(\d-1)}  \, dt=  CD^2 \left(\g(R+\eta)- \g(T)\right) ,  \end{eqnarray*}
 with $C$ depending only on $\d$.
 Inserting the definition of $T$ and rearranging terms, one easily checks that we obtain~\eqref{disc3}.
There remains to treat the case where $D\le 0$.  This time, we let 
$$T \le  t \le R-\eta, \qquad    T:= \Big( (R-\eta)^\d - \frac{D}{C\|\mu\|_{L^\infty(B_R)}}\Big)^{\frac1\d}$$
and if $C$ is  well-chosen  we have \begin{eqnarray*}
 -\int_{\pa  B_t}\frac{\pa u_{\rrtt} } {\pa \nu} &=&-\int_{ B_t}\Delta u_{\rrtt} = \cd \int_{B_t}\Big(\sum_{i=1}^N \delta_{x_i}^{(\rrtt_i)} - d \mu\Big)\\
 &\le &\cd\left(D  +   \int_{B_R\backslash B_r }d \mu\right)\le \frac{\cd}{2} D   ,
\end{eqnarray*}and the rest of the proof is analogous, integrating from $T$ to $R-\eta$. 
\end{proof}

The next lemma is similar to \cite[Prop. 2.5]{ls2}.

\begin{lem}\label{prop:fluctenergy}
Let  $\varphi$ be a  Lipschitz function in $U$ with bounded support.  Let $\Omega$ be an open set   with finite perimeter containing a $1$-neighborhood of the support of $\varphi$ in $U$.  For any configuration $\XN$ in $U$, letting $u$ be defined as in~\eqref{defv} (resp.~$v$ as in~\eqref{defu}), we have
\begin{equation} 
\label{fluctuationsbu}
\left|\int_{\R^\d} \varphi \, \left(\sum_{i=1}^N \delta_{x_i} -  d\mu \right) \right|
\leq
C \|\nab\varphi\|_{L^\infty(\Omega)} \(  ( |\p \Omega|^{\hal} + |\Omega|^{\hal}) \|\nab u_{\rrtt} \|_{L^2(\Omega )}+   |\Omega|  \|\mu\|_{L^\infty(\Omega)}\right),
\end{equation}  (and resp.~the same with $v_{\rrtt}$ in place of~$u_{\rrtt}$),
where $C$ depends only on $\d$ and $\rrtt$ is computed with respect to any set containing $\Omega$.
\end{lem}

\begin{proof}

We may find $\chi$  a smooth cutoff function equal to $1$ in a $1/2$-neighborhood of the support of $\varphi$  and equal to $0$ outside $\Omega $, such that $\int |\nab \chi|^2 \le C |\partial \Omega|  $. Integrating~\eqref{defv} against $\chi$ we thus get 
$$\left|\int \chi  \Big( \sum_{i=1}^N \delta_{x_i}^{(\rrtt_i)} - d\mu\Big) \right|\le  \frac{1}{\cd}\|\nab \chi\|_{L^2}\|\nab u_{\rrtt}\|_{L^2(\Omega)}\le C |\partial \Omega|^{\hal} \|\nab u_{\rrtt}\|_{L^2(\Omega)},$$ where $C$ depends only on $\d$.
It follows that letting $\#I$ denote the number of balls $B(x_i, \frac14)$ intersecting  $\Omega$,  we also have
\begin{equation}
\label{contrnbpoints}
\# I \le  \int_{\Omega} d \mu + C|\partial  \Omega|^{\hal}      \|\nab u_{\rrtt} \|_{L^2(\Omega)}.
\end{equation}
Secondly, integrating~\eqref{defv} against $\varphi$, we  have 
\begin{equation}\label{rel3}
\left|\int_U  \Big( \sum_{i=1}^N \delta_{x_i}^{(\rrtt_i)} - d \mu\Big) \varphi\right|= \frac{1}{\cd}\left| \int_{U}\nabla u_{\rrtt} \cdot \nab \varphi\right|\le C |\Omega|^{\hal} \|\nab \varphi\|_{L^\infty} \|\nab u_{\rrtt}\|_{L^2(\Omega)}.
\end{equation}
On the other hand, since by definition $\rrtt_i\le \frac 14$ for each $i$, we have
\begin{equation}
\left|\int_U \Big( \sum_{i=1}^N( \delta_{x_i} - \delta_{x_i}^{(\rrtt_i)}) \Big) \varphi\right|\le \# I   \|\nab \varphi\|_{L^\infty} .\end{equation}
Combining this with~\eqref{rel3} and  ~\eqref{contrnbpoints}, we get the result. 
\end{proof}

\section{Proof of the screening result}\label{appa}
The goal of this appendix is to prove the screening result of  Propositions \ref{outscreen}.
This follows from  adapting  and  optimizing  the procedure from \cite{ss1,rs,ps}, in particular \cite{ps} simplified to the Coulomb case.

Let us  first informally describe thing for the outer screening. 
We will work with ``electric fields" $E$ which are meant to be gradients of the potentials $u$  of~\eqref{defv} or $w$ of~\eqref{defw}, or more generally to satisfy relations of the form 
\be\label{dive}
-\div E= \cd\(\sum_{i=1}^n \delta_{x_i}-\mu\).\ee
A truncated version of $E$ can be defined just as in~\eqref{formu}: for any $E$ satisfying a relation of the form~\eqref{dive} we let 
\be \label{eer}E_{\rrc}= E- \sum_{i=1}^n \nab \f_{\rrc_i}(x-x_i)\ee
where $\rrc_i$ is as in~\eqref{rrc}.

Assume we are given a configuration $X$ (with unspecified number of points) in a hyperrectangle, together with its electric field $E$, and assume roughly that we control well its energy near the boundary of a hyperrectangle $Q_T$ of sidelengths close to $T$. The goal of the screening is to modify the configuration $X$ and the electric field $E$ only outside of $Q_{T-1}$ and to extend them to a ``screened" configuration $X^0$ and a ``screened" electric field $E^0$ in $Q_{T+\l}\in \mathcal{Q}_{T+\l}$  in such a way that 
$$\left\{ \begin{array}{ll}
- \div E^0= \cd(\sum_{p\in X^0} \delta_p - \mu) & \text{in} \ Q_{T+\l}\cap U\\
E^0 \cdot \nu =0 & \text{on} \ \partial (Q_{T+\l}\cap U)\end{array}\right.$$
This implies in particular that the screened system is neutral, i.e the number of points of $X^0$ must be equal to $\mu(Q_{T+\l}\cap U)$. We note that in the Neumann case where $\Omega$ can intersect $\partial U$, the desired boundary condition is already satisfied for the original field on $\partial U$, so there is no need to modify it near $\partial U$.

 The screened electric field $E^0$ may not be a gradient, however thanks to Lemma \ref{projlem} its energy provides an upper bound for computing $\G(X^0, Q_{T+\l}\cap U)$. The goal of the construction is to show that we can build $E^0$ and $X^0$ without adding too much energy to that of the original configuration, which will allow to bound $\G(X^0, Q_{T+\l} \cap U ) $ in terms of  $\H_U(X, \Omega)$.
In order to accomplish this, we will split the region to be filled into cells where we  solve appropriate elliptic problems and estimate the energies by elliptic regularity estimates. In order to ``absorb" and screen the effect of the possibly rough data on $\partial Q_T$, we need a certain distance $\l$, which has to be large enough in terms of the energy of $E$, this leads to the ``screenability condition" bound on $\l$, as previously mentioned.

\subsection{Finding a good boundary}

We focus on the outer screening proof,  the proof of the inner case is analogous (for details  of what to do near the corners, one may refer to \cite{rns}).

 Assume then that $\Omega=Q_R\cap U$. Since $U$ is assumed to be a disjoint union of parallel hyperrectangles, $\Omega$ is itself a hyperrectangle. 
 
We are given a configuration $X_n$ in $Q_R\cap U $ with $\tilde \l \ge \l \ge C$,  and $u$ is as in~\eqref{eqsp}

We set $E=\nab u$ with the notation $E_{\rrc}$ defined in~\eqref{eer}. We also let 
\be \label{defM} M:= \int_{(Q_{T+4}\backslash Q_{T-4})\cap U }|\nab u_{\rrtt}|^2 .\ee
By a pigeonhole principle, 
there exists a $T \in  [ R-2 \tilde \ell, R-\tilde \ell]$ such that 
 \begin{equation}\label{defMt}
M:=  \int_{ (Q_{T+4}\backslash Q_{T-4})  \cap U}  |\nab u_{\rrtt}|^2 \le \frac{S(X_n)}{\tilde \ell} \ee
\be
   M_\l:= \max_{x} \int_{ (Q_{T+4}\backslash Q_{T-4}) \cap\car_\l(x) \cap U}  |\nab u_{\rrtt}|^2 \le S'(X_n),\end{equation} resp. with $Q_{T+4}\backslash Q_{T-4}$.
   
    We recall that  on $\partial U$ we have a zero Neumann boundary condition for $u$ so the desired final condition is already satisfied there. 
 
 By a mean value argument we can find $\Gamma$   a piecewise affine  boundary (with slopes in a given set, alternating only at distances bounded above and below) of a set containing $Q_T\cap U $ and contained in $Q_{T+1}\cap U$ 
   such that 
\begin{equation}\label{bonbord}
\int_{\Gamma \cap U } |E_{\rrc}|^2 \le C M,
\quad \sup_x \int_{\Gamma\cap Q(x,\l)\cap U} |E_{\rrc}|^2 \le C M_\l.
   \end{equation}
   
We note that as soon as $\tilde \ell$ is large enough,   we only consider regions at distance $ \ge 1$ from $\pa \Omega$, so there is no difference between $\rrc$ and $\rrtt$ there. 

   
   We take it to be the boundary relative to $U$ 
   of a set containing $Q_T \cap U$ and contained in $Q_{T+1}\cap U$, and we then complete it by   a subset  $\Gamma'$ of $\partial U$ in such a way that  $\Gamma\cup \Gamma'$ then encloses a closed  domain of $\bar U \cap \bar Q_T$. We also recall that by assumption $U$ is a union of hyperrectangles and that $\partial Q_R$ is parallel to the sides of $U$.      
   In all cases we denote by 
$\mathcal O$ (like ``old")  the part of $Q_{T+1}\cap U $ delimited by $\Gamma\cup \Gamma'$ and by $\New$ (like ``new") the set $\Omega \backslash \Old$.  We keep $X_n$ and $E$ unchanged in $\Old$ and discard the points of $X_n$ in $\New$ to replace them by new ones.
We note that the good boundary $\Gamma$ may intersect some $B(x_i,\rrc_i)$ balls centered at points of $X_n$. These balls will need to be ``completed", i.e. the contributions of $\delta_{x_i}^{(\rrc_i)} \indic_{Q_T\backslash\Old}$ to be retained.

\subsection{Preliminary lemmas}
We start with a series of preliminary results which will be the building blocks  for the construction of $\Escr$. 

\begin{lem}[Correcting fluxes on rectangles]\label{lem57}
Let~$H$ be a hyperrectangle of~$\R^\d$ with sidelengths in~$[\ell, C\ell]$ with~$C$ depending only on~$\d$.
 Let~$g\in L^2(\pa  H) $.
Then there exists a constant~$C$ depending only $\d$ such that the mean zero  solution of
\begin{equation}\label{eqnu}
\left\{\begin{array}{ll}
-\Delta h  = \int_{\partial H} g    & \text{in} \  H\\
\partial_\nu h =g & \text{on} \ \pa H \end{array}\right.
\end{equation}
satisfies the estimate
\begin{equation}\label{estlcs2}
\int_{ H}  |\nab h|^2 \le  C \ell  \int_{\pa   H} |g|^2.
\end{equation}
\end{lem}
\begin{proof}
This is   \cite[Lemma 5.8]{rs}. 
\end{proof}
 
The next lemma serves to complete the smeared charges  which were ``cut" into two pieces by the choice of the good boundary. The proof can be deduced from an inspection of that of  \cite[Lemma 6.6]{ps}.

\begin{lem}[Completing charges near the boundary]\label{chargesnearbdry}Let $\mathcal{R}$ be a hyperrectangle in $\R^\d$ of center $0$ and sidelengths in $[a,Ca]$ with $C$ depending only on $\d$. Let $F$ be a face of $\mathcal R$. Let $X_n$ be a configuration of  points  contained in an $1/4$-neighborhood of $F$. Let $c$ be a constant such that
 \begin{equation}\label{CR}
c |F|=\cd\int_{{\mathcal R}}\sum_{i \in\C_{\mathcal R}} \delta_{x_i}^{(\rrc_i)}\,.
 \end{equation}
The mean-zero solution to \[
 \left\{\begin{array}{ll}
  -\Delta h=\cd \sum_{p\in \C}\delta_p^{(\rrc_p)}& \text{ in }{\mathcal R}\, ,\\[2mm]
  \pa_\nu h=0& \text{ on }\pa {\mathcal R}\setminus  F\, ,\\[2mm]
\pa_\nu h=c& \text{ on } F
\end{array}
\right.
\] 
satisfies 
\begin{equation}\label{estboundary}
 \int_{{\mathcal R}}|\nab h|^2\le C \( n^2 a^{2-\d} + \sum_{i\neq j } \g(x_i-x_j) +\sum_{i=1}^n \g(\rrc_i)\)
\end{equation} where $C$ depends only on $\d, a, b$.
\end{lem}
\begin{proof}
We split $h=u+v$ where 
$$ \left\{\begin{array}{ll}
  -\Delta u=\cd \sum_{i}\delta_{x_i}^{(\rrc_i)} - c& \text{ in }{\mathcal R}\ \\[2mm]
  \pa_\nu u=0& \text{ on }\pa {\mathcal R},\end{array}
\right.$$
and
 $$ \left\{\begin{array}{ll}
  -\Delta v=c \frac{|F|}{|\mathcal R|}& \text{ in }{\mathcal R}\ \\[2mm]
  \pa_\nu v=0& \text{ on }\pa {\mathcal R}\setminus  F\ \\[2mm]
\pa_\nu v=c& \text{ on } F.
\end{array}
\right.$$
The $v$ part is explicitly computable and has energy bounded by $Cc^2 a^\d\le C \# \C^2 a^{2-\d} $.
For the $u$ part, we observe that 
$$u=\cd\sum_{i}\int G_{\mathcal R} (x,y) \delta_{x_i}^{(\rrc_i)}(y)$$ where $G_{\mathcal R}(x,y)$ is the Neumann Green function of the 
hyperrectangle with background $1$, as in Proposition \ref{proappgreen}.  Using the estimate~\eqref{borneGU}, we have 

$$G_{\mathcal R} (x,y) \le C \g(x-y)$$ hence we deduce the result.
\end{proof}

\subsection{Main proof}


\smallskip 
\noindent

We let $I_\pa $ be  the indices corresponding to the points of $X_n$ whose smeared charges touch $\Gamma$, i.e.
\begin{equation}\label{defl0}
 I_{\pa}=\left\{ i \in[1,n] :\ B(x_i,\rrc_i)\cap \Gamma \neq\varnothing\right\} 
\end{equation}
and define $$\N= \# I_{\pa} + \#\( \{i, x_i \in \Old\}\backslash I_{\pa}\).$$
The goal of the construction is to place an additional  
$\mn- \N$ points in $(Q_R\cap U)\backslash \Old$, where $\mn=\mu(Q_R\cap U)$.

Next, we partition $(Q_R\cap U)\backslash \Old$ into hyperrectangles $H_k$ (or intersections of hyperrectangles with $(Q_R\cap U) \backslash \Old$) with sidelengths $\in [\ell/C,C\ell]$ for some positive constant $C>0$ (we note that we may always make sure in the construction of $\Gamma$ that the shapes formed by $H_k \backslash \Old$ are nondegenerate)
 in such a way that 
 letting $m_k$ be the constant such that 
 \begin{equation}\label{defmi}
 \cd m_k |H_k| =  \int_{\Gamma \cap \pa  H_k}   E_{\rrc} \cdot{ \nu}-n_k
 ,\end{equation}  with  $\nu$ denoting the outer unit normal to $\Old$ and $$n_k:= \cd\int_{H_k} \sum_{i \in I_\partial  } \delta_{x_i}^{(\rrc_i)},$$
 we have $\int_{H_K} (\mu+m_k ) \in \mathbb{N}$.
 This is possible if $|m_k|<\hal m$ (recall $\mu \ge m$) and  can be done by constructing successive strips as in Lemma \ref{rect}, as soon as $\l>C $ for some $C>0$ depending only on $\d$ and $m$.

  We will give below a condition for $|m_k|<\hal m  $.  Now define $$\tilde \mu= \mu+ \sum_k \indic_{H_k} m_k.$$
  Since $$\N=-\frac{1}{\cd}\int_{\Gamma} E_{\rrc} \cdot \nu+ \frac1\cd \sum_k n_k +\int_{\Old } d\mu$$ and $\mn=\mu(\Omega)$, 
  in view of 
 ~\eqref{defmi}  we may check that 
  \be \label{intmut} \int_{\New}\tilde \mu= \mn-\N.\ee
      \noindent
 {\bf Step 1: Defining $ \Escr$}.\\
 We define $\Escr_{\rrc}$ as a sum $ E_1 + E_2 +E_3$, some of these terms being zero except for $H_k$ that has some boundary in common with $\Gamma$, then denoted $F_k$. 
 
The first vector field contains the contribution of the completion of the smeared charges belonging to $I_{\pa}$.
We let
 $$E_{1}:=\sum_k \indic_{H_k} \nab h_{1,k}$$ where $h_{1,k}$ is the solution of
\begin{equation}\label{defh1}
\left\{\begin{array}{ll}
 -\Delta h_{1,k}=\cd \sum_{i\in I_\pa}\delta_{x_i}^{(\rrc_i)}& \text{ in } H_k
,\\[3mm]
 \pa_\nu h_{1,k}=0& \text{ on }\pa H_k \setminus \Gamma \ ,\\[3mm]
 \pa_\nu h_{1,k}=\frac{- n_k}{| F_k|  }    &\text{ on } F_k,
\end{array}
\right.
\end{equation}
We note that the definition of $ n_k$ makes this equation solvable.

The second vector field is defined to be $E_{2}= \sum_{k} \indic_{H_k} \nab h_{2,k}$ with  
\begin{equation*}
\left\{\begin{array}{ll}
   -\Delta  h_{2,k}= \cd m_k &  \text{ in } H_k\ ,\\[3mm]
   \pa_\nu h_{2,k}= g_k
  &\text{ on }\pa  H_k ,
\end{array}
\right.
\end{equation*} where we let $g_k=0$ if $ H_k$ has no face in common with $\Gamma$ and otherwise 
\begin{equation}\label{gi} g_k= - E_{\rrc}\cdot {\nu}    + \frac{n_k}{|F_k|} 
 \end{equation}
with $E_{\rrc} \cdot \vec{\nu}$ taken with respect to the outer normal to $\Old$.  We note that this is solvable in view of~\eqref{defmi} and the definitions of  $n_i$.

The third vector field consists in the potential generated by a sampled configuration $Z_{\mn - \N}$ in $Q_R\cap U \backslash \Old$:  we let $E_3= \nab h_3$ where $h_3$ solves 
 \begin{equation}\label{defh4}
 \left\{\begin{array}{ll}
   -\Delta  h_3=\cd \left(\sum_{j=1}^{\mn-\N} \delta_{z_j}- \tilde \mu\right)& \text{ in } \New 
   \\[3mm]
   \pa_\nu h_3=0& \text{ on }\pa \New .
\end{array}
\right. 
\end{equation}
We note that this equation is solvable since~\eqref{intmut} holds.
We then define in $\New$, 
$
 \Escr_{\rrc}=( E_1+E_2+E_3)\indic_{\New}+ E_{\rrc}\indic_{\Old}$
and  $Y_\mn= \{X_n, B(x_i, \rrc_i) \cap \Old \neq \varnothing\} \cup \{Z_{\mn-\N}\}$ 
Finally, we let  
$$\Escr= \Escr_{\rrc}+\sum_{i=1}^{\mn} \nab \f_{\bar\rr_i}(x-y_i)$$  where the $\bar\rr_i$ are the minimal distances as in~\eqref{rrc} of $Y_\mn$. Note that  for the points near $\Gamma$, these may not correspond to the previous minimal distances for the configuration $X_n$ or $Z_{\mn-\N}$, which is why we use a different notation.

We note that  the normal components are always constructed to be continuous across interfaces,  so that no divergence is created there, and so    $\Escr$ thus defined satisfies 
\be\label{divescr} \left\{\begin{array}{ll} - \div \Escr= \cd(\sum_{i\in Y_\mn} \delta_{y_i}-\mu) \quad & \text{in} \ \Omega\\
\Escr\cdot \nu= 0 & \text{on} \ \partial \Omega.\end{array}\right.
\ee

\smallskip 

\noindent
{\bf Step 2: Controlling $m_k$.}
First we control the $ n_k$. 
The results of Lemma \ref{coronp}   allow to show that 
\begin{equation}\label{nalpha} n_k^2  \le C \int_{H_k}|E_{\rrc}|^2
  \le  C M_\l,\qquad   \sum_k n_k^2  \le C M. \end{equation}
We note that it follows in the same way that $\# I_\pa \le CM \le C \frac{S(X_n)}{\tilde \ell}$ with~\eqref{defMt}.

To control $m_k$ we write that in view of~\eqref{defmi}  and~\eqref{bonbord}
\begin{equation}
|m_k| \le C \ell^{-\d}    \int_{\Gamma \cap \pa  H_k}  |E_{\rrc}| +|n_k|\l^{-\d} . \end{equation}
Using the Cauchy-Schwarz inequality  we bound 
$$   \int_{\Gamma\cap \pa H_k} |E_{\rrc}|\le \ell^{\frac{\d-1}{2}}M_\l^{\hal}.$$
We conclude that 
\be \label{bmi}|m_k| \le C \ell^{-\frac{\d}{2}-\frac{1}{2}}M_\l^{\hal}+C \ell^{-\d} M_\ell^\hal\le C \ell^{-\frac{\d}{2}-\frac{1}{2}}M_\l^{\hal} .\ee
The condition $|m_k|<\hal m$ thus is implied by 
$$
C M_\l^{\hal} \l^{\frac{-  \d-1}{2}} <  \hal m .$$
This is the screenability condition~\eqref{screenab}.
As an alternate, we can also bound 
$$\left|\int_{\New} \mu - \tilde \mu\right| \le C \sum_k | m_k | \ell^\d \le  C\ell^{\frac{\d}{2}-\hal} M^{\hal} + C  M \le C \ell^{\d-1}+ C \frac{S(X_n) }{ \tilde \ell} $$ in view of~\eqref{bonbord} and~\eqref{defMt}, thus completing the proof of~\eqref{bornimp}.
In the same way, using Cauchy-Schwarz, we may also write  that 
$$ m_k^2 \le C \ell^{-2\d} \int_{\Gamma\cap \pa H_k} |E_{\rrc}|^2 \ell^{\d-1}+   Cn_k^2 \ell^{-2\d}$$
and thus 
$$\int_{\New} (\mu - \tilde \mu)^2 \le C \sum_k m_k^2 \ell^\d \le  C \ell^{-1} \int_{\Gamma} |E_{\rrc}|^2 + M \ell^{-\d} \le  C \frac{S(X_n)}{\tilde \ell \ell}$$
in view of~\eqref{defMt} and~\eqref{bonbord}, thus proving~\eqref{mmut2}.

 \smallskip

\noindent
{\bf Step 3: Estimating the energy of $\Escr$}.
To estimate the energy of $\Escr$ we need to evaluate $\int_{\Omega} |\Escr_{\bar \rr}|^2$.
First, for $E_1$ we use  Lemma \ref{chargesnearbdry} and combine it with~\eqref{supplem} applied with $\eta_i=\frac14$ to bound $\sum_{p\neq q} \g(p-q)$ by the energy in a slightly larger set, thus we are led to  
\begin{equation*}
\int_{\New}  |(E_1)_{\rrc}|^2 \le C\(\sum_k  n_k^2  + CM\) \le C M,
\end{equation*}
where we have used~\eqref{defM}, ~\eqref{nalpha}, and the geometric properties of 
$H_k$.

\smallskip

For $E_{2}$ we use 
 Lemma \ref{lem57} to get 
\begin{equation*}
 \int_{H_k} |E_{2} |^2 \le 
 C\l
\(\int_{\partial H_k\cap \Gamma} |E_{\rrc}|^2+Cn_k^2 \). 
\end{equation*}
Summing over $k$ and using~\eqref{bonbord}, we obtain 
\begin{equation*}
 \sum_k \int_{H_k} |E_{2} |^2 \le 
  C\l M    .  \end{equation*}
  For $E_3$ we  use that,  by definition of $\G$, 
  \be \label{bornh3} \int_{\Omega \backslash \Old } |\nab h_{3,\rrh}|^2 \le 2\cd \( \G(Z_{\mn-\N}, \tilde \mu, \New) -\sum_{j=1}^{\mn-\N} \GG(z_j) \) + \cd \sum_{j=1}^{\mn-\N} \g(\rrh_j)    + C (\mn-\N) \ee since   $\int_{\R^\d}|\f_{\eta}|\le C$ for each $\eta$ (see~\eqref{intf}).  
   Since $E= E_{\rrc}=\nab u_{\rrc}$ in $\Old$,
 we deduce that 
 \begin{align*}  
\int_{\Omega}|\Escr_{\rrc}|^2  
& \leq \int_{\mathcal{O} } |\nab u_{\rrc}|^2 + C\l  M   
 + C \( 2\cd \( \G(Z_{\mn-\N}, \tilde \mu, \New) -\sum_{j=1}^{\mn-\N} \GG(z_j) \) + \cd \sum_{j=1}^{\mn-\N} \g(\rrc_j)\)
 \\ & \qquad 
 + C (\mn-\N) .
 \end{align*}
To estimate $\G(Y_{\mn}, \mu, \Omega)$ we use Lemma \ref{projlem}, the definition of $\G$ and ~\eqref{supplem}, which tells us that to go from $\rrc$ (with possibly intersecting balls)  to $\bar \rr$, we just  need to add  the ``new interactions" $\sum_{(i,j) \in J} \g(x_i-z_j) $.
This yields 
\begin{align*}
\lefteqn{
\G(Y_\mn,\Omega)
} \ & 
\\ &
\leq
 \frac{1}{2\cd} \int_{\mathcal{O} } |\nab u_{\rrc}|^2    - \hal \sum_{i=1}^\mn \g(\rrh_i) -\sum_{i=1}^\mn \int_\Omega \f_{ \rrh_i}(y-y_i) d\mu(y)
+C
 \sum_{(i,j) \in J} \g(x_i-z_j)  \\ & \quad
+\sum_{j=1}^{\mn-\N} \GG(z_j)
   + C\l  M   
 + C \bigg( \G(Z_{\mn-\N}, \tilde \mu, \New) -\sum_{j=1}^{\mn-\N} \GG(z_j) \bigg) 
 +C\sum_{j=1}^{\mn-\N} \g(\rrh_j) + C(\mn-\N).
\end{align*}
Since on the other hand 
  $$\H_U(X_n, \Omega)=  \frac{1}{2\cd}\( \int_{\Omega}|\nab u_{\rrc}|^2- \cd\sum_{i=1}^{n} \g(\rrc_i)\) - \sum_{i=1}^n \int_{\Omega} \f_{\rrc_i}(x-x_i) d\mu(x)$$  
it follows that 
\begin{align}
\label{concs}
\lefteqn{
\G(Y_\mn,\mu,\Omega)-\H_U(X_n, \Omega)
} \quad & 
\\ & \notag
\leq - \frac{1}{2\cd}\int_{\Omega\backslash \Old} |\nab u_{\rrc}|^2  +\hal \sum_{\{ i\in \{1,\ldots,n\}\,:\, x_i \notin \Old\}} \g(\rrc_i) +C  \sum_{j=1}^{\mn-\N}\g(\rrh_j) +C \l M
\\ & \notag \quad  
+  C\G(Z_{\mn-\N}, \tilde \mu, \New)+
 C\sum_{(i,j) \in J} \g(x_i-z_j) + C (n-\N)+ C(\mn-\N).
 \end{align}
On the other hand, since $\mathcal O $ contains $Q_{T-4}\cap \Omega $, we have
\begin{align}\label{alig3}
\lefteqn{
\frac{1}{2\cd}\Bigg( \cd \sum_{\{ i\in \{1,\ldots,n\}\,:\, x_i \notin \Old\}}  \g(\rrc_i) -\int_{\Omega\backslash \mathcal O} |\nab u_{\rrc}|^2\Bigg)  
}   \qquad  &
\\  \notag & 
\leq \frac{1}{2\cd} \int_{(Q_{T+4}\backslash Q_{T-4})\cap U } |\nab u_{\rrtt}|^2 + \frac{1}{2\cd}\Bigg( \cd \sum_{\{ i\in \{1,\ldots,n\}\,:\, x_i \notin \Old\}}  \g(\rrc_i) -\int_{\Omega\backslash Q_{T-4}} |\nab u_{\rrc}|^2\Bigg)
\\  \notag & 
\leq   \frac{M}{2\cd}+ C (n-\N), 
\end{align} 
where we bounded the second term in the right-hand side by  using Lemma \ref{monoto} to change $\rrc$ into $\frac14$ and then bounded $\sum \g(\frac14)$ for $x_i \notin \Old$  by the number of points not in $\mathcal O$. 
We may also write using  \eqref{15} and using that $\rrh=\rrtt$ in this case, 
\begin{equation}\label{alig4}
\sum_{j=1}^{\mn-\N} \g(\rrh_j) \le  C 
\( \G(Z_{\mn-\N}, \tilde \mu, \New) + ( \mn-\N)\).
\end{equation}
Inserting \eqref{alig3} and \eqref{alig4} into \eqref{concs} and using \eqref{defMt}, we find 
\begin{multline*}
\G(Y_\mn,\mu,\Omega)-\H_U(X_n, \Omega)\\ \le  C\l \frac{S(X_n)}{\tilde \ell} + C \G(Z_{\mn-\N}, \tilde \mu, \New) +
 C\sum_{(i,j) \in J} \g(x_i-z_j)+C( |n-\mn|+ |\mn-\N|).\end{multline*}
Using~\eqref{bornimp} and $\mu (\New) \le C \tilde \ell R^{\d-1}$  allows to  bound the last term on the right side,  and then we get~\eqref{nrjy}.

\end{document}